\documentclass[11pt,english]{article}

\usepackage[latin9]{inputenc}
\usepackage[active]{srcltx}
\usepackage{verbatim}
\usepackage{calc}
\usepackage{amsthm}
\usepackage{amsmath}
\usepackage{amssymb}
\usepackage{esint}
\usepackage{float}
\usepackage{subfigure}

\providecommand{\keywords}[1]{\noindent \textbf{\textit{Keywords---}} #1}

\providecommand{\MSC}[1]{\noindent \textbf{\textit{AMS Subject Classifications---}} #1 \newline}

\makeatletter
\theoremstyle{plain}
\newtheorem{thm}{\protect\theoremname}
  \theoremstyle{plain}
  \newtheorem{criterion}[thm]{\protect\criterionname}
  \theoremstyle{remark}
  \newtheorem{rem}[thm]{\protect\remarkname}
  \theoremstyle{plain}
  \newtheorem{prop}[thm]{\protect\propositionname}
  \theoremstyle{definition}
  \newtheorem{defn}[thm]{\protect\definitionname}
  \theoremstyle{plain}
  \newtheorem{lem}[thm]{\protect\lemmaname}

\usepackage{constants}

\def\imag{\operatorname{i}}

\def\I{\mathbb{I}}

\def\F{\mathbb{F}}
\def\N{\mathbb{N}}
\def\R{\mathbb{R}}
\def\i{\mbox{i}}
\def\Kt{\widetilde{K}}

\usepackage{mathtools}

\makeatother

\usepackage{babel}
  \providecommand{\criterionname}{Criterion}
  \providecommand{\definitionname}{Definition}
  \providecommand{\lemmaname}{Lemma}
  \providecommand{\propositionname}{Proposition}
  \providecommand{\remarkname}{Remark}
\providecommand{\theoremname}{Theorem}

\begin{document}

\title{Greed is Super: A Fast Algorithm for Super-Resolution}
\author{Armin Eftekhari and Michael B.\ Wakin\footnote{Both authors are with the Department of Electrical Engineering and Computer Science at the Colorado School of Mines. The authors' email addresses are armin.eftekhari@gmail.com and mwakin@mines.edu. AE is the corresponding author. This work was partially supported by NSF CAREER Grant CCF-1149225 and NSF Grant CCF-1409261.
}
}
\maketitle

\begin{abstract}
We present a fast two-phase algorithm for super-resolution with strong theoretical guarantees. Given the low-frequency part of the spectrum of a sequence of impulses, Phase I consists of a greedy algorithm that roughly estimates the impulse positions. These estimates are then refined by local optimization in Phase II.

In contrast to the convex relaxation proposed by Cand\`{e}s et al., our approach has a low computational complexity but requires the impulses to be separated by an additional logarithmic factor to succeed. The backbone of our work is the fundamental work of Slepian et al.\ involving discrete prolate spheroidal wave functions and their unique properties. 
\end{abstract}


\MSC{94A12, 94A15, 42A99}


\section{Introduction}\label{sec:intro}
Many sensing mechanisms have finite resolution or bandwidth.
Provided with the low-frequency content of the signal,
\emph{super-resolution} is then the problem of (partially or completely)
recovering the high-frequency content of the signal. More concretely, here we restrict ourselves to the  problem set up next.

Consider the time interval $\I = [0, 1)$.  	
For integer $K$, $\tau\in\R^K$, and $\alpha\in\R^K$---all unknown---consider the signal $x_{\tau,\alpha}(t)=\sum_{i=1}^K \alpha[i]\cdot \delta(t\ominus \tau[i])$ where $\delta(\cdot)$ is the Dirac delta function  and $\ominus$ denotes subtraction with
wraparound on $\I$.\footnote{Later on, we will slightly modify the notation in the interest of mathematical rigor.} The signal $x_{\tau,\alpha}(\cdot)$  can be  characterized
by its Fourier series coefficients $\{\widehat{x}_{\tau,\alpha}[l]\}_l$, 
where 
\begin{equation*}
\widehat{x}_{\tau,\alpha}[l] = \left\langle x_{\tau,\alpha}(t),e^{\i 2\pi lt} \right\rangle_{\I},\qquad l\in\mathbb{Z}.
\end{equation*}

For a cut-off frequency $f_C \in \mathbb{N}$, we wish to recover $K$,  $\tau$, and $\alpha$ from the low-frequency content of $x_{\tau,\alpha}(\cdot)$, namely the coefficients $\{\widehat{x}_{\tau,\alpha}[l]\}$, $|l|\le f_C$. Equivalently, through an ideal low-pass filter with
cut-off frequency $f_C$, we observe $y(t) :=
\sum_{i=1}^K \alpha[i]\cdot  D_{f_C} (t\ominus \tau[i])$  and wish to recover the unknowns. Here, $D_{f_C} (\cdot)$ is the
Dirichlet kernel\footnote{The Dirichlet kernel is sometimes referred to as the ``digital'' sinc.} of width approximately $1/f_C$ in time. 

\subsection{Our Approach}
We focus on estimating the positions $\tau$, since an estimate of the amplitudes $\alpha$ can subsequently be obtained using least-squares. When $K=1$, the matched
filter (e.g.,  \cite{Eftekhari2013}) 
provides the optimal solution to the problem. Our approach is to generalize the matched filter as follows.

We propose to iteratively find the largest peak of the measured signal $y(\cdot)$ and, in order to avoid falsely detecting nearby points in subsequent iterations, erase the neighborhood of each peak. Unfortunately, because of the heavy tail and slow decay of the Dirichlet kernel, this approach is only effective when the impulses are widely separated, the noise is negligible, and the dynamic range $\max_i\left|\alpha[i]\right|/\min_i\left|\alpha[i]\right|$ is small.

To overcome this setback, we first \emph{filter}\footnote{Filtering a signal $a(\cdot)$ with another signal $b(\cdot)$ (both in $L_2(\I)$) produces their  \emph{circular convolution} $[a\circledast b](\cdot)=\int_{\I} a(t)\cdot b(\cdot\ominus t)\, dt \in L_{1}(\I)$. } the measurement signal $y(\cdot)$ with a \emph{kernel} $g_{\sigma,f_C}(\cdot)$ that is band-limited to $[-f_C,f_C]$ in frequency and decays rapidly outside of the (typically small) interval $[-\sigma,\sigma]$ in time. More specifically, set  $N=2f_C+1$ for short. After setting $\sigma=\frac{c}{N}$ for a factor $c$, our approach is to first filter $y(\cdot)$ with $g_{\sigma,N}(\cdot)$ and then iteratively select the peaks of the output of the filter, while removing the neighborhood of each peak to avoid false detections (as outlined in the previous paragraph).

The obtained estimate of the position vector $\tau$ can then be refined by posing super-resolution as a non-convex program, which we solve (using the projected Newton's method) with the output of the greedy search above as the initial point.  

For the choice of kernel $g_{\sigma,N}(\cdot)$, we recommend the top \emph{discrete prolate spheroidal wave function} (DPSWF)~\cite{Slepian1978}.\footnote{DPSWFs are also known as the ``Slepian functions'' in honor of David S.\ Slepian.} Given $\sigma\in (0,\frac{1}{2})$ and integer $N=2f_C+1$, the top DPSWF $\psi_{0,\sigma,N}(\cdot)$ is optimal in that, among all signals supported on $\I$ in time and $[-f_{C},f_{C}]$ in frequency, $\psi_{0,\sigma,N  }(\cdot)$ is maximally concentrated (in $L_2$ sense) on the small interval $[0,\sigma]\cup [1-\sigma,1)$ in time (see Figure \ref{fig:fig1a}). 
%

The resulting ``two-phase'' algorithm is very fast, in part because fast and convenient  means for generating DPSWFs exist \cite{Osipov2013}. Moreover,  in the absence of noise, this algorithm exactly recovers the impulse positions. As the noise level increases, the quality of the output gradually deteriorates. We will  thoroughly verify these claims in later sections. 

As an example, let the cut-off frequency $f_{C}=50$, and set 
$$ \tau=[0.2995 ~ 0.3663 ~ 0.4332 ~ 0.5000 ~ 0.5668 ~ 0.6337 ~ 0.7005]^T,\qquad \mbox{(positions)}$$
$$\alpha=[10 ~ -1 ~ 1 ~ -3 ~ 2 ~ -5 ~ 2]^T.\qquad \mbox{(amplitudes)}$$ 
The measured (low-frequency) signal $y(\cdot)$ is depicted in Figure  \ref{fig:fig1b}. Note that the impulses are separated by roughly only $3/f_C$. We then set $\sigma=\frac{3/2}{2f_C+1}=0.0149$ for the top DPSWF $\psi_{0,\sigma,N}(\cdot)$. In this case, the greedy step produces an estimate $\widehat{\tau}$ which satisfies $\|\widehat{\tau}-\tau\|_{\infty}\le 0.001$. This estimate is then refined via Newton's method to recover $\tau$ perfectly up to machine precision.
\begin{figure}[H]
\begin{center}

\subfigure[]{
\includegraphics[width=3in]{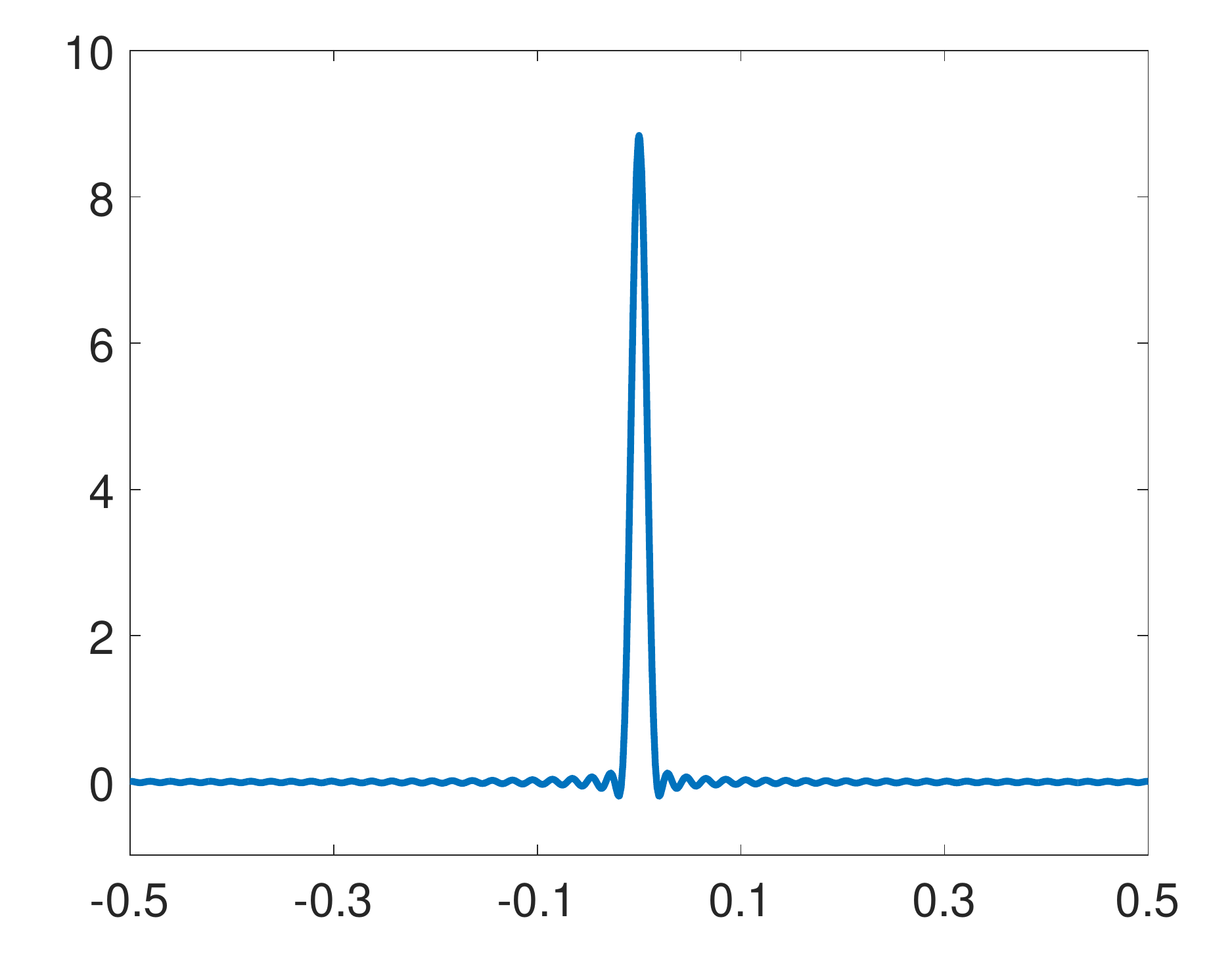}
\label{fig:fig1a}
}

\subfigure[]{
\includegraphics[width=3in]{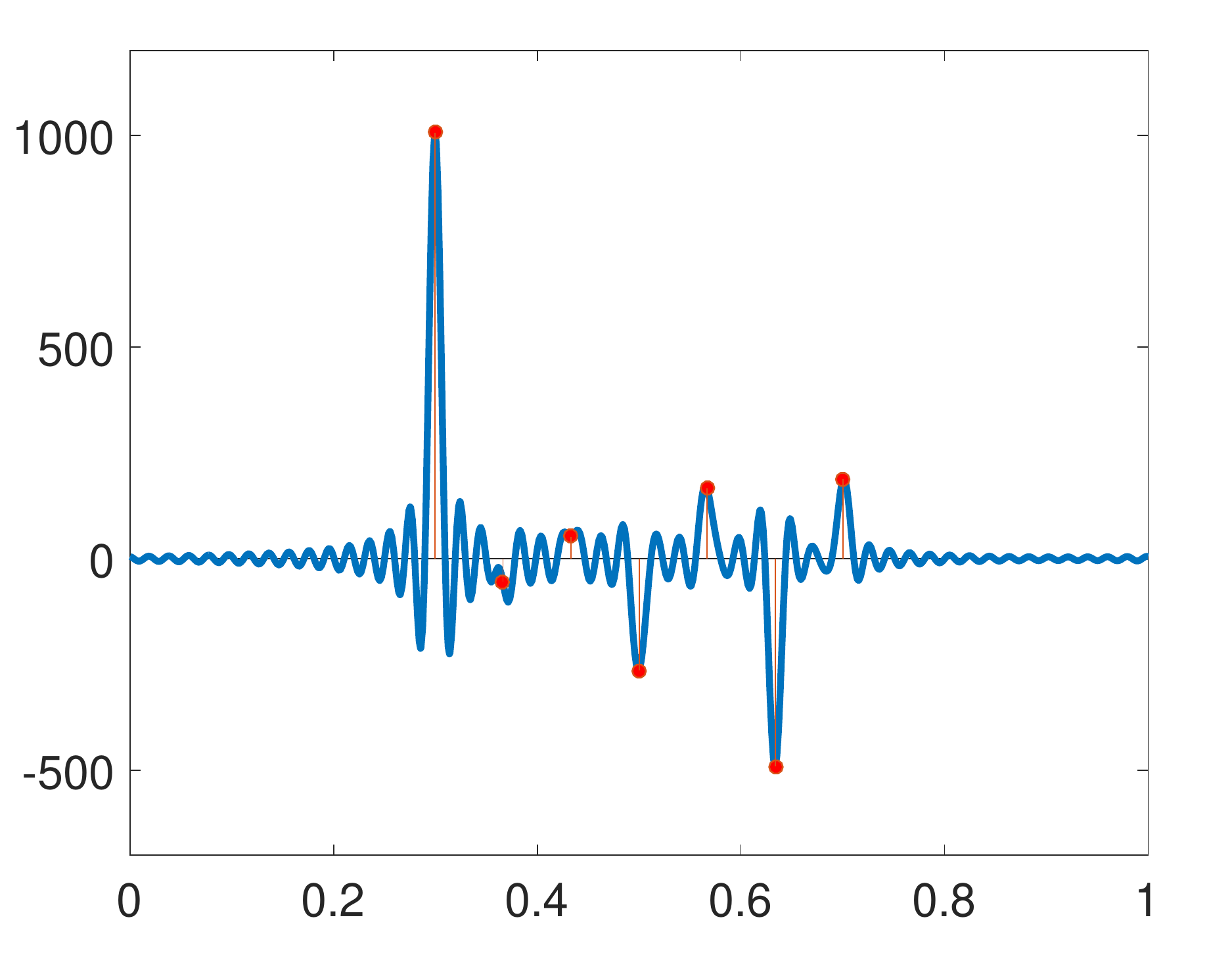}
\label{fig:fig1b}
}
\subfigure[]{
\includegraphics[width=3in]{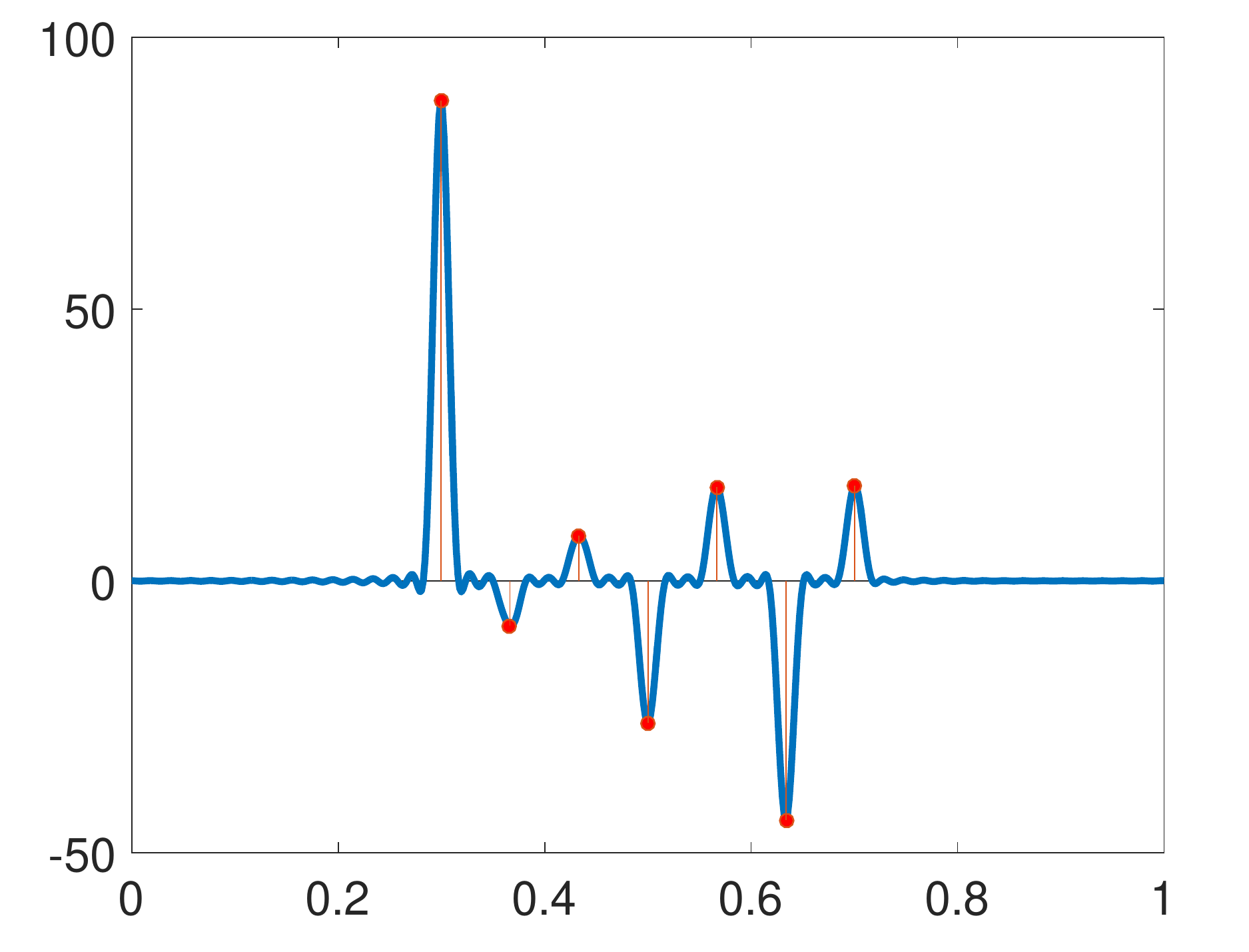}
\label{fig:fig1c}
}

\caption{(a) Graph of the top DPSWF $\psi_{0,\sigma,N}(\cdot)$ versus time for bandwidth $f_C=50$ and \emph{effective} duration of approximately $2\sigma = \frac{3}{2f_C+1}$ in time. Note the sharp decay away from the origin.  (For clarity, the domain here is $[-1/2,1/2)$ instead of $[0,1)$ in the text.)  (b) An example of a low-resolution signal (blue) and original impulse positions (red). (c) Signal in part (b) filtered by  $\psi_{0,\sigma,f_C}(\cdot)$. Note that the peaks provide a good estimate of the unknown impulse positions. The two-phase algorithm in this paper builds on this insight to return the precise location of impulses.     The horizontal axis in all graphs show the time-domain.}
\label{fig:fig1}
\end{center}
\end{figure}

There are different ways in which this super-resolution algorithm may be generalized.  An extension to higher dimensions is of interest in, say, image processing, and replacing the Dirac delta function with a general template establishes a connection with the broad existing literature on \emph{deconvolution} \cite{Meister2009}.

\subsection{Organization}
This paper is organized as follows. Section \ref{sec:Setup 1} gives a formal statement of the problem and collects the notation. The two-phase algorithm for super-resolution is developed in Sections \ref{sec:A-greedy-algorithm signal domain} and \ref{sec:Phase-II:-Local}. Phase I consists of a greedy algorithm that initializes the local optimization in Phase II. The final product is presented in Algorithms I and II (on pages \pageref{fig:alg I} and \pageref{fig:alg II}, respectively) and is accessible even without reading the rest of the paper. The MATLAB code for the two-phase algorithm is also available online.\footnote{http://inside.mines.edu/$\sim$mwakin/publications.html$\#$software}

Following the publication of \cite{Candes2014}, a steady stream of  good research has gradually enriched our knowledge of this topic. Among others, \cite{Candes2014} was followed by  \cite{Tang2013,Duval2015,Azais2015,Fyhn2013,
Demanet2013,Liao2014,Fannjiang2012}. A brief survey and comparison is presented in Section \ref{sec:prior art SUPER}, which is by no means exhaustive. We remark that an excerpt of this work previously appeared in \cite{Eftekhari2013a}.

The theoretical guarantees for our algorithm  consist of Proposition \ref{thm:algorithm 1} for Phase I and Theorem \ref{thm:grad des algo} for Phase II (and are proved in Sections \ref{sec:proof of phase I} and \ref{sec:Proof-of-Theorem grad descent}, respectively). Our results are asymptotic and hold as $f_C\rightarrow\infty$. Naturally, these results rely heavily on certain asymptotic ($N=2f_C+1\rightarrow\infty$) properties of the kernel $g_{\sigma,N}(\cdot)$ (that we identify and collect in Criteria \ref{cri:(Sharp-decay-of}, \ref{lem:corr decay away}, and \ref{fact: slow decay in bound-1}). 
The decision to opt for asymptotic guarantees was  driven by the asymptotic nature of the existing machinery to study the properties of DPSWFs (the recommended kernel here). Indeed, based on empirical observations and preliminary analysis, we conjecture that DPSWFs satisfy these  asymptotic criteria; formally proving this conjecture remains a topic of ongoing work.\footnote{We may add that the asymptotic properties of DPSWFs are not fully understood in the particular  setting studied here (where $\sigma \propto \frac{1}{f_C}$ as opposed to constant $\sigma$ as in \cite{Slepian1978,Slepian1965}). We also remark that, if not DPSWFs, it remains conceivable that some other functions may satisfy these or similar criteria.} Finally, despite the asymptotic nature of our results,  the proposed two-phase algorithm is successful in simulations with cut-off frequency $f_C$ as low as $50$ (as we saw in the example earlier).

Lastly, to keep this paper short, we deferred the elementary calculations to the accompanying document \cite{support}.

\subsection{Contributions}\label{sec:contributions}
In this work, we develop and analyze a two-phase algorithm to resolve impulses from low-pass frequency information. Greedy algorithms for super-resolution have already appeared in the literature  \cite{Fannjiang2012}. However, we are convinced that the present method offers certain advantages (particularly in terms of  computational complexity) that are absent from the existing literature (see Section \ref{sec:prior art SUPER}). 

We tend to  hold a similar conviction regarding  the theoretical contribution of this work. In this aspect, precedents for a two-phase approach based on a good initialization  followed by local optimization have appeared in other contexts \cite{Keshavan2009,Candes2015}. Nevertheless, to the best of our knowledge, this is the first work that develops theoretical guarantees for using a  second-order optimization algorithm in the second phase. Rather unfortunately, that laborious task is partially responsible for the large volume of this work.

Another aspect of this work is its use of prolate functions \cite{Hogan2012,Osipov2013,Slepian1978}, and particularly the top DPSWF $\psi_{0,\sigma,N}(\cdot)$. Originally published in a series of landmark papers in 1960s and 1970s, prolate functions---designed as a highly localized basis for band-limited functions---largely influenced  harmonic analysis    for years that followed. As shown here and in \cite{Davenport2012}, prolate functions have the potential to play an important role in a variety of problems in modern signal processing as well.  Perhaps another contribution of our work is then to ignite the interest of readers in these functions and their uses. 

\section{Problem Setup\label{sec:Setup 1}}

Consider the interval $\I=[0,1)$ in time. 
 We 
let $\oplus$ and $\ominus$ denote the addition and subtraction operators
modulo one. For example, for $\rho_{1},\rho_{2}\in\I$,
\begin{equation}
d(\rho_{1},\rho_{2}):=\min\left(\rho_{1}\ominus\rho_{2},\rho_{2}\ominus\rho_{1}\right),\label{eq:wrip around}
\end{equation}
 is the \emph{wraparound distance} between $\rho_{1}$ and $\rho_{2}$ 
(see e.g., \cite{Candes2014}).  We study (atomic) measures on $\I$ of the form
\begin{equation}
x_{\tau,\alpha}=\sum_{i=1}^{K}\alpha[i]\cdot \delta_{\tau[i]},\label{eq:xtaualpha}
\end{equation}
for integer $K$, vector of locations $\tau\in\I^{K}$ (with distinct entries), and vector
of amplitudes $\alpha\in\mathbb{R}^{K}$. Here, $\delta_{\tau[i]}$
is the Dirac  measure  translated by $\tau[i]\in\I$. Upon existence, $x_{\tau,\alpha}$ can be completely characterized
by its Fourier series  $\widehat{x}_{\tau,\alpha}$, that
is
\[
x_{\tau,\alpha}=\sum_{l=-\infty}^{\infty}\widehat{x}_{\tau,\alpha}[l]\cdot E_{l}(\cdot),
\]
with
\begin{equation}
\widehat{x}_{\tau,\alpha}[l]=\left\langle x_{\tau,\alpha},E_{l}(\cdot)\right\rangle =\sum_{i=1}^{K}\alpha[i]\cdot e^{-\imag2\pi l\tau[i]},\qquad l\in\mathbb{Z}.\label{eq:F coeffs of train}
\end{equation}
\begin{equation*}
E_l(t):=e^{\i 2\pi l t},\qquad t\in\mathbb{R}.
\end{equation*}
{Inner products everywhere are computed on $\I$ in this work.}

Given only the (possibly noisy) low-frequency content of $x_{\tau,\alpha}$,
we wish to infer the number of impulses $K$, positions $\tau\in \I^K$, and amplitudes $\alpha\in\R^K$. More
specifically, for cut-off frequency $f_{C}\in\N$, if 
\begin{equation}\label{eq:main}
y(\cdot):=\left[\mathcal{Q}_{\F}\left(x_{\tau,\alpha}+n\right)\right](\cdot)\in L_2(\I)
\end{equation}
denotes the (low-frequency) measurement signal, we wish to recover the unknowns: $K$, $\tau$, and $\alpha$.  Here, $\mathcal{Q}_{\F}(\cdot):L_2(\I)\rightarrow L_2(\I)$
is the ideal low-pass filter that restricts the frequency content
of its input signal to $\F:=[-f_{C}:f_{C}]=\{-f_C,-f_C+1,\cdots,f_C\}$.
  Also, $n(\cdot)\in L_{\infty}(\I)$ is the
(real-valued) noise signal. 

The low-pass measurement signal $y(\cdot)$ and noise $n(\cdot)$ may be written as
\begin{equation}
y(\cdot)=\sum_{l=-f_{C}}^{f_{C}}\widehat{y}[l]\cdot E_{l}(\cdot),\label{eq:meas}
\end{equation}
\begin{equation}
n(\cdot)=\sum_{l=-\infty}^{\infty}\widehat{n}[l]\cdot E_{l}(\cdot)
,\label{eq:noise}
\end{equation}
where $\widehat{y}=\{\widehat{y}[l]\}_l$ and $\widehat{n}=\{\widehat{n}[l]\}_l$ are the corresponding Fourier series. 
Then, \eqref{eq:main} may equivalently be written as 
\begin{equation}
\widehat{y}[l]=\widehat{x}_{\tau,\alpha}[l]+\widehat{n}[l],\qquad l\in\F=[-f_{C}:f_{C}].\label{eq:meas freq}
\end{equation}
To reiterate, given $y(\cdot)$ or its  nonzero Fourier series coefficients $\{\widehat{y}[l]\}_{l\in\F}$,
we wish to recover $K$, $\tau$, and $\alpha$.

\subsection{Notation}
\label{sec:notation SUPER}
Before going any further, let us collect the notation used throughout this paper. Absolute constants are denoted by $C_{1},C_{2},\cdots$. In addition, $C$ denotes  a constant that  might change in each appearance. We will occasionally use the convention that $[a:b]=\{a,a+1,\cdots,b\}$ for integers $a\le b$. 

The standard asymptotic notation is freely used in this work  and is reviewed 
next for the reader's convenience.
\begin{itemize}
\item For functions $a,b:\mathbb{C}\rightarrow\mathbb{C}$, $a(\theta)=O(b(\theta))$
asymptotically as $\theta\rightarrow\infty$ if there exists positive
constants $\Cl{Odef}$ and $\Cl{Odef2}$ such that
\[
\left|a(\theta)\right|\le\Cr{Odef}\cdot \left|b(\theta)\right|,\qquad|\theta|>\Cr{Odef2}.
\]
\item We use the conventions \cite{Knuth1976} that
\[
a(\theta)=\Omega\left(b(\theta)\right)\Longleftrightarrow b(\theta)=O\left(a(\theta)\right),
\]
\[
a(\theta) =\Theta\left(b(\theta)\right) \Longleftrightarrow 
a(\theta) = O\left(b(\theta)\right) \mbox{ and } a(\theta) = \Omega\left(b(\theta)\right),
\]
asymptotically as $\theta\rightarrow\infty$.
\item Lastly, $a(\theta)=o(b(\theta))$ asymptotically as $\theta\rightarrow\infty$
if, for every $\epsilon>0$, there exists $\delta=\delta(\epsilon)$
such that
\[
|a(\theta)|\le\epsilon\cdot |b(\theta)|,\qquad |\theta|>\delta(\epsilon).
\]
In particular, as long as $\lim_{\theta\rightarrow \infty} b(\theta)\ne 0$, we have that
\[
a(\theta)=o\left(b(\theta)\right)\Longleftrightarrow\lim_{\theta\rightarrow\infty}\frac{a(\theta)}{b(\theta)}=0.
\]
\end{itemize}

Recall that $\I=[0,1)$. The natural norms on $L_{2}(\I)$,
$L_{1}(\I)$, and $L_{\infty}(\mathbb{I})$ are denoted by the shorthands  $\|\cdot\|_{L_{2}}$,
$\|\cdot\|_{L_{1}}$, and $\|\cdot\|_{L_{\infty}}$, respectively.
For the wraparound metric $d(\cdot,\cdot)$ defined in (\ref{eq:wrip around}), the Hausdorff distance between
sets  $\mathbb{A}$ and $\mathbb{B}$ (both subset of $\mathbb{I}$) is defined as
\begin{equation}\label{eq:Hausdoff dist}
d(\mathbb{A},\mathbb{B}):=\max\left\{ \sup_{a\in \mathbb{A}}\inf_{b\in \mathbb{B}}d(a,b)\,,\,\sup_{b\in \mathbb{B}}\inf_{a\in \mathbb{A}}d(a,b)\right\} .
\end{equation}
Effectively, $d(\mathbb{A},\mathbb{B})$ controls the distance from any point on $\mathbb{A}$ to $\mathbb{B}$ and vice verse.  
With some abuse of notation, we define the Hausdoff distance of two
vectors in the natural way (as the distance between the finite sets formed
by their entries).

Throughout, $\circledast$ stands for circular convolution, which corresponds to point-wise multiplication in the Fourier series domain. 
 Lastly, to unburden the notation, we occasionally suppress the dependence on different quantities if there is no ambiguity. 

%
%
\section{Phase I: Initialization \label{sec:A-greedy-algorithm signal domain}}

Rather than recovering the amplitudes $\alpha$, we focus
on estimating the positions of impulses $\tau$. Indeed, given
an estimate of the positions, an estimate for the amplitudes
 readily follows from a simple least-squares calculation.

In this section, we present a simple iterative algorithm that, given the noisy
Fourier coefficients of $x_{\tau,\alpha}$ on the interval
$\F=[-f_{C}:f_{C}]$, approximately recovers the position vector $\tau$ under a certain
separability condition (that we will specify shortly). 

This algorithm
 requires a band-limited \emph{kernel}. 
More specifically, let $N:=2f_{C}+1=|\F|$
for short. Then, for $\sigma\in (0,\frac{1}{2})$, we assume that the kernel
$g_{\sigma,N}(\cdot)=g(\cdot;\sigma,N)\,:\,\I\rightarrow\mathbb{R}$
is band-limited to $\mathbb{F}$ and decays sharply away from the
origin so that $|g_{\sigma,N}(t)|$ is small when $t\in [\sigma,1-\sigma]$. The next statement formally lists the
requirements on the kernel.
\begin{criterion}
\label{cri:(Sharp-decay-of}For integer $f_{C}$, set $N=2f_{C}+1$
for short. For $c=c(N)>0$, let $\sigma=\frac{c}{N}<\frac{1}{2}$.
The kernel $g_{\sigma,N}(\cdot)\,:\,\I\rightarrow\mathbb{R}$
satisfies the following requirements.
\begin{itemize}
\item First, $g_{\sigma,N}(\cdot)$ has unit-energy, $\|g_{\sigma,N}(\cdot)\|_{L_{2}}=1$,
and is band-limited to $\F=[-f_{C}:f_{C}]$ (so that the Fourier coefficients $\{\widehat{g}_{\sigma,N}[l]\}$ vanish when $l\notin\F$). 

\item Second, $g_{\sigma,N}(\cdot)$ is symmetric about $\frac{1}{2}$ (so that 
$g_{\sigma,N}(t)=g_{\sigma,N}(1-t)$ for every $t\in\I$).

\item Lastly, as $c,N\rightarrow\infty$ with $c=c(N)=O(\log N)$, the decay of $g_{\sigma,N}(\cdot)$ away from the origin is asymptotically quantified as
%
\[
\left|g_{\sigma,N}(t)\right|=\frac{O(e^{-\Cl{decay}c})}{\sqrt{N}\sin\left(\pi t\right)},\qquad\sigma\le t\le\frac{1}{2},
\]
\[
g_{\sigma,N}(0)=\Omega(1/\sqrt{\sigma})=\Omega\left(\sqrt{\frac{N}{c}}\right),
\]for some constant $\Cr{decay}>0$. 
\end{itemize}

\end{criterion}
 As mentioned earlier, the success of the Phase I algorithm, summarized in Figure \ref{fig:alg I}, hinges on Criterion
\ref{cri:(Sharp-decay-of}. Throughout this section, we assume
the existence of a kernel $g_{\sigma_{1},N}(\cdot)$ for which Criterion
 \ref{cri:(Sharp-decay-of} holds with $c=c_{1}=c_{1}(N)$ and $\sigma_{1}=\frac{c_{1}}{N}$. 

\begin{center}
\begin{figure}[h]
\begin{center}

\fbox{\begin{minipage}[t]{.9\columnwidth}%
\textbf{Algorithm I (initialization)}

\vspace*{\bigskipamount}

\textbf{Input:}
\begin{itemize}
\item A cut-off frequency $f_{C}\in\mathbb{N}$ and a measurement signal
$y(\cdot)$ that is band-limited to $\F=[-f_{C}:f_{C}]$ (see (\ref{eq:meas freq})).
\item With $N=2f_{C}+1$ and $0<\sigma_{1}<\frac{1}{2}$, a kernel $g_{\sigma_{1},N}(\cdot)$ (see Criterion \ref{cri:(Sharp-decay-of}).
\item A threshold $\eta>0$.
\end{itemize}
\textbf{Output: }
\begin{itemize}
\item An estimate of $K$ and $\tau$, denoted here by
$\widetilde{K}$ and $\tau^{0}\in\mathbb{I}^{\widetilde{K}}$, respectively.
\end{itemize}
\vspace{2 mm}
\begin{enumerate}
\item Compute $z_{\sigma_{1}}^{1}(\cdot):=(g_{\sigma_{1,N}}\circledast y)(\cdot)$. Here, $\circledast$ stands for circular convolution. 
\item Set $j=1$. As long as $\|z_{\sigma_{1}}^{j}(\cdot)\|_{L_{\infty}}>\eta$,
repeat the following (where $d(\cdot,\cdot)$ is the wraparound metric defined in \eqref{eq:wrip around}):

\begin{enumerate}
\item $\tau^{0}[j]=\arg\max_{t\in\mathbb{I}}\left|z_{\sigma_{1}}^{j}(t)\right|$.
\item $z_{\sigma_{1}}^{j+1}(t)=\begin{cases}
z_{\sigma_{1}}^{j}(t) & d\left(t,\tau^{0}[j]\right)>2\sigma_{1},\\
0 & d\left(t,\tau^{0}[j]\right)\le2\sigma_{1}.
\end{cases}$
\item $j\leftarrow j+1$.
\end{enumerate}
\item Set $\widetilde{K}=j-1$ to be the estimate of number of impulses
$K$. Also, return $\tau^{0}\in\mathbb{I}^{\widetilde{K}}$ as the
estimate of their locations $\tau$. \end{enumerate}

\end{minipage}}
\end{center}

\caption{Algorithm I (initialization)}
\label{fig:alg I}
\end{figure}
\end{center}

%

In a nutshell, Algorithm I iteratively finds the largest peaks of
$z_{\sigma_{1}}^{1}(\cdot)=[g_{\sigma_{1,N}}\circledast y](\cdot)$, and in order to avoid
(falsely) detecting the nearby points in the next iteration, erases
the neighborhood (of radius $2\sigma_{1}$) of each peak. In fact,
Algorithm I may loosely be considered as the extension of \emph{orthogonal
matching pursuit }to a continuous domain \cite{Mallat1993}.

As we describe next, under Criterion \ref{cri:(Sharp-decay-of}, Algorithm
I returns a reliable estimate of $K$ and $\tau$ as long as the impulse locations
$\{\tau[i]\}_i$ are well-separated and the dynamic range of $x_{\tau,\alpha}$
is not too large. To be concrete, the\emph{ separation }of $\tau$
is defined as follows \cite{Candes2014}: \begin{equation}\label{eq:sep}
\mbox{sep}(\tau):=\min_{\substack{i,j\in[1:K]\\ i\ne j}}d\left(\tau[i],\tau[j]\right).
\end{equation}In addition, we define the dynamic range of $x_{\tau,\alpha}$ as
follows:
\begin{equation}
\mbox{dyn}(x_{\tau,\alpha}):=\frac{\max_{i\in[1:K]}\left|\alpha[i]\right|}{\min_{i\in[1:K]}\left|\alpha[i]\right|}.\label{eq:dyn range def}
\end{equation}
\-The performance guarantee for Algorithm 1 is summarized below and
proved in Section \ref{sec:proof of phase I}.
\begin{prop}
\label{thm:algorithm 1}\textbf{\emph{[Performance of Algorithm I]}} Fix a measure $x_{\tau,\alpha}$ with number
of impulses $K$, vector of positions $\tau\in\I^{K}$ with distinct entries, and vector
of amplitudes $\alpha\in\mathbb{R}^{K}$ defined as in (\ref{eq:xtaualpha}). With the cut-off frequency $f_{C}\in\N$ and $\F=[-f_{C}:f_{C}]$,
let $y(\cdot)$ be the (possibly noisy) measurement signal band-limited to $\F$. The
nonzero Fourier coefficients of $y(\cdot)$ are
\[
\widehat{y}[l]=\widehat{x}_{\tau,\alpha}[l]+\widehat{n}[l],\qquad l\in\F,
\]
where $\widehat{x}_{\tau,\alpha}$ and $\widehat{n}$ are the Fourier
series  of $x_{\tau,\alpha}$ and the noise $n(\cdot)$,
respectively (see (\ref{eq:meas freq})).

For $N=2f_{C}+1$ and $c_{1}=c_{1}(N)>0$, set $\sigma_{1}=\frac{c_{1}}{N}<\frac{1}{2}$.
In what follows, $c_{1},N\rightarrow\infty$, and $c_{1}=\Theta(\log N)$ with a sufficiently large lower bound.%
\footnote{That is, for  large enough factors $\alpha \le \beta$ specified in the proof
and for sufficiently large $N$, we assume that $\alpha \log N \le c_{1}\le \beta \log N$.%
} Suppose that the kernel $g_{\sigma_{1,N}}(\cdot)$ satisfies Criterion
\ref{cri:(Sharp-decay-of}, and that the threshold $\eta$ in Algorithm
I is specified as
\[
\eta=2\|n(\cdot)\|_{L_{\infty}}.
\]
Then, the output of Algorithm I assymptotically (i.e., for large enough $N$)\footnote{In particular, $N$ must be large enough so that $\mbox{sep}(\tau)\ge 4\sigma_{1}$ and $\mbox{dyn}(x_{\tau,\alpha})= O(\mbox{NSR}/\sqrt{\sigma_1})$, where $\mbox{NSR}$ stands for the noise-to-signal ratio (as specified in \eqref{eq:pre SNR}).}    satisfies the following: 
\begin{itemize} \item $\widetilde{K}=K$, i.e., Algorithm I correctly estimates the number of impulses, and
\item $d(\tau^0,\tau) \le \sigma_1$, i.e., the Hausdorff distance between the vector of true positions $\tau$ and the estimates returned by Algorithm I is small.
\end{itemize}
\end{prop}

\begin{rem}
Algorithm I returns an initial estimate of $\tau\in\mathbb{I}^{K}$,
namely $\tau^{0}\in\I^{K}$. In the second part of this work,
we refine this initial estimate by solving a local optimization program.
In particular, asymptotically, we will be able to recover $\tau$
exactly from noise-free low-frequency measurements.
\end{rem}
%

\section{Phase II: Local Optimization\label{sec:Phase-II:-Local}}

This section presents a method to refine the estimate of $\tau$ produced
by Algorithm I (namely,  $\tau^0\in\mathbb{I}^{\widetilde{K}}$).  In particular, asymptotically and in the absence of noise, we will be
able to exactly recover $\tau$. 

With $c_{2}=c_{2}(N)>0$ to be specified later, set $\sigma_{2}=\frac{c_{2}}{N}$%
. For cut-off frequency $f_{C}$, recall that $\F=[-f_{C}:f_{C}]$
and that $N=|\F|=2f_{C}+1$. Throughout this section, we consider
the unit-energy kernel $g_{\sigma_{2},N}(\cdot)$ which is band-limited
to $\F$ by design. We first filter the low-frequency measurement signal
$y(\cdot)$ with the kernel $g_{\sigma_{2},N}(\cdot)$. More precisely, 
for $t\in\I$, we set
\begin{align}
z_{\sigma_{2}}(t) & :=\left(g_{\sigma_{2},N}\circledast y\right)(t)\nonumber \\
 & =\left(g_{\sigma_{2},N}\circledast x_{\tau,\alpha}\right)(t)+\left(g_{\sigma_{2},N}\circledast n\right)(t)\qquad\mbox{(see (\ref{eq:meas freq}) and the text below)}\nonumber \\
 & =\sum_{i=1}^{K}\alpha[i]\cdot g_{\sigma_{2},N}(t\ominus\tau[i])+\left(g_{\sigma_{2},N}\circledast n\right)(t)
\qquad \mbox{(see \eqref{eq:xtaualpha})}
 \nonumber \\
 & =:\sum_{i=1}^{K}\alpha[i]\cdot g_{\sigma_{2},N}(t\ominus\tau[i])+n_{\sigma_{2}}(t),\qquad\mbox{(}n_{\sigma_{2}}(\cdot)\mbox{ is the filtered noise)}\label{eq:filt meas}
\end{align}
where the second line above uses the assumption that $g_{\sigma_{2},N}(\cdot)$
is also band-limited to $\F$. Let $\widehat{g}_{\sigma_{2},N}$ and $\widehat{n}_{\sigma_{2}}$
be the corresponding Fourier  series. 
Note that
$\widehat{g}_{\sigma_{2},N}$ is supported only on the interval $\F$
and so is $\widehat{n}_{\sigma_{2}}$. In light of  (\ref{eq:filt meas}),
the Fourier coefficients of $z_{\sigma_{2}}(\cdot)$ can then be written
as
\begin{align}
\widehat{z}_{\sigma_{2}}[l] & =\widehat{g}_{\sigma_{2},N}[l]\cdot\widehat{y}[l]\nonumber \\
 & =\widehat{g}_{\sigma_{2},N}[l]\sum_{i=1}^{K}\alpha[i]e^{-\imag2\pi l\tau[i]}+\widehat{n}_{\sigma_{2}}[l],\qquad l\in\mathbb{F},\label{eq:phase II meas}
\end{align}
where the second line follows from a direct calculation.
For a more compact representation, we abuse the notation by letting $\widehat{z}_{\sigma_{2}},\widehat{n}_{\sigma_{2}}\in\mathbb{C}^{N}$ also 
denote the vectors formed by the Fourier series coefficients
of $z_{\sigma_{2}}(\cdot)$ and $n_{\sigma_{2}}(\cdot)$ on $\F$, respectively.
Then, the vector form of \eqref{eq:phase II meas} is simply
\begin{equation}
\widehat{z}_{\sigma_{2}}=G_{\tau}\cdot\alpha+\widehat{n}_{\sigma_{2}},\label{eq:def of z hat}
\end{equation}
where $G_{\tau}\in\mathbb{C}^{N\times K}$ is constructed out of the (modulated) 
Fourier coefficients of the kernel:
\begin{equation}\label{eq:def of G_tau}
G_{\tau}[l,i]=\widehat{g}_{\sigma_{2},N}[l]\cdot e^{-\imag2\pi l\cdot\tau[i]},\qquad l\in\F,\: i\in[1:K].
\end{equation}
Alternatively, $G_{\tau}\alpha\in\mathbb{C}^{N}$ is the vector
formed by the Fourier series coefficients of the filtered measure $(g_{\sigma_{2},N}\circledast x_{\tau,\alpha})(\cdot)$ on $\mathbb{F}$. 

In this section, we will rely on $g_{\sigma_{2},N}(\cdot)$ satisfying
a number of properties. The first criterion, among other things,  specifies
how small the correlation between the kernel and its shifted copy should be (when separated properly). The second criterion concerns the behavior of kernel near the origin (and that it must be ``flat'' in a very small interval near the origin). 

\begin{criterion}
\label{lem:corr decay away} For integer $f_{C}$, set $N=2f_{C}+1$ for short. For $c=c(N)>0$, let $\sigma=\frac{c}{N}<\frac{1}{2}$. The kernel
$g_{\sigma,N}(\cdot)$ has unit energy $\|g_{\sigma,N}(\cdot)\|_{L_{2}}=1$,
is band-limited to $\F$, and symmetric about $\frac{1}{2}$ (as in Criterion \ref{cri:(Sharp-decay-of}). Moreover,
$g_{\sigma,N}(\cdot)$ satisfies the following.

For $\rho_{1},\rho_{2}\in\mathbb{I}$ with $d(\rho_{1},\rho_{2})\ge2\sigma$,
it holds asymptotically that
\begin{equation}
\left|\left\langle g_{\sigma,N}(t\ominus\rho_{1}),g_{\sigma,N}(t\ominus\rho_{2})\right\rangle \right|=\frac{O(e^{-\Cr{decay}c})}{N\cdot \sin\left(\pi\cdot d(\rho_{1},\rho_{2})\right)},\label{eq:far corr of signal}
\end{equation}
\begin{equation}
\left|\left\langle g_{\sigma,N}(t\ominus\rho_{1}),g_{\sigma,N}'(t\ominus\rho_{2})\right\rangle \right|=\frac{O(e^{-\Cr{decay}c})}{\sin\left(\pi\cdot d(\rho_{1},\rho_{2})\right)},\label{eq:far cor of der}
\end{equation}
\begin{equation}
\left|\left\langle g'_{\sigma,N}(t\ominus\rho_{1}),g_{\sigma,N}'(t\ominus\rho_{2})\right\rangle \right|=\frac{N\cdot O(e^{-\Cr{decay}c})}{\sin\left(\pi\cdot d(\rho_{1},\rho_{2})\right)},\label{eq:far cor of der der}
\end{equation}
when $c,N\rightarrow\infty$ and $c=O(\log N)$. Here, $d(\rho_{1},\rho_{2})$
is the wraparound distance between $\rho_{1}$ and $\rho_{2}$ (see (\ref{eq:wrip around})),  and $g_{\sigma,N}'(\cdot)$
denotes the derivative of $g_{\sigma,N}(\cdot)$ with respect to its
argument. 
\end{criterion}

\begin{criterion}
\label{fact: slow decay in bound-1}For integer $f_{C}$, set $N=2f_{C}+1$ for short.
For $c=c(N)>0$, let $\sigma=\frac{c}{N}<\frac{1}{2}$. 
 There exists $h(\sigma,N)\le \sigma$ (depending only on $\sigma$ and $N$), for which the kernel  $g_{\sigma,N}(\cdot)$
satisfies the following. 

For $\rho_{1},\rho_{2}\in\mathbb{I}$, suppose that $d(\rho_{1},\rho_{2})\le h(\sigma,N)\le\sigma$.
Then, it holds asymptotically that $\|g_{\sigma,N}'(\cdot)\|_{L_{2}}^{2}=\Omega(N^{2})$, 
and furthermore  
\begin{equation}
\left\langle g_{\sigma,N}(t\ominus\rho_{1}),g_{\sigma,N}(t\ominus\rho_{2})\right\rangle =1-O(1)\cdot d^2(\rho_{1},\rho_{2}),\label{eq:near auto corr}
\end{equation}
\begin{equation}
\left|\left\langle g_{\sigma,N}(t\ominus\rho_{1}),g_{\sigma,N}'(t\ominus\rho_{2})\right\rangle \right|=\Omega\left(N^{2}\right)\cdot d(\rho_{1},\rho_{2}),\label{eq:near der corr}
\end{equation}
\begin{equation}
\mbox{sign}\left(\left\langle g_{\sigma,N}(t\ominus\rho_{1}),g_{\sigma,N}'(t\ominus\rho_{2})\right\rangle \right)=\mbox{sign}\left(\rho_{1}\ominus\rho_{2}-\frac{1}{2}\right),\label{eq:sign match}
\end{equation}
as $c,N\rightarrow\infty$ with $c=O(\log N)$. Above, $\mbox{sign}(\cdot)$ returns the sign, of course.  
\end{criterion}
Throughout this section, we assume that the kernel $g_{\sigma_{2},N}(\cdot)$
satisfies both Criteria \ref{lem:corr decay away} and \ref{fact: slow decay in bound-1}
with $\sigma=\sigma_{2}=\frac{c_{2}}{N}$. We will soon specify $c_{2}=c_{2}(N)$
in relation to $c_{1}$ (from Phase I).

 Define $G_{\rho}\in\mathbb{C}^{N\times \widetilde{K}}$ similar to \eqref{eq:def of G_tau} (but with $\widetilde{K}$ instead of $K$), and consider the objective function
\begin{equation}\label{eq:def of f func}
f(\rho,\beta):=\left\Vert G_{\rho}\beta-\widehat{z}_{\sigma_{2}}\right\Vert _{2}^{2},\qquad\rho\in\I^{\widetilde{K}},\:\beta\in\mathbb{R}^{\widetilde{K}},
\end{equation}
with $\rho$ and $\beta$ being vectors of positions and amplitudes, 
respectively. For a fixed $\rho\in\I^{\widetilde{K}}$, minimizing $f(\rho,\cdot)$
is a simple least-squares problem:
\begin{align}
  \min_{\rho\in\I^{\widetilde{K}}}
  \min_{\beta\in\mathbb{R}^{\widetilde{K}}}
  f(\rho,\beta) & =\min_{\rho\in\I^{\widetilde{K}}}
  \min_{\beta\in\mathbb{R}^{\widetilde{K}}}\left\Vert G_{\rho}\beta-\widehat{z}_{\sigma_{2}}\right\Vert _{2}^{2}
  \nonumber \\
& =\min_{\rho\in\I^{\widetilde{K}}}f\left(\rho,\beta_{\rho}\right)\qquad\left(\beta_{\rho}:=G_{\rho}^{\dagger}\cdot\widehat{z}_{\sigma_{2}}\right)\nonumber \\
& =\min_{\rho\in\I^{\widetilde{K}}}
\left\Vert \left(I_{N}-\mathcal{P}_{\rho}\right)\widehat{z}_{\sigma_{2}}\right\Vert _{2}^{2}\qquad\left(\mathcal{P}_{\rho}=G_{\rho}G_{\rho}^{\dagger}\in \mathbb{C}^{N\times N}\right)\nonumber \\
 & =:\min_{\rho\in\I^{\widetilde{K}}}F(\rho).\label{eq:main program}
\end{align}
Above, $G_{\rho}^{\dagger}$ is the Moore-Penrose pseudo-inverse of
$G_{\rho}\in\mathbb{C}^{N\times \widetilde{K}}$. Also, $\mathcal{P}_{\rho}=G_{\rho}G_{\rho}^{\dagger}$
is the orthogonal projection onto $\mbox{span}(G_{\rho})$.

Suppose that Proposition \ref{thm:algorithm 1} is in force so that, in particular, $\widetilde{K}=K$. 
Now,  \eqref{eq:def of z hat} suggests that minimizing $F(\cdot)$ in \eqref{eq:main program}  might reliably estimate the true vector of
positions $\tau\in\I^K$. In fact, in the absence of noise,
$\tau$ is indeed a solution to Program \eqref{eq:main program} (with $\Kt=K$).\footnote{In general, given
an estimate $\widetilde{\tau}\in\I^{\Kt}$ of $\tau\in\I^K$, an estimate of amplitudes
$\alpha\in\mathbb{R}^{K}$ is simply $\beta_{\widetilde{\tau}}=G_{\widetilde{\tau}}^{\dagger}\cdot\widehat{z}_{\sigma_{2}}\in\R^{\Kt}$.}

However, even in the absence of noise, 
the super-resolution problem (Program \eqref{eq:main program}) might
have multiple local minima in which an optimization algorithm might
get trapped. 
The key insight that  resolves this issue is that, under Proposition \ref{thm:algorithm 1},   the outcome of Algorithm
I (namely, $\tau^{0}\in\mathbb{I}^{\Kt}$ with $\widetilde{K}=K$) is close enough to $\tau$
so that a local optimization algorithm (initialized at $\tau^{0}$)
converges to $\tau$ (or its small vicinity).

To formalize matters, we cast the local optimization step    as follows. Under Proposition \ref{thm:algorithm 1}, recall that   $d(\tau^{0},\tau)\le\sigma_{1}$.
To incorporate this prior knowledge, we add this constraint to Program
\eqref{eq:main program} to obtain the box-constrained program
\begin{equation}
\min_{\rho\in\mathbb{B}(\tau^{0},\sigma_{1})}F(\rho),\label{eq:regularized problem}
\end{equation}
where
\begin{equation}
\mathbb{B}(\tau^{0},\sigma_{1}):=\left\{ \rho\in\I^{\widetilde{K}}\,:\, d(\rho,\tau^{0})\le\sigma_{1}\right\} \subset\I^{\widetilde{K}},\label{eq:ball}
\end{equation}
is a ball of radius $\sigma_{1}$ centered at the initial estimate
$\tau^{0}$ from Algorithm I.\footnote{Note that we used $\widetilde{K}$ to define the ball (instead of $K$), so as to develop Phase II independent of Proposition \ref{thm:algorithm 1}. The theoretical guarantees for Phase II, however, do indeed depend on the success of Phase I. Specifically, when it comes to the theory of Phase II, we will assume that Proposition \ref{thm:algorithm 1} is in force: $\tau^0\in \I^{\Kt} $ with $\widetilde{K}=K$,  and $d(\tau^0,\tau)\le \sigma_1$. } Given $\tau^{0}$, one might use any
constrained optimization algorithm to solve Program \eqref{eq:regularized problem}.

Before discussing two such algorithms, let us shed  light on the
geometry of the ball $\mathbb{B}(\tau^{0},\sigma_{1})$, when the entries of $\tau^0$ are distinct,\footnote{For example,  under Proposition
\ref{thm:algorithm 1},  the   entries of $\tau^{0}$ are distinct
asymptotically, i.e., for sufficiently large $N$.} whereby 
\[
\mathbb{B}(\tau^{0},\sigma_{1})=\left\{ \rho\in\I^{\Kt}\,:\, d(\rho[i],\tau^{0}[i])\le\sigma_{1},\quad i\in[1:\Kt]\right\}.
\]
when $N$ is sufficiently large. 
If the entries of $\tau^{0}$ are distinct \emph{and} away from
the origin, then we have the simpler expression 
\begin{equation}
\mathbb{B}(\tau^{0},\sigma_{1})=\Big[\tau^{0}[1]-\sigma_{1},\tau^{0}[1]+\sigma_{1}\Big]\times\cdots\times\Big[\tau^{0}[\Kt]-\sigma_{1},\tau^{0}[\Kt]+\sigma_{1}\Big].\label{eq:simple box}
\end{equation}
Furthermore, the set of \emph{active coordinates }for $\rho\in\mathbb{B}(\tau^{0},\sigma_{1})$ consists of the coordinates on the boundary of $\mathbb{B}(\tau^0,\sigma_1)$, that is
\begin{equation}
\mathbb{A}(\rho):=\left\{ i\,:\, d\left(\rho[i],\tau^{0}[i]\right)=\sigma_{1}\right\} \subseteq [1:\widetilde{K}],\qquad\rho\in\mathbb{B}(\tau^{0},\sigma_{1}).\label{eq:active constraints}
\end{equation}
When (\ref{eq:simple box}) holds, for instance, $i\in \mathbb{A}   (\rho)$
if simply $\rho[i]=\tau^{0}[i]\pm\sigma_{1}$. Naturally, $\mathbb{A}^C(\rho)$ (namely, the complement of $\mathbb{A}(\rho)$)  consists of \emph{inactive coordinates} of $\rho$.

We now turn to the details of solving Program \eqref{eq:regularized problem}.    The \emph{gradient projection algorithm} is an obvious candidate for a first-order
method here.
\footnote{Alternatively, one may use the \emph{conditional gradient method} instead
of the gradient projection algorithm \cite{Kelley1999}.%
} At iteration $j\ge1$, one sets
\begin{equation}
\tau^{j}=\mathcal{P}_{\mathbb{B}(\tau^{0},\sigma_{1})}\left(\tau^{j-1}-\delta^{j}\cdot\frac{\partial F}{\partial\rho}(\tau^{j-1})\right)\in \mathbb{I}^{\widetilde{K}}
,\label{eq:the grad desc alg}
\end{equation}
with step size $\delta^{j}>0$ at the $j$th iteration\emph{. }Above, $\mathcal{P}_{\mathbb{B}(\tau^{0},\sigma_{1})}(\cdot)$
 (namely, the projection operator 
onto the ball $\mathbb{B}(\tau^{0},\sigma_{1})$) ensures
that $\tau^{j}$ remains a feasible point of Program \eqref{eq:regularized problem}
at the $j$th iteration. We in fact find an explicit expression for
the gradient of $F(\cdot)$ in the supporting document \cite{support}:
\begin{equation}
\frac{\partial F}{\partial\rho}(\rho)=-2\cdot\mbox{diag}(\beta_{\rho})\cdot G_{\rho}^{*}L\left(I_{N}-\mathcal{P}_{\rho}\right)\widehat{z}_{\sigma_{2}}\in \mathbb{R}^{\Kt}
 ,\qquad\rho\in\I^{\widetilde{K}}.\label{eq:grad of F}
\end{equation}
Here, $L\in\mathbb{C}^{N\times N}$ is a diagonal matrix with $L[l,l]=\imag2\pi l$,
$l\in\F$. Also, $\beta_{\rho}=G_{\rho}^{\dagger}\cdot\widehat{z}_{\sigma_{2}}$ and $\mathcal{P}_{\rho}=G_{\rho}G_{\rho}^{\dagger}$. Moreover, $\mbox{diag}(\beta_\rho)$ is the diagonal matrix formed by the vector $\beta_\rho$. Without a formal proof we remark that 
 the gradient projection algorithm converges to $\tau$ (or its
neighborhood when there is noise)%
.

A valuable fact here is that once the gradient projection algorithm
identifies an active coordinate, that coordinate remains unchanged
in future iterations. More specifically, if $d(\tau^j[i],\tau^0[i]) = \sigma_1$, 
 then $d(\tau^{j'}[i],\tau^0[i]) = \sigma_1$ for all future iterations  $j'\ge j$ \cite{Kelley1999}.\footnote{A similar phenomenon is true of any convex feasible set (and not just box constraints).}

From a practical standpoint, however, deploying a first-order method
(such as the gradient projection algorithm above) is somewhat unwise
since the initial estimate $\tau^{0}$ is generally too close to $\tau$
and, as a result, $\frac{\partial F}{\partial\rho}(\tau^{0})\approx0$.%
{} This in turn results in a slow---linear to be precise---convergence
rate.

Actually, the local nature of this problem encourages second-order
methods as a viable alternative here. 
%
To proceed, let 
\[
\frac{\partial^{2}F}{\partial\rho^{2}}(\rho)=\left[\frac{\partial^{2}F}{\partial\rho[i]\partial\rho[j]}(\rho)\right]_{i,j}\in\mathbb{R}^{\Kt\times \Kt}
\]
denote the Hessian of $F(\cdot)$ at $\rho\in\I^{\Kt}$, and 
define the \emph{reduced Hessian} 
at $\rho$ to be
\begin{equation}
\mathbb{R}^{\Kt\times \Kt}\ni\mathcal{R}\left(\frac{\partial^{2}F}{\partial\rho^{2}}(\rho)\right)=\begin{cases}
\delta_{i,j} & i\in \mathbb{A}(\rho)\,\mbox{or}\, j\in \mathbb{A}(\rho),\\
\frac{\partial^{2}F}{\partial\rho[i]\partial\rho[j]}\left(\rho\right) & \mbox{otherwise}.
\end{cases}\label{eq:red Hess}
\end{equation}
Here, $\delta_{i,j}$ is the Kronecker delta function, $\delta_{i,j}=1$
if $i=j$ and $\delta_{i,j}=0$ if $i\ne j$. 
Because $F(\cdot)$ is a smooth function, its reduced  Hessian is positive semi-definite near a solution $\widetilde{\tau}$ of Program \eqref{eq:regularized problem} \cite{Kelley1999}, i.e., 
\[
\mathcal{R}\left(\frac{\partial^{2}F}{\partial\rho^{2}}(\rho)\right)\succcurlyeq 0 ,\qquad\mbox{when }d\left(\rho,\widetilde{\tau}\right)\mbox{ is small}.
\]
Therefore, hypothetically, if $\mathbb{A}(\widetilde{\tau})$ (namely, the active coordinates of $\widetilde{\tau}$) 
were known  and $d(\tau^0,\widetilde{\tau})$ was small, we could have calculated the rest of coordinates of $\widetilde{\tau}$ 
by applying the basic unconstrained Newton's method (using  the reduced Hessian in \eqref{eq:red Hess} and assuming its invertibility). 

Of course, we will not know the active constraints until the problem
is solved. Instead, we must use the \emph{projected Newton's
method}. In words, at each iteration, the projected Newton's method
carefully underestimates the active coordinates. Then, the (estimated)
inactive coordinates are updated using  an (unconstrained) Newton's
step, and the (estimated) active coordinates are in turn updated using
the gradient projection step. Loosely speaking, the fact that
active coordinates remain unchanged under the gradient projection
algorithm is the key to  the success of projected Newton's method.

To formally write down the iterations of the projected Newton's method \cite[Algorithm 5.5.2]{Kelley1999}, we record  a couple 
more definitions. For $\varepsilon>0$, the $\epsilon$-active coordinates
of $\rho$ are collected in the set
\begin{equation}
\mathbb{A}_{\epsilon}(\rho):=\left\{ i\,:\,\sigma_{1}-\epsilon\le d\left(\rho[i],\tau^{0}[i]\right)\le\sigma_{1}\right\} \subseteq[1:\Kt],\qquad \rho\in\mathbb{B}(\tau^{0},\sigma_{1}).
\label{eq:e active const}
\end{equation}
In particular, $\mathbb{A}_0(\rho) = \mathbb{A}(\rho)$ (see \ref{eq:active constraints}).  The $\epsilon$-reduced Hessian is defined similar to \eqref{eq:red Hess} as
\begin{equation}
\mathbb{R}^{\Kt\times \Kt}\ni\mathcal{R}_{\epsilon}\left(\frac{\partial^{2}F}{\partial\rho^{2}}(\rho)\right)=\begin{cases}
\delta_{i,j} & i\in \mathbb{A}_\epsilon(\rho)\mbox{ or } j\in \mathbb{A}_\epsilon(\rho),\\
\frac{\partial^{2}F}{\partial\rho[i]\partial\rho[j]}\left(\rho\right) & \mbox{otherwise}.
\end{cases}\label{eq:e red Hess}
\end{equation}
Also, let us give an
explicit (if not elegant) expression for the Hessian (which is verified
in the accompanying document \cite{support}):    
\begin{align}
\mathbb{R}^{\Kt\times \Kt}\ni \frac{\partial^{2}F}{\partial\rho^{2}}(\rho) & =-2\cdot\mbox{diag}\left(\beta_{\rho}\right)\cdot G_{\rho}^{*}L^{2}G_{\rho}\cdot\mbox{diag}\left(\beta_{\rho}\right)\nonumber \\
 & \qquad-2\cdot\mbox{diag}\left(\beta_{\rho}\right)\cdot\mbox{diag}\left(G_{\rho}^{*}L^{2}\left(I_{N}-\mathcal{P}_{\rho}\right)\widehat{z}_{\sigma_{2}}\right)\nonumber \\
 & \qquad-2\left[\mbox{diag}\left(\beta_{\rho}\right)G_{\rho}^{*}LG_{\rho}-\mbox{diag}\left(G_{\rho}^{*}L\left(I_{N}-\mathcal{P}_{\rho}\right)\widehat{z}_{\sigma_{2}}\right)\right]\nonumber \\
 & \qquad\cdot\left(G_{\rho}^{*}G_{\rho}\right)^{-1}\cdot\left[G_{\rho}^{*}L^{*}G_{\rho}\cdot\mbox{diag}(\beta_{\rho})-\mbox{diag}\left(G_{\rho}^{*}L\left(I_{N}-\mathcal{P}_{\rho}\right)\widehat{z}_{\sigma_{2}}\right)\right].\label{eq:Hess exp}
\end{align}
The quantities involved ($\beta_{\rho}$, $G_{\rho}$, $L$, $\mathcal{P}_{\rho}$,
and $\widehat{z}_{\sigma_{2}}$) were defined earlier. Algorithm II (in Figure \ref{fig:alg II}) 
  describes how to refine the initial estimate $\tau^{0}\in\I^{\Kt}$
using the projected Newton's algorithm.%
\begin{center}
\begin{figure}[H]
\begin{center}

\fbox{\begin{minipage}[t]{.9\columnwidth}%
\textbf{Algorithm II (local optimization)}

\vspace*{\bigskipamount}

\textbf{Input:}
\begin{itemize}
\item Cut-off frequency $f_{C}$ and measurement signal $y(\cdot)$, band-limited
to $\F=[-f_{C}:f_{C}]$, with the corresponding Fourier coefficients collected in $\widehat{y}\in\mathbb{C}^{N}$ (with $N=2f_{C}+1$). 
(See (\ref{eq:meas}) and (\ref{eq:meas freq}).)
\item A kernel $g_{\sigma_{2},N}(\cdot)$, band-limited to $\F$, with the corresponding Fourier coefficients collected  in $\widehat{g}_{\sigma_{2},N}\in\mathbb{C}^{N}$. (See
Criteria \ref{lem:corr decay away} and \ref{fact: slow decay in bound-1}.)
\item From Algorithm I, an initial estimate of the vector of locations $\tau^{0}\in\mathbb{I}^{\widetilde{K}}$,  and $\sigma_1\in(0,\frac{1}{2})$. 
\item A margin $0<\epsilon^0<\sigma_{1}$, and a termination threshold $\eta>0$.
\end{itemize}
\textbf{Output: }
\begin{itemize}
\item An estimate $\widetilde{\tau}\in\mathbb{I}^{\widetilde{K}}$
of the true vector of locations $\tau$.
\end{itemize}

\begin{enumerate}
\item Compute $\widehat{z}_{\sigma_{2}}=\widehat{g}_{\sigma_{2},N}\odot\widehat{y}$.
Here, $\odot$ stands for entry-wise (Hadamard) product.
\item Set $j=0$ and repeat:

\begin{enumerate}
\item Compute the gradient of $F(\cdot)$ at $\tau^{j}\in\I^{\widetilde{K}}$ 
(see (\ref{eq:grad of F})).
\item Calculate the reduced Hessian of $F(\cdot)$ at $\tau^{j}$, i.e., $\mathcal{R}_{\epsilon^j}(\frac{\partial^{2}F}{\partial\rho^{2}}(\tau^{j}))$. (See (\ref{eq:e active const}-\ref{eq:Hess exp}).)
 If the reduced Hessian is not a positive
definite matrix, exit with a failure message.
\item Calculate the descent direction $
\mathbb{R}^{\Kt}\ni 
 v^{j}:=\left(\mathcal{R}_{\epsilon^j}\left(\frac{\partial^{2}F}{\partial\rho^{2}}(\tau^{j})\right)\right)^{-1}\cdot\frac{\partial F}{\partial\rho}(\tau^{j})
$. 

\item For $\lambda\ge 0$, set $\tau^{j}(\lambda)=\mathcal{P}_{\mathbb{B}(\tau^{0},\sigma_{1})}\left(\tau^{j}\ominus \lambda\cdot v^{j}\right)$, where $\mathcal{P}_{\mathbb{B}(\tau^{0},\sigma_{1})}(\cdot)$
is the projection onto the ball $\mathbb{B}(\tau^{0},\sigma_{1})$ (see (\ref{eq:ball})). 

\item If $\|\tau^{j}(1)-\tau^j\|_2 \le \eta$, exit. Otherwise, pick $\epsilon^{j+1}=\min[\|\tau^{j}(1)-\tau^j\|_2,\sigma_1]$.   
\item \emph{Line search: }Find the least integer $m$ such that $
F\left(\tau^{j}(\lambda)\right)-F\left(\tau^{j}\right)\le\frac{-10^{-4}}{\lambda}\left\Vert \tau^{j}(\lambda)\ominus \tau^{j}\right\Vert _{2}^{2}
$, holds for $\lambda=2^{-m}$.
\item Set $\tau^{j+1}=\tau^{j}(2^{-m})$.
\item $j\leftarrow j+1$.
\end{enumerate}
\item Output $\widetilde{\tau}=\tau^{j}\in\mathbb{I}^{\widetilde{K}}$ as
the estimate of the true location vector $\tau$.

\end{enumerate}

\end{minipage}}
\end{center}
\caption{Algorithm II (local optimization)}
\label{fig:alg II}
\end{figure}
\end{center}

%

Under Criteria \ref{lem:corr decay away} and \ref{fact: slow decay in bound-1},
and when Proposition \ref{thm:algorithm 1} is in force, Algorithm II successfully
refines our estimate of the true position vector $\tau$. Convergence
of the projected Newton's algorithm 
to $\tau$ (or its small vicinity) is guaranteed by the next result, which
is proved in Section \ref{sec:Proof-of-Theorem grad descent}. We
remark that, while not the focus of this work, similar guarantees
hold for the gradient projection algorithm outlined in \eqref{eq:the  grad desc alg}.

\begin{thm}
\label{thm:grad des algo} 
\textbf{\emph{[Performance of Algorithm II]}} For integer $N$ and $0<c_{1}<c_{2}$
(both functions of $N$), let $\sigma_{1}=\frac{c_{1}}{N}$ and $\sigma_{2}=\frac{c_{2}}{N}$.
 Let $\tau^0\in \I^{\Kt}$ be the output of Algorithm I, and  suppose that Proposition \ref{thm:algorithm 1}  is in force, so that in particular $\Kt=K$. Suppose also that the kernel $g_{\sigma_{2},N}(\cdot)$
satisfies Criteria \ref{lem:corr decay away} and \ref{fact: slow decay in bound-1}
(with $\sigma=\sigma_{2}$). Lastly, assume that $2\sigma_{1}\le h(\sigma_{2},N)\le\sigma_{2}$
(see Criterion \ref{fact: slow decay in bound-1}).

Then, as long as
\[
\frac{\left\Vert n(\cdot)\right\Vert _{L_2}}{\|\alpha\|_{\infty}}=\frac{O(1)}{\mbox{dyn}\left(x_{\tau,\alpha}\right)^{2}}\le1,
\]
with a small enough constant, any limit point of Algorithm II is a stationary  point
$\widetilde{\tau}\in\mathbb{I}^{\Kt}$ with $\Kt=K$, and 
\begin{equation}
d(\widetilde{\tau},\tau)\le\min\left(O(1)\cdot\mbox{dyn}(x_{\tau,\alpha})^{2}\cdot\frac{\left\Vert n(\cdot)\right\Vert _{L_{2}}}{\left\Vert \alpha\right\Vert _{2}},2\sigma_{1}\right),
\label{eq:final dist bound}
\end{equation}
asymptotically as $c_{1},c_{2},N\rightarrow\infty$ and $c_1,c_{2}=\Theta(\log N)$ (with a large enough lower bound). Above, the metric and the dynamic range $\mbox{dyn}(x_{\tau,\alpha})$ were defined in (\ref{eq:Hausdoff dist}) and (\ref{eq:dyn range def}), respectively, and  $\|n(\cdot)\|_{L_{2}}$
is the energy of the additive noise (see (\ref{eq:main})).
\end{thm}

A few remarks are in order. 
\begin{rem}\textbf{[Noise-free] }
From (\ref{eq:final dist bound}), we observe that Phase II refines the output of Phase I when the dynamic range and noise level are both moderate. In particular, in the absence of noise, Phase II exactly identifies the correct support: $\widetilde{\tau}=\tau$. 
\end{rem}
\begin{rem}\textbf{[Separation]} \label{rem:separation}
For the two-phase algorithm to succeed (i.e., for Proposition \ref{thm:algorithm 1} and Theorem \ref{thm:grad des algo} to hold), the spike locations should be well-separated. In particular, for sufficiently large $f_C$, one needs 
\begin{equation}\label{eq:asymp sep rem}
\mbox{sep}(\tau) \ge 4\sigma_1= \Omega(1) \cdot \frac{\log f_C}{f_C}
\end{equation}
(as indicated in Proposition \ref{thm:algorithm 1}). 
 
 In contrast, super-resolution via  convex relaxation requires a  separation of $\Omega(1/f_C)$ \cite{Candes2014}. It is not clear whether the extra logarithmic factor in (\ref{eq:asymp sep rem}) is an artifact of the proofs of Proposition \ref{thm:algorithm 1} or Theorem \ref{thm:grad des algo}. We  also recently learned about similar rates (obtained with different techniques) in the context of edge detection from limited Fourier measurements   \cite{Cochran2013}. It appears that further work is needed to find  possible connections and to determine whether the required separation in (\ref{eq:asymp sep rem}) is optimal. 
\end{rem}

\begin{rem}\textbf{[Computational complexity] } As mentioned earlier, the two-phase algorithm for super-resolution is very fast, in part because fast and convenient  means for generating the kernels (namely, DPSWFs, which we recommend) exist, and partly because the search space in Phase II is $K$-dimensional where $K$ (the number of impulses) is often small (see Program (\ref{eq:regularized problem})). Also confer Section \ref{sec:prior art SUPER}. 
\end{rem}

\section{Prior Art}\label{sec:prior art SUPER}

By leveraging the \emph{sparsity} of the signal model in \eqref{eq:xtaualpha}, Cand\`{e}s et al.\ \cite{Candes2014} proposed a super-resolution algorithm that involves solving a convex program---a (typically expensive) SDP to be precise. In the absence of noise, this SDP  precisely recovers the sparse measure $x_{\tau,\alpha}$. More generally, the energy of the smoothed error signal scales with the noise level \cite{Candes2013}. Later, these results were translated into  bounds on the distance between the estimated and true impulse positions \cite{Fernandez-Granda2013}. We remark that  \cite{Candes2014} was followed by several good papers, including \cite{Tang2013,Duval2015,Azais2015,Fyhn2013,
Demanet2013,Liao2014}, that either proposed new super-resolution algorithms or  improved the computational complexity and performance of existing methods. 

But perhaps \cite{Fannjiang2012} is more relevant to the present work. There,  Fannjiang et al.\ modified the orthogonal matching pursuit algorithm to handle the highly coherent over-sampled DFT matrix. To improve the robustness of the algorithm, a local optimization step is skillfully  implemented in each step of their algorithm. This step refines one impulse position $\tau[i]$ at a time while keeping the rest of $\tau$  fixed. The present work differs from \cite{Fannjiang2012} in its use of prolate functions, and  in the depth of its theoretical guarantees. In particular, \cite{Fannjiang2012} does not seem to offer an analogue of Theorem \ref{thm:grad des algo}.

For the sake of demonstration, we compared our algorithm with those in \cite{Candes2014,Fannjiang2012}. Each $x_{\tau,\alpha}$ was generated  with  number of impulses $K=14$,\footnote{For a fair comparison, we assumed that $K$ is known in advance so as to match the setup of \cite{Fannjiang2012}.} uniformly random positions $\tau\in\I^K$, and amplitudes $\alpha\in\R^K$ drawn independently from zero-mean Gaussian distribution with variance $(2f_C+1)^{-1}$.  Additionally, we made sure that the impulse positions were  well-separated: $\mbox{sep}(\tau)\ge 2/f_C$ for every $x_{\tau,\alpha}$. The cut-off frequency was set to $f_C=50$, and we set $\sigma_1=\frac{3/2}{2f_C+1}$ and $\sigma_2=3\sigma_1/2$ in our algorithm.  
Additive low-pass Gaussian noise with energy  $(2f_C+1)\nu^2$ was then added to the observations. Figure  \ref{fig:super-res results} compares the (Hausdorff) distance of the estimated and true impulse positions for various values of $\nu$, and the run-times of the algorithms.

In about $9\%$ of the noise-free trials, the two-phase algorithm failed to exactly recover the impulse positions (but the error was still very small). In these trials, the initial estimate (output of Algorithm I) was not sufficiently close to the true impulse positions and, as a result, the local optimization phase (Algorithm II) converged to a local (as opposed to global) minimum. Recall that, according to Remark \ref{rem:separation}, the two-phase algorithm requires a separation of nearly $\log(f_C)/f_C$ to succeed (in contrast to the separation of $2/f_C$ is this experiment). 


\begin{center}
\begin{figure}[h]
\begin{center}
\includegraphics[scale=.5]{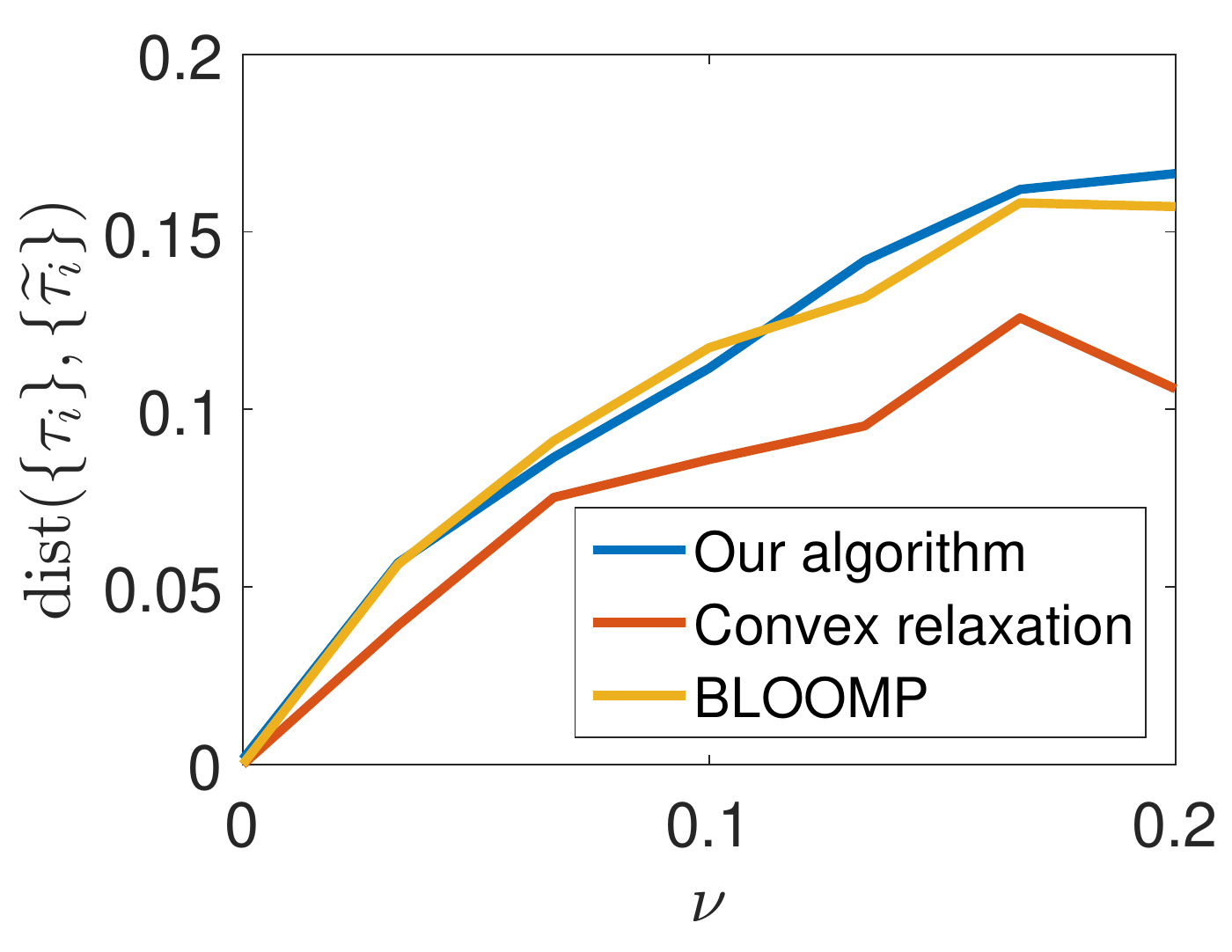}

\caption{Comparing our algorithm to the super-resolution algorithms in \cite{Candes2014,Fannjiang2012}: the horizontal axis reflects the noise level and the vertical axis displays the error, namely the distance between the estimated and true impulse positions. The average run-time for our algorithm, Fannjiang's, and Cand\`{e}s' were  $0.4$, $4$, and $8.3$ seconds, respectively on a laptop computer. (We made no attempts to optimize our code.)
}
\label{fig:super-res results}
\end{center}
\end{figure}
\end{center}

The super-resolution problem in this paper and the problem of line spectral estimation are closely related (once the time and frequency domains are exchanged) \cite{thomson1982spectrum,tang2015near}. We particularly recognize  Thomson's multitaper algorithm for spectral estimation \cite{thomson1982spectrum} due to its use of prolate functions and its popularity. In Thomson's algorithm, to lower the estimation bias, data is passed through multiple \emph{tapers}. The spectra of different channels are then  averaged (often with weights) to estimate the spectrum of the underlying random process (that generated the data). Because of their finite support, orthogonality, and negligible \emph{spectral leakage}, the Fourier series of the DPSWFs (also known as DPSSs) constitute an ideal choice for the tapers. Beyond these commonalities, our work is set apart from \cite{thomson1982spectrum} in its particular model (combination of impulses), different operating regimes (diminishing $\sigma$ here versus fixed $\sigma$ in \cite{thomson1982spectrum}), and 
the strong supporting theory provided here.

\section{Theory}

\subsection{Proof of Proposition \protect\ref{thm:algorithm 1} (Phase I)}\label{sec:proof of phase I}

Asymptotically (i.e., for large enough $N$), it holds that
\begin{equation}\label{eq:asymp sep}
\mbox{sep}\left(\tau\right)\ge4\sigma_{1}=\frac{4c_{1}}{N},
\end{equation}
because the entries of $\tau$ are assumed to be distinct. 
That is to say that $\tau$ is asymptotically well-separated for our
purposes here, as we see shortly. For $t\in\I$, we next observe that
\begin{align}
z_{\sigma_{1}}^{1}(t) & =(g_{\sigma_{1,N}}\circledast y)(t)\qquad \mbox{(see Algorithm I)}\nonumber \\
 & =(g_{\sigma_{1},N}\circledast x_{\tau,\alpha})(t)+(g_{\sigma_{1,N}}\circledast n)(t)\nonumber \\
 & =\sum_{i=1}^{K}\alpha[i]\cdot g_{\sigma_{1},N}(t\ominus\tau[i])+(g_{\sigma_{1},N}\circledast n)(t)
\qquad \mbox{(see \eqref{eq:xtaualpha})} \nonumber \\
 & =:\sum_{i=1}^{K}\alpha[i]\cdot g_{\sigma_{1},N}(t\ominus\tau[i])+n_{\sigma_{1}}(t).\label{eq:rep of chi1-1}
\end{align}
The second line above holds because, by assumption, $g_{\sigma_{1},N}(\cdot)$
too is band-limited to $\F=[-f_{C}:f_{C}]$. 

Under Criterion \ref{cri:(Sharp-decay-of}, the fast decay of the
kernel $g_{\sigma_{1},N}(\cdot)$ guarantees that $z_{\sigma_{1}}^{1}(t)$
is small when $t$ is away from the impulses and large otherwise.
Indeed, for $t\in\mathbb{I}$, whenever 
$$
\min_{i}d(t,\tau[i])\ge\sigma_{1},
$$
we argue as follows. Without loss of generality, let $\tau[1]$ be
the location of the closest impulse to $t$, $\tau[2]$ the second
closest impulse, and so on. Then the fact that $\mbox{sep}(\tau)\ge2\sigma_{1}$
(asymptotically) implies that
\begin{equation}\label{eq:separation in aI}
d(t,\tau[i]),\,\, d(t,\tau[i+1])\ge i\cdot\sigma_{1},\qquad i\in[1:K]\,\mbox{ and odd}.
\end{equation}
Then it follows from (\ref{eq:rep of chi1-1}) and Criterion \ref{cri:(Sharp-decay-of}
that
\begin{align}
\left|z_{\sigma_{1}}^{1}(t)\right| & \le\max_{i}\left|\alpha[i]\right|\cdot\sum_{i=1}^{K}\left|g_{\sigma_{1,N}}(t\ominus\tau[i])\right|+\left\Vert n_{\sigma_{1}}(\cdot)\right\Vert _{L_{\infty}} 
\nonumber \\
 & =\max_{i}\left|\alpha[i]\right|\cdot\frac{O(e^{-\Cr{decay}c_{1}})}{\sqrt{N}}\sum_{i=1}^{K}\frac{1}{\sin\left(\pi(t\ominus\tau[i])\right)}+\left\Vert n_{\sigma_{1}}(\cdot)\right\Vert _{L_{\infty}}\nonumber \\
 & \le\max_{i}\left|\alpha[i]\right|\cdot\frac{O(e^{-\Cr{decay}c_{1}})}{\sqrt{N}}\sum_{i=1}^{K}\frac{1}{\sin\left(\pi\cdot d(t,\tau[i])\right)}+\left\Vert n_{\sigma_{1}}(\cdot)\right\Vert _{L_{\infty}} \qquad \mbox{(see \eqref{eq:wrip around})} \nonumber \\
 & \le\max_{i}\left|\alpha[i]\right|\cdot\frac{O(e^{-\Cr{decay}c_{1}})}{\sqrt{N}}\sum_{0<i\sigma_{1}\le\frac{1}{2}}\frac{1}{\sin\left(\pi i\sigma_{1}\right)}+\left\Vert n_{\sigma_{1}}(\cdot)\right\Vert _{L_{\infty}},\qquad \mbox{(see \eqref{eq:separation in aI})}\label{eq:far away bound 1}
\end{align}
asymptotically. We can further simplify the bound above by asymptotically
controlling the summation in the last line as follows:
\begin{align}
& \sum_{0<i\sigma_{1}\le\frac{1}{2}}\frac{1}{\sin\left(\pi i\sigma_{1}\right)} \nonumber\\
& \le\frac{1}{\sin\left(\pi\sigma_{1}\right)}+\sigma_{1}^{-1}\int_{\sigma_{1}}^{\frac{1}{2}}\frac{1}{\sin\left(\pi t\right)}\, dt
\qquad 
\left( 
\sin(\pi t) \mbox{ is increasing on } [0,1/2]
\right)
\nonumber
\\
 & \le\frac{1}{\sin\left(\pi\sigma_{1}\right)}+\sigma_{1}^{-1}\sqrt{\frac{1}{2}-\sigma_{1}}\cdot\sqrt{\int_{\sigma_{1}}^{\frac{1}{2}}\frac{1}{\sin^{2}\left(\pi t\right)}\, dt}
\qquad \mbox{(Cauchy-Shwarz inequality)} 
\nonumber \\
 & \le\frac{1}{\sin\left(\pi\sigma_{1}\right)}+\frac{\sqrt{\cot(\pi\sigma_{1})}}{\sigma_{1}\sqrt{2\pi}}
\nonumber \\
 & =O\left(\left(\frac{N}{c_{1}}\right)^{\frac{3}{2}}\right).
 \qquad \left( \sigma_1 = \frac{c_1}{N},\,\, c_1=o(N)\right)\label{eq:sum to int}
\end{align}
Substituting the estimate above back into (\ref{eq:far away bound 1}),
we find that
\begin{align}
\left|z_{\sigma_{1}}^{1}(t)\right| 
& \le\max_{i}\left|\alpha[i]\right|\cdot\frac{O(e^{-\Cr{decay}c_{1}})}{\sqrt{N}}\sum_{0<i\sigma_{1}\le\frac{1}{2}}\frac{1}{\sin\left(\pi i\sigma_{1}\right)}+\left\Vert n_{\sigma_{1}}(\cdot)\right\Vert _{L_{\infty}}\nonumber\\
 & =\max_{i}\left|\alpha[i]\right|\cdot O(Nc_{1}^{-\frac{3}{2}}e^{-\Cr{decay}c_{1}})+\left\Vert n_{\sigma_{1}}(\cdot)\right\Vert _{L_{\infty}}\nonumber\\
 & \le \max_{i}\left|\alpha[i]\right|\cdot O(Ne^{-C c_{1}})+\left\Vert n_{\sigma_{1}}(\cdot)\right\Vert _{L_{\infty}}\qquad \left(c_1\rightarrow\infty \right)\nonumber \\
 & \le2\left\Vert n_{\sigma_{1}}(\cdot)\right\Vert _{L_{\infty}},\qquad \left( c_1 = \Theta(\log N)  \mbox{ with large enough lower bound}\right)
 \label{eq:far away bound 11}
\end{align}
asymptotically.
Let us simplify the noise term $\|n_{\sigma_{1}}(\cdot)\|_{L_{\infty}}$. Note
that
\begin{align}
& \left\Vert n_{\sigma_{1}}(\cdot)\right\Vert _{L_{\infty}} \nonumber\\
& =\left\Vert  \left( g_{\sigma_{1},N}\circledast n \right)(\cdot) \right\Vert _{L_{\infty}}\qquad \mbox{(see \eqref{eq:rep of chi1-1})} \nonumber\\
 & =\left\Vert \int_{\mathbb{I}}g_{\sigma_{1},N}(t')\cdot n(t-t')\, dt'\right\Vert _{L_{\infty}}\nonumber\\
 & \le\left\Vert g_{\sigma_{1},N}(\cdot)\right\Vert _{L_{1}}\left\Vert n(\cdot)\right\Vert _{L_{\infty}}\qquad 
 \mbox{(Holder inequality)} \nonumber\\
 & \le\left\Vert g_{\sigma_{1},N}(\cdot)\right\Vert _{L_{2}}\left\Vert n(\cdot)\right\Vert _{L_{\infty}}=\left\Vert n(\cdot)\right\Vert _{L_{\infty}}. \qquad \mbox{(Cauchy\textendash Schwarz inequality, and Criterion \ref{cri:(Sharp-decay-of})}
 \label{eq:two noise terms}
\end{align}
Overall, from \eqref{eq:far away bound 11}, we conclude
that
\begin{equation}
\left|z_{\sigma_{1}}^{1}(t)\right|
\le 2\left\| n_{\sigma_1}(\cdot) \right\|_{L_{\infty}}
\le2\left\Vert n(\cdot)\right\Vert _{L_{\infty}},\qquad
\mbox{if }\min_{i}d(t,\tau[i])\ge\sigma_{1},
\label{eq:upp bnd away}
\end{equation}
asymptotically. In words, $|z_{\sigma_{1}}^{1}(\cdot)|$ is small
away from the impulses.

At impulses, on the contrary, $|z_{\sigma_{1}}^{1}(\cdot)|$
remains large as we argue next. Without loss of generality, consider the first impulse positioned at $\tau[1]$.
 We observe that
 \begin{align*}
& \left|z_{\sigma_{1}}^{1}(\tau[1])\right| \nonumber\\
& =\left|\sum_{i=1}^{K}\alpha[i]\cdot g_{\sigma_{1},N}(\tau[1]\ominus\tau[i])+n_{\sigma_{1}}(\tau[1])\right|
\qquad \mbox{(see \eqref{eq:far away bound 1})}
\nonumber \\
 & \ge\left|\alpha[1] \right| \cdot \left|g_{\sigma_{1},N}(0)\right|-\max_{i}\left|\alpha[i]  \right|
\cdot \sum_{i=2}^{K} \left|g_{\sigma_{1},N}(\tau[1]\ominus\tau[i])\right|-\left\Vert n_{\sigma_{1}}(\cdot)\right\Vert _{L_{\infty}}\nonumber \\
 & =|\alpha[1]|\cdot\Omega\left(\sqrt{\frac{N}{c}}\right)-2\left\Vert n_{\sigma_{1}}(\cdot)\right\Vert _{L_{\infty}} 
 \qquad \mbox{(Criterion \ref{cri:(Sharp-decay-of}, similar to \eqref{eq:far away bound 1})}
 \nonumber \\
 & \ge\min_{j}|\alpha[j]|\cdot\Omega\left(\sqrt{\frac{N}{c}}\right)-2\left\Vert n(\cdot)\right\Vert_{L_{\infty}}.
\qquad \mbox{(see \eqref{eq:two noise terms})}
 \end{align*}
Next, we introduce the dynamic range of the signal (namely, $\mbox{dyn}(x_{\tau,\alpha})$) in order to simplify the expressions. More specifically, we continue by writing that 
\begin{align}
& \left|z_{\sigma_{1}}^{1}(\tau[1])\right| \nonumber\\
 & \ge \frac{\max_{j}|\alpha[j]|}{\mbox{dyn}(x_{\tau,\alpha})}\cdot \Omega\left(\sqrt{\frac{N}{c}}\right)-2\left\Vert n(\cdot)\right\Vert _{L_{\infty}}
\qquad \mbox{(see \eqref{eq:dyn range def})}
 \nonumber \\
 & \ge\frac{\max_{j}|\alpha[j]|}{\mbox{dyn}(x_{\tau,\alpha})}\cdot \Omega\left(\sqrt{\frac{N}{c}}\right)-2\left\Vert n(\cdot)\right\Vert _{L_{\infty}},
 \label{eq:lower bnd-1}
\end{align}
asymptotically.  
 In words, (\ref{eq:lower bnd-1})
states that $z_{\sigma_{1}}^{1}(\tau[i])$ is bounded away from zero (for every $i$).
Put differently, for large enough $N$, there exists a constant $\Cl{templo}>0$ 
such that
\begin{equation}\label{eq:lower bnd-11}
\left|z_{\sigma_{1}}^{1}(\tau[i])\right|\ge\Cr{templo}\cdot \frac{\|\alpha\|_\infty}{\mbox{dyn}(x_{\tau,\alpha})}\sqrt{\frac{N}{c}}-2\left\Vert n(\cdot)\right\Vert _{L_{\infty}},\qquad i\in[1:K].
\end{equation}
By comparing (\ref{eq:upp bnd away}) and (\ref{eq:lower bnd-11}), we observe that if
\begin{equation}
\frac{\|n(\cdot)\|_{L_{\infty}}}{\|\alpha\|_{\infty}}\le\frac{\Cr{templo}}{4}\cdot\frac{1}{\mbox{dyn}(x_{\tau,\alpha})}\cdot\sqrt{\frac{N}{c}},\label{eq:pre SNR}
\end{equation}
the lower bound is (\ref{eq:lower bnd-11}) does not exceed the upper
bound in (\ref{eq:upp bnd away}). All quantities $\|\alpha\|_{\infty}$,
$\mbox{dyn}(x_{\tau,\alpha})$, and $\|n(\cdot)\|_{L_{\infty}}$ are
independent of $c$ and $N$. Consequently, (\ref{eq:pre SNR}) is
met asymptotically (i.e., for large enough $N$). As a result, $\tau^{0}[1]$
(where $|z_{\sigma_{1}}^{1}(\cdot)|$ achieves its maximum on $\mathbb{I}$)
is within a radius $\sigma_{1}$ of the set $\tau$, i.e.,
\[
\min_{i}d(\tau^{0}[1],\tau[i])\le\sigma_{1}<\frac{1}{2}.
\]
Without loss of generality, suppose that $\tau[1]$ is the unique
entry of $\tau$ that achieves the minimum above, i.e.\ $d(\tau^{0}[1],\tau[1])=\min_{i}d(\tau^{0}[1],\tau[i])$.
Indeed, the uniqueness is guaranteed because $\tau$ is asymptotically well-separated (see \eqref{eq:asymp sep}). \textbf{ }Then, according to (\ref{eq:upp bnd away})\emph{,
}setting to zero a neighborhood of radius $2\sigma_{1}$ of $\tau^{0}[1]$
(to obtain $z_{\sigma_{1}}^{2}(\cdot))$ removes the \emph{bump }located
at $\tau[1]$. At the same time, since $\mbox{sep}(\tau)\ge4\sigma_{1}$
by \eqref{eq:asymp sep}, altering this neighborhood does not remove the bumps
located at $\tau[i]$, $i>1$. Therefore, $K$ repetitions of this
process recovers every member of $\tau$ to a precision of $\sigma_{1}$.
The algorithm terminates after $K$ iterations (so that $\widetilde{K}=K$)
because
\[
\|z_{\sigma_{1}}^{K+1}(\cdot)\|_{L_{\infty}}\le2\left\Vert n(\cdot)\right\Vert _{L_{\infty}}=\eta,
\]
asymptotically and according to (\ref{eq:upp bnd away}). In other
words, at this point, all the bumps have been removed and we have
reached the noise/interference level. This completes the proof of
Proposition \ref{thm:algorithm 1}.

%

\subsection{Proof of Theorem  \protect\ref{thm:grad des algo} (Phase II)\label{sec:Proof-of-Theorem grad descent}}

At this point, we begin to study the performance of Algorithm II.
\emph{Stationarity} is a necessary (first-order) condition for a feasible
point in $\mathbb{B}(\tau^{0},\sigma_{1})$ to be a local minimizer
of Program \eqref{eq:regularized problem}. In a constrained program,
a feasible point is stationary if the gradient of the objective function
makes an acute angle with every feasible direction. To be concrete,
we recall the definition of a stationary point \cite{Kelley1999} (slightly adjusted to match our settings).

\begin{defn}
\label{def: Stationary-points-of}\textbf{{[}Stationary point{]}}
\emph{In Program }\eqref{eq:regularized problem},\emph{ }$\rho_{s}\in\mathbb{B}(\tau^{0},\sigma_{1})$\emph{
is a stationary point if and only if
\[
\left\langle \frac{\partial F}{\partial\rho}(\rho_{s}),
\mbox{sign}\left( (\rho_s\ominus\rho) -\frac{1}{2}\right)
\right\rangle \ge0,\qquad\forall\rho\in\mathbb{B}(\tau^{0},\sigma_{1}).
\]
The entries of the sign vector above are $\{\mbox{sign}((\rho_s[i]\ominus 
\rho[i])-\frac{1}{2})\}$, $i\in[1:\Kt]$. 
}
\end{defn}
While not the focus of our analysis, one can establish that
the gradient projection algorithm outlined in  \eqref{eq:the grad desc alg}  (with appropriate step sizes $\{\delta^{j}\}$) always converges to a stationary
point of Program \eqref{eq:regularized problem}. (Also confer  {\cite[Theorem 5.4.6]{Kelley1999}.)%

Similarly, we prove next that the projected Newton's method  in Algorithm II converges
to a stationary point of Program \eqref{eq:regularized problem}. This
claim depends on the following result adapted from \cite[Theorem 5.5.2]{Kelley1999}.
\begin{prop}
\emph{\textbf{{[}Convergence to a stationary point{]} }} \label{prop:Convergence-to-a stationary point}Any
limit point of the sequence $\{\tau^{j}\}_{j}$ produced by Algorithm
II is a stationary point of Program \eqref{eq:regularized problem}
if
\begin{itemize}
\item the gradient   is Lipschitz continuous, i.e.,
$$
\left\| \frac{\partial F}{\partial \rho}(\rho_1)- \frac{\partial F}{\partial \rho}(\rho_2)\right\|_2 \le L \cdot d\left( 
\rho_1,\rho_2
\right),\qquad \forall\rho\in \mathbb{B}(\tau^0,\sigma_1), 
$$
for some finite $L$,
\item the Hessian is positive definite on the feasible set, i.e.,
$$
\frac{\partial^2 F}{\partial \rho^2}(\rho) \succ  0,\qquad \forall \rho\in \mathbb{B}(\tau^0,\sigma_1),
$$
\item both the spectral norm and the condition number of the Hessian are bounded on $\mathbb{B}(\tau^0,\sigma_1)$, and 
\item lastly, $0<\overline{\epsilon}\le \epsilon^j<\sigma_1$ for every $j$ and for some $\overline{\epsilon}$.
\end{itemize}
\end{prop}
By (\ref{eq:grad of F}), $\frac{\partial F}{\partial\rho}(\cdot)$
is continuous, and since $\mathbb{B}(\tau^{0},\sigma_{1})$ is compact,
$\frac{\partial F}{\partial\rho}(\cdot)$ is Lipschitz contiuous too.
In Appendix \ref{sec:Hessian is PD section}, we establish that $\frac{\partial^{2}F}{\partial\rho^{2}}(\cdot)$
is asymptotically positive definite on $\mathbb{B}(\tau^{0},\sigma_{1})$ (and moreover bounded from below by a positive factor of identity matrix) 
as long as
\[
\frac{\left\Vert n(\cdot)\right\Vert _{L_2}}{\|\alpha\|_{\infty}}=\frac{O(1)}{\mbox{dyn}\left(x_{\tau,\alpha}\right)^{2}},
\]
with a small enough constant. 
Then, since the eigenvalues of a matrix are continuous functions of
its entries, it follows that both spectral norm and condition number
of the Hessian are bounded on $\mathbb{B}(\tau^{0},\sigma_{1})$. The last item in Proposition \ref{prop:Convergence-to-a stationary point} holds by design (see Algorithm II).  
In summary, Proposition \ref{prop:Convergence-to-a stationary point}
is in force and any limit point of Algorithm II is a stationary point
of Program \eqref{eq:regularized problem}.

Upon existence, let $\widetilde{\tau}\in\I^{K}$ denote
one such  limit point which
, by Definition \ref{def: Stationary-points-of}, satisfies
\begin{equation}
\left\langle \frac{\partial F}{\partial\rho}(\widetilde{\tau}),
\mbox{sign}\left(
\left(\widetilde{\tau}\ominus \rho\right) - \frac{1}{2}
\right)
\right\rangle
\ge0,\qquad\forall\rho\in\mathbb{B}(\tau^{0},\sigma_{1}).\label{eq:cons of stationarity}
\end{equation}
 To control the distance of $\widetilde{\tau}$ from the true position  
vector $\tau$, we upper-bound the above inner product as follows.
See Appendix \ref{sec:Proof-of-Lemma grad is away-1} for the proof.
\begin{lem}
\label{lem:grad is away from zero}For integer $N$ and $0<c_{1}<c_{2}$
(both functions of $N$), let $\sigma_{1}=\frac{c_{1}}{N}$ and $\sigma_{2}=\frac{c_{2}}{N}$.
Suppose that the kernel $g_{\sigma_{2},N}(\cdot)$ satisfies Criteria
\ref{lem:corr decay away} and \ref{fact: slow decay in bound-1}
(with $\sigma=\sigma_{2}$). Suppose also that $2\sigma_{1}\le h(\sigma_{2},N)\le\sigma_{2}$ (see Criterion \ref{fact: slow decay in bound-1})%
. Lastly, define $F(\cdot)$ as in (\ref{eq:main program}),  and recall
the quantities involved there.

Then, for every $\rho\in\mathbb{B}(\tau^{0},\sigma_{1})$, it holds
asymptotically that
\begin{align*}
\left\langle \frac{\partial F}{\partial\rho}(\rho),
\mbox{sign}
\left(
(\rho \ominus \tau)- \frac{1}{2}
\right)
\right\rangle  & =-\frac{\Omega(N)}{\mbox{dyn}(x_{\tau,\alpha})^{2}}\cdot\|\alpha\|_{2}^{2} \cdot d(\rho,\tau)
+O(e^{-Cc_{2}})\cdot\|\alpha\|_{2}^{2}
\\
 & \qquad+O(N)\cdot\left\Vert n(\cdot)\right\Vert _{L_{2}}\left\Vert \alpha\right\Vert _{2}+O(N)\cdot\left\Vert n(\cdot)\right\Vert _{L_{2}}^{2},
\end{align*}
when $c_{1},c_{2},N\rightarrow\infty$ and $c_1,c_{2}=\Theta(\log N)$ (with a large enough lower bound). 
\end{lem}
We are now ready to complete the proof of Theorem \ref{thm:grad des algo}.
Since $\tau\in\mathbb{B}(\tau^{0},\sigma_{1})$ too, in light of (\ref{eq:cons of stationarity})
and Lemma \ref{lem:grad is away from zero}, we can write that
\begin{align*}
0 & \le\left\langle \frac{\partial F}{\partial\rho}(\widetilde{\tau}),
\mbox{sign}
\left(
\left(\widetilde{\tau}\ominus \tau \right) - \frac{1}{2}
\right)
\right\rangle 
= -\frac{\Omega(N)}{\mbox{dyn}(x_{\tau,\alpha})^{2}}\cdot\|\alpha\|_{2}^{2}\cdot d(\widetilde{\tau},\tau)+O(e^{-Cc_{2}})\cdot\|\alpha\|_{2}^{2}\\
 & \qquad \qquad \qquad \qquad \qquad \qquad\qquad\qquad+O(N)\cdot\left\Vert n(\cdot)\right\Vert _{L_{2}}\left\Vert \alpha\right\Vert _{2}+O(N)\cdot\left\Vert n(\cdot)\right\Vert _{L_{2}}^{2},
\end{align*}
which simplifies to
\begin{align*}
\frac{d(\widetilde{\tau},\tau)}{\mbox{dyn}(x_{\tau,\alpha})^{2}} & =O\left(e^{-Cc_{2}}+\frac{\left\Vert n(\cdot)\right\Vert _{L_{2}}}{\left\Vert \alpha\right\Vert _{2}}+\frac{\left\Vert n(\cdot)\right\Vert _{L_{2}}^{2}}{\left\Vert \alpha\right\Vert _{2}^{2}}\right)\\
 & =O\left(e^{-Cc_{2}}+\frac{\left\Vert n(\cdot)\right\Vert _{L_{2}}}{\left\Vert \alpha\right\Vert _{2}}\right),\qquad\mbox{if }\left\Vert n(\cdot)\right\Vert _{L_{2}}\le\|\alpha\|_{2},
\end{align*}
asymptotically. %
 This completes the proof of Theorem \ref{thm:grad des algo} since
we already know that
\[
\widetilde{\tau},\tau\in\mathbb{B}(\tau^{0},\sigma_{1})\Longrightarrow d(\widetilde{\tau},\tau)\le d(\widetilde{\tau},\tau^{0})+d(\tau^{0},\tau)\le2\sigma_{1},
\]
under Proposition \ref{thm:algorithm 1}.

\section*{Acknowledgments}
AE  acknowledges Ben Adcock, Aditya Viswanathan, and Anne Gelb for pointing out the possible connection between our work and \cite{Cochran2013}. Part of this research was conducted when AE was a graduate fellow at the  Statistical and Applied Mathematical Sciences Institute (SAMSI) and later a visitor at  the Institute for Computational and Experimental Research in Mathematics (ICERM). AE is grateful for their hospitality and kindness. 

\appendix

\bibliographystyle{plain}
\bibliography{References}

\begin{thebibliography}{10}

\bibitem{Azais2015}
J.M. Azais, Y.~De~Castro, and F.~Gamboa.
\newblock Spike detection from inaccurate samplings.
\newblock {\em Applied and Computational Harmonic Analysis}, 38(2):177--195,
  2015.

\bibitem{Candes2013}
E.J. Cand\`{e}s and C.~Fernandez-Granda.
\newblock Super-resolution from noisy data.
\newblock {\em Journal of Fourier Analysis and Applications}, 19(6):1229--1254,
  2013.

\bibitem{Candes2014}
E.J. Cand\`{e}s and C.~Fernandez-Granda.
\newblock Towards a mathematical theory of super-resolution.
\newblock {\em Communications on Pure and Applied Mathematics}, 67(6):906--956,
  2014.

\bibitem{Candes2015}
E.J. Cand\`{e}s, X.~Li, and M.~Soltanolkotabi.
\newblock Phase retrieval via {W}irtinger flow: {T}heory and algorithms.
\newblock {\em IEEE Transactions on Information Theory}, 61(4):1985--2007,
  2015.

\bibitem{Cochran2013}
D.~Cochran, A.~Gelb, and Y.~Wang.
\newblock Edge detection from truncated {F}ourier data using spectral
  mollifiers.
\newblock {\em Advances in computational mathematics}, pages 1--26, 2013.

\bibitem{Davenport2012}
M.A. Davenport and M.B. Wakin.
\newblock Compressive sensing of analog signals using discrete prolate
  spheroidal sequences.
\newblock {\em Applied and Computational Harmonic Analysis}, 33(3):438--472,
  2012.

\bibitem{Demanet2013}
L.~Demanet, D.~Needell, and N.~Nguyen.
\newblock Super-resolution via superset selection and pruning.
\newblock {\em arXiv preprint arXiv:1302.6288}, 2013.

\bibitem{Duval2015}
V.~Duval and G.~Peyre.
\newblock Exact support recovery for sparse spikes deconvolution.
\newblock {\em Foundations of Computational Mathematics}, pages 1--41, 2015.

\bibitem{Eftekhari2013}
A.~Eftekhari, J.~Romberg, and M.B. Wakin.
\newblock Matched filtering from limited frequency samples.
\newblock {\em IEEE Transactions on Information Theory}, 59(6):3475--3496,
  2013.

\bibitem{Eftekhari2013a}
A.~Eftekhari and M.B. Wakin.
\newblock Greed is super: {A} new iterative method for super-resolution.
\newblock In {\em IEEE Global Conference on Signal and Information Processing
  (GlobalSIP)}, 2013.

\bibitem{support}
A.~Eftekhari and M.B. Wakin.
\newblock Supplementary material for ``{G}reed is super: {A} fast algorithm for
  super-resolution''.
\newblock Technical report, Colorado School of Mines, 2015.

\bibitem{Fannjiang2012}
A.~Fannjiang and W.~Liao.
\newblock Coherence pattern-guided compressive sensing with unresolved grids.
\newblock {\em SIAM Journal on Imaging Sciences}, 5(1):179--202, 2012.

\bibitem{Fernandez-Granda2013}
C.~Fernandez-Granda.
\newblock Support detection in super-resolution.
\newblock {\em arXiv preprint arXiv:1302.3921}, 2013.

\bibitem{Fyhn2013}
K.~Fyhn, H.~Dadkhahi, and M.F. Duarte.
\newblock Spectral compressive sensing with polar interpolation.
\newblock In {\em Proceedings of the IEEE International Conference on
  Acoustics, Speech and Signal Processing (ICASSP)}, 2013.

\bibitem{Hogan2012}
J.A. Hogan and J.D. Lakey.
\newblock {\em Duration and bandwidth limiting: {P}rolate functions, sampling,
  and applications}.
\newblock Applied and Numerical Harmonic Analysis. Birkhauser Boston, 2012.

\bibitem{Kelley1999}
C.T. Kelley.
\newblock {\em Iterative methods for optimization}.
\newblock Frontiers in Applied Mathematics. Society for Industrial and Applied
  Mathematics, 1999.

\bibitem{Keshavan2009}
R.H. Keshavan and S.~Oh.
\newblock A gradient descent algorithm on the {G}rassman manifold for matrix
  completion.
\newblock {\em arXiv preprint arXiv:0910.5260}, 2009.

\bibitem{Knuth1976}
D.E. Knuth.
\newblock Big omicron and big omega and big theta.
\newblock {\em ACM Sigact News}, 8(2):18--24, 1976.

\bibitem{Liao2014}
W.~Liao and A.~Fannjiang.
\newblock {MUSIC} for single-snapshot spectral estimation: {S}tability and
  super-resolution.
\newblock {\em Applied and Computational Harmonic Analysis}, 2014.

\bibitem{Mallat1993}
S.~G Mallat and Z.~Zhang.
\newblock Matching pursuits with time-frequency dictionaries.
\newblock {\em IEEE Transactions on Signal Processing}, 41(12):3397--3415,
  1993.

\bibitem{Meister2009}
A.~Meister.
\newblock {\em Deconvolution problems in nonparametric statistics}.
\newblock Lecture Notes in Statistics. Springer Berlin Heidelberg, 2009.

\bibitem{Osipov2013}
A.~Osipov, V.~Rokhlin, and H.~Xiao.
\newblock {\em Prolate spheroidal wave functions of order zero: {M}athematical
  tools for bandlimited approximation}.
\newblock Applied Mathematical Sciences. Springer, 2013.

\bibitem{Slepian1965}
D.~Slepian.
\newblock Some asymptotic expansions for prolate spheroidal wave functions.
\newblock {\em Journal of Mathematical Physics}, 44(2):99--140, 1965.

\bibitem{Slepian1978}
D.~Slepian.
\newblock Prolate spheroidal wave functions, {F}ourier analysis and uncertainty
  {V}.
\newblock {\em Bell Systems Technical Journal}, 57(5):1371--1429, 1978.

\bibitem{tang2015near}
G.~Tang, B.N. Bhaskar, and B.~Recht.
\newblock Near minimax line spectral estimation.
\newblock {\em IEEE Transactions on Information Theory}, 61(1):499--512, 2015.

\bibitem{Tang2013}
G.~Tang, B.N. Bhaskar, P.~Shah, and B.~Recht.
\newblock Compressed sensing off the grid.
\newblock {\em IEEE Transactions on Information Theory}, 59(11):7465--7490,
  2013.

\bibitem{thomson1982spectrum}
D.J. Thomson.
\newblock Spectrum estimation and harmonic analysis.
\newblock {\em Proceedings of the IEEE}, 70(9):1055--1096, 1982.

\end{thebibliography}

\section{Toolbox}
\label{sec:Toolbox}
This section collects a number of results which are frequently 
invoked in the rest of the appendix.

In what follows, with integer $N$ and $c=c(N)>0$, we assume that
$\sigma=\frac{c}{N}$,  and consider a kernel $g_{\sigma,N}(\cdot)=g(\cdot;\sigma,N)$
that satisfies Criteria  \ref{lem:corr decay away}  
and \ref{fact: slow decay in bound-1}. 
\begin{lem}
\label{lem:sum of corr away}\textbf{ }For integer $N$ and $c=c(N)>0$,
let $\sigma=\frac{c}{N}$. Consider a kernel $g_{\sigma,N}(\cdot)$
that satisfies Criterion \ref{lem:corr decay away}. Fix $i\in[1:K]$, and $\rho\in\mathbb{I}^{K}$ with distinct entries. 
Then, it holds asymptotically that
\end{lem}
\begin{equation}
\sum_{j\in[1:K]\backslash\{i\}}\left|\left\langle g_{\sigma,N}(t\ominus\rho[i]),g_{\sigma,N}(t\ominus\rho[j])\right\rangle \right|=O(e^{-Cc}),\label{eq:fast decay}
\end{equation}
\begin{equation}
\sum_{j\in[1:K]\backslash\{i\}}\left|\left\langle g_{\sigma,N}(t\ominus\rho[i]),g_{\sigma,N}'(t\ominus\rho[j])\right\rangle \right|=O(e^{-Cc}),\label{eq:fast decay of derivates 0}
\end{equation}
\begin{equation}
\sum_{j\in[1:K]\backslash\{i\}}\left|\left\langle g'_{\sigma,N}(t\ominus\rho[i]),g_{\sigma,N}'(t\ominus\rho[j])\right\rangle \right|=O(e^{-Cc}),\label{eq:fast decay of derivates 0-1}
\end{equation}
when $c,N\rightarrow\infty$ and $c=\Theta(\log N)$ 
(with a large enough lower bound).
\begin{proof}
These inequalities are direct consequences of Criterion \ref{lem:corr decay away}.
Indeed, since the entries of $\rho$ are distinct, $\mbox{sep}(\rho)\ge2\sigma$ asymptotically (i.e., for large enough $N$). Then, to prove \eqref{eq:fast decay}, we  write that
\begin{align*}
& \sum_{j\in[1:K]\backslash\{i\}}\left|\left\langle g_{\sigma,N}(t\ominus\rho[i]),g_{\sigma,N}(t\ominus\rho[j])\right\rangle \right| \nonumber\\
& =O\left(\frac{e^{-\Cr{decay}c}}{N}\right)\sum_{j\in[1:K]\backslash\{i\}}\frac{1}{\left|\sin\left(\pi\cdot d(\rho[i],\rho[j])\right)\right|}
\qquad \mbox{(see Criterion \ref{lem:corr decay away})}
\\
 & =O\left(\frac{e^{-\Cr{decay}c}}{N}\right)\sum_{0<l\cdot2\sigma\le\frac{1}{2}}\frac{1}{\sin\left(\pi\cdot l\cdot2\sigma\right)}
\qquad \left( 
\mbox{sep}(\tau)\ge 2\sigma, \mbox{ asymptotically}
\right) 
 \\
 & =O\left(\frac{e^{-\Cr{decay}c}}{N}\right)\left(\frac{N}{c}\right)^{\frac{3}{2}} \qquad \mbox{(similar to \eqref{eq:sum to int})}\\
 & =O\left(N^{\frac{1}{2}}e^{-Cc}\right)
\qquad \left( c\rightarrow\infty \right)
 \\
 & =O\left(e^{-Cc}\right).
 \qquad \left( c = \Theta(\log N) \right)
\end{align*}
 The next inequalities in Lemma \ref{lem:sum of corr away}
are proved similarly and we omit the details here. 
\end{proof}
\begin{lem}
\label{lem:sum of cross corr}\textbf{ }For integer $N$ and $c=c(N)>0$,
let $\sigma=\frac{c}{N}$. Consider a kernel $g_{\sigma,N}(\cdot)$  
that satisfies Criterion \ref{lem:corr decay away}. Fix $i\in[1:K]$, and $\rho_{1},\rho_{2}\in\mathbb{I}^{K}$. Suppose that $\rho_{2}[j]\ne\rho_{1}[i]$,
for every $j\ne i$. %
{} Then, it holds asymptotically that
\begin{equation}
\sum_{j\in[1:K]\backslash\{i\}}\left|\left\langle g_{\sigma,N}(t\ominus\rho_{1}[i]),g_{\sigma,N}(t\ominus\rho_{2}[j])\right\rangle \right|=O(e^{-Cc}),\label{eq:fast mutual decay}
\end{equation}
\begin{equation}
\sum_{j\in[1:K]\backslash\{i\}}\left|\left\langle g_{\sigma,N}(t\ominus\rho_{1}[i]),g_{\sigma,N}'(t\ominus\rho_{2}[j])\right\rangle \right|=O(e^{-Cc}),\label{eq:fast decay of derivates}
\end{equation}
when $c,N\rightarrow\infty$ and $c=\Theta(\log N)$ (with a large enough lower bound). 
\end{lem}
\begin{proof}
Note that, by hypothesis, the vector formed from $\{\rho_{1}[i]\}\cup\{\rho_{2}[j]\}_{j\ne i}$
has distinct entries to which we can apply Lemma \ref{lem:sum of corr away}.
This completes the proof of Lemma \ref{lem:sum of cross corr}.
\end{proof}
A few more technical lemmas are in order.
In what follows, $G_{\rho}=G_{\rho,\sigma,N}\in\mathbb{C}^{N\times K}$
is defined similar to \eqref{eq:def of G_tau} for $\rho\in\mathbb{I}^{K}$,
$\sigma<\frac{1}{2}$, and integer $N$. %
Among other things, the next result states that the columns of $G_{\rho}$
are nearly orthonormal as long as $\rho$ is well-separated. %

\begin{lem}
\label{lem:props of Psi} For integer $N$ and $c=c(N)>0$, let $\sigma=\frac{c}{N}$. Consider a kernel $g_{\sigma,N}(\cdot)$ that
satisfies Criterion \ref{lem:corr decay away}.\textbf{ }Fix $\rho\in\I^{K}$
with distinct entries 
 and recall \eqref{eq:def of G_tau}. 
Then, it holds asymptotically that

\begin{equation}
\left\Vert I_{K}-G_{\rho}^{*}G_{\rho}\right\Vert =O(e^{-Cc}),\label{eq:props of Psi no 2}
\end{equation}
\begin{equation}
\left\Vert G_{\rho}\right\Vert \le1+O(e^{-Cc}),\label{eq:props of Psi no 1}
\end{equation}
\begin{equation}
\left\Vert \left(G_{\rho}^{*}G_{\rho}\right)^{-1}\right\Vert \le1+O(e^{-Cc}),\label{eq:props of Psi 3}
\end{equation}
\begin{equation}
\left\Vert G_{\rho}^{\dagger}\right\Vert \le1+O(e^{-Cc}),\label{eq:propes of Psi 3.1}
\end{equation}
\begin{equation}
\left\Vert G_{\rho}^{*}LG_{\rho}\right\Vert =O(e^{-Cc}),\label{eq:props of Psi no 5}
\end{equation}
\begin{equation}
\left\Vert \left\Vert g'_{\sigma,N}(\cdot)\right\Vert _{L_2}^{2}\cdot I_{K}-G_{\rho}^{*}L^{*}LG_{\rho}\right\Vert =O(e^{-Cc}),\label{eq:props of Psi no 5.1}
\end{equation}
as $c,N\rightarrow\infty$ and $c=\Theta(\log N)$ (with a sufficiently large lower bound). 
Above, $A^{\dagger}$ is the pseudo-inverse of $A$, and $\|A\|$ returns
its spectral norm of $A$. Entries of the diagonal matrix $L\in \mathbb{C}^{N\times N}$  are specified as $L[l,l]=\i 2\pi l$, $l\in\mathbb{F}$.  
{} The inverse of $G_{\rho}^{*}G_{\rho}$ exists, so that  (\ref{eq:props of Psi 3})
and (\ref{eq:propes of Psi 3.1}) are well-defined. %

Moreover, suppose that $\rho_{1},\rho_{2}\in\I^{K}$ satisfy $\rho_{1}[i]\ne\rho_{2}[j]$
for all $i\ne j$.%
{} Then, it holds asymptotically that
\begin{equation}
\left\Vert G_{\rho_{1}}^{*}G_{\rho_{2}}-M_{\rho_{1},\rho_{2}}\right\Vert =O(e^{-Cc}),\label{eq:props of Psi 4}
\end{equation}
where the entries of diagonal matrix $M_{\rho_{1},\rho_{2}}\in\mathbb{R}^{K\times K}$
 are specified as 
\begin{equation}
M_{\rho_{1},\rho_{2}}[i,i]=\left\langle g_{\sigma,N}(t\ominus\rho_{1}[i]),g_{\sigma,N}(t\ominus\rho_{2}[i])\right\rangle ,\qquad i\in[1:K].
\label{eq:def of M12}
\end{equation}
It also holds asymptotically that
\begin{equation}
\left\Vert G_{\rho_{1}}^{*}LG_{\rho_{2}}-M_{\rho_{1},\rho_{2}}^{d}\right\Vert =O(e^{-Cc}),\label{eq:props of Psi 6-1}
\end{equation}
where the diagonal matrix $M_{\rho_{1},\rho_{2}}^{d}\in\mathbb{R}^{K\times K}$
is defined with
\begin{equation}
M{}_{\rho_{1},\rho_{2}}^{d}[i,i]=\left\langle g_{\sigma,N}(t\ominus\rho_{1}[i]),g_{\sigma,N}'(t\ominus\rho_{2}[i])\right\rangle ,\qquad i\in[1:K].
\label{eq:def of M12d}
\end{equation}
In addition,
\begin{equation}
\left\Vert G_{\rho_{1}}-G_{\rho_{2}}\right\Vert \le\sqrt{2}\left\Vert I_{K}-M_{\rho_{1},\rho_{2}}\right\Vert ^{\frac{1}{2}}+O(e^{-Cc}),\label{eq:diff between Gs}
\end{equation}
\begin{equation}
\left\Vert G_{\rho_{1}}^{\dagger}-G_{\rho_{2}}^{\dagger}\right\Vert =O(1)\left\Vert I_{K}-M_{\rho_{1},\rho_{2}}\right\Vert ^{\frac{1}{2}}+O(e^{-Cc}).\label{eq:diff between Gdaggers}
\end{equation}
\end{lem}
\begin{proof}
Because $\|g_{\sigma,N}(\cdot)\|_{L_{2}}=1$ by Criterion \ref{lem:corr decay away}, the diagonal
entries of $G_{\rho}^{*}G_{\rho}$ equal to one, and in fact
\begin{equation}
\left(I_{K}-G_{\rho}^{*}G_{\rho}\right)[i,j]=\begin{cases}
0 & i=j,\\
-\left\langle g_{\sigma,N}(t\ominus\rho[i]),g_{\sigma,N}(t\ominus\rho[j])\right\rangle  & i\ne j,
\end{cases}\label{eq:Gres 1}
\end{equation}
where $I_{K}$ is the $K\times K$ identity matrix. Let $\lambda_{l}(A)$
return the $l$th eigenvalue of a square matrix $A$. Then, using the
Gershgorin disc theorem, we can write that
\begin{align}
\left\Vert I_{K}-G_{\rho}^{*}G_{\rho}\right\Vert  & =\max_{l\in[1:K]}\left|\lambda_{l}\left(I_{K}-G_{\rho}^{*}G_{\rho}\right)\right|
\nonumber
\\
 & \le\max_{i\in[1:K]}\sum_{j\ne i}\left|\left\langle g_{\sigma,N}(t\ominus\rho[i]),g_{\sigma,N}(t\ominus\rho[j])\right\rangle \right|
\qquad \mbox{(see \eqref{eq:Gres 1})} 
\nonumber
 \\
 & =O(e^{-Cc}).
 \qquad \mbox{(see (\ref{eq:fast decay}))}
 \label{eq:Gres disc 2}
\end{align}
 This establishes (\ref{eq:props of Psi no 2}). It also follows
that
\begin{align*}
\left\Vert G_{\rho}\right\Vert ^{2}-1 
& =\left\Vert G_{\rho}^{*}G_{\rho}\right\Vert -1\\
 & \le\left\Vert I_{K}-G_{\rho}^{*}G_{\rho}\right\Vert\nonumber\\
  & =O(e^{-Cc}),
\end{align*}
which implies (\ref{eq:props of Psi no 1}). 
Similarly, letting
$\sigma_{i}(A)$ be the $l$th singular value of a matrix $A$, we can write that
\begin{align*}
\min_{i\in[1:K]}\sigma_{i}\left(G_{\rho}^{*}G_{\rho}\right)-1 & =\min_{i\in[1:K]}\lambda_{i}\left(G_{\rho}^{*}G_{\rho}\right)-1\\
 & =-\max_{i\in[1:K]}\lambda_{i}\left(I_{K}-G_{\rho}^{*}G_{\rho}\right)\\
 & \ge-\max_{i\in[1:K]}\left|\lambda_{i}\left(I_{K}-G_{\rho}^{*}G_{\rho}\right)\right|\\
 & =-\left\Vert I_{K}-G_{\rho}^{*}G_{\rho}\right\Vert \\
 & =-O(e^{-Cc}).
 \qquad \mbox{(see \eqref{eq:props of Psi no 2})}
\end{align*}
It immediately follows that
\[
\left\Vert \left(G_{\rho}^{*}G_{\rho}\right)^{-1}\right\Vert 
= \frac{1}{\min_{i} \sigma_i\left( G_\rho^* G_\rho\right)}
\le\frac{1}{1-O(e^{-Cc})}=1+O(e^{-Cc}),
\]
as claimed in (\ref{eq:props of Psi 3}). Additionally,  (\ref{eq:propes of Psi 3.1})
follows directly from (\ref{eq:props of Psi no 1}) and (\ref{eq:props of Psi 3}).
 We next  observe that
\[
(G_{\rho}^{*}LG_{\rho})[i,j]=\begin{cases}
0 & i=j,\\
\left\langle g_{\sigma,N}(t\ominus\rho[i]),g_{\sigma,N}'(t\ominus\rho[j])\right\rangle  & i\ne j,
\end{cases}
\]
where we used the fact that $\left\langle g_{\sigma,N}(t\ominus\rho[i]),g_{\sigma,N}'(t\ominus\rho[i])\right\rangle =0$
because $g_{\sigma,N}(t)$ is symmetric about $t=\frac{1}{2}$ (and
hence $g_{\sigma,N}'(\cdot)$ is anti-symmetric about $t=\frac{1}{2}$).
Using the Gershgorin disc theorem once more, it follows that
\begin{align*}
\left\Vert G_{\rho}^{*}LG_{\rho}\right\Vert  & =\max_{l\in[1:K]}\left|\lambda_{l}\left(G_{\rho}^{*}LG_{\rho}\right)\right|\\
 & \le\max_{i\in[1:K]}\sum_{j\ne i}\left|(G_{\rho}^{*}LG_{\rho})[i,j]\right|\\
 & =\max_{i\in[1:K]}\sum_{j\ne i}\left|\left\langle g_{\sigma,N}(t\ominus\rho[i]),g_{\sigma,N}'(t\ominus\rho[j])\right\rangle \right|
\qquad \mbox{(see \eqref{eq:def of G_tau})}
 \\
 & =O(e^{-Cc}),
 \qquad \mbox{(see \eqref{eq:fast decay of derivates 0})}
\end{align*}
where the last line uses (\ref{eq:fast decay of derivates 0})%
. This establishes (\ref{eq:props of Psi no 5}). The proof of (\ref{eq:props of Psi no 5.1}) is similar to that of
(\ref{eq:props of Psi no 2}) and is omitted here. %

Next, by the definition of $M_{\rho_1,\rho_2}$ in \eqref{eq:def of M12}, it holds that
\[
\left(G_{\rho_{1}}^{*}G_{\rho_{2}}-M_{\rho_{1},\rho_{2}}\right)[i,j]=\begin{cases}
0 & i=j,\\
\left\langle g_{\sigma,N}(t\ominus\rho_{1}[i]),g_{\sigma,N}(t\ominus\rho_{2}[j])\right\rangle  & i\ne j.
\end{cases}
\]
We can therefore write that
\begin{align*}
 & \left\Vert G_{\rho_{1}}^{*}G_{\rho_{2}}-M_{\rho_{1},\rho_{2}}\right\Vert \\
 & \le\max\left[\left\Vert G_{\rho_{1}}^{*}G_{\rho_{2}}-M_{\rho_{1},\rho_{2}}\right\Vert _{1,1},\left\Vert G_{\rho_{1}}^{*}G_{\rho_{2}}-M_{\rho_{1},\rho_{2}}\right\Vert _{\infty,\infty}\right]
\quad 
 \left( 
 \|A\| \le \max\left[  \|A\|_{1,1}\,,\, \|A\|_{\infty,\infty} \right]
 \right)\\
 & =\max\Bigg[\max_{i\in[1:K]}\sum_{j\ne i}\left|\left\langle g_{\sigma,N}(t\ominus\rho_{1}[i]),g_{\sigma,N}(t\ominus\rho_{2}[j])\right\rangle \right|\\
& \qquad \qquad ,\max_{j\in[1:K]}\sum_{i\ne j}\left|\left\langle g_{\sigma,N}(t\ominus\rho_{1}[i]),g_{\sigma,N}(t\ominus\rho_{2}[j])\right\rangle \right|\Bigg]\\
& =O(e^{-Cc}),
 \qquad \mbox{(see \eqref{eq:fast mutual decay})}
\end{align*}
where $\|A\|_{1,1}$ and $\|A\|_{\infty,\infty}$ are $\ell_{1}\rightarrow\ell_{1}$
and $\ell_{\infty}\rightarrow\ell_{\infty}$ operator norms of matrix
$A$. 
This proves (\ref{eq:props of Psi 4}). Similarly, recalling \eqref{eq:def of M12d}, we note that
\[
\left(G_{\rho_{1}}^{*}LG_{\rho_{2}}-M{}_{\rho_{1},\rho_{2}}^{d}\right)[i,j]=\begin{cases}
0 & i=j,\\
\left\langle g_{\sigma,N}(t\ominus\rho_{1}[i]),g'_{\sigma,N}(t\ominus\rho_{2}[j])\right\rangle  & i\ne j,
\end{cases}
\]
from which it follows that 
\begin{align*}
 & \left\Vert G_{\rho_{1}}^{*}LG_{\rho_{2}}-M{}_{\rho_{1},\rho_{2}}^{d}\right\Vert \\
 & \le\max\left[\left\Vert G_{\rho_{1}}^{*}LG_{\rho_{2}}-M_{\rho_{1},\rho_{2}}^{d}\right\Vert _{1,1},\left\Vert G_{\rho_{1}}^{*}LG_{\rho_{2}}-M_{\rho_{1},\rho_{2}}^{d}\right\Vert _{\infty,\infty}\right]
\,\, \left(\|A\|\le \max\left[ \|A\|_{1,1} \,,\,
\|A\|_{\infty,\infty}
 \right]
 \right) 
 \\
 & =\max\Big[\max_{i\in[1:K]}\sum_{j\ne i}\left|\left\langle g_{\sigma,N}(t\ominus\rho_{1}[i]),g'_{\sigma,N}(t\ominus\rho_{2}[j])\right\rangle \right| \nonumber\\
& \qquad\quad\quad , \max_{j\in[1:K]}\sum_{i\ne j}\left|\left\langle g_{\sigma,N}(t\ominus\rho_{1}[i]),g'_{\sigma,N}(t\ominus\rho_{2}[j])\right\rangle \right|\Big]\\
 & =O(e^{-Cc}).
 \qquad 
 \mbox{(see \eqref{eq:fast decay of derivates})}
\end{align*}
This establishes (\ref{eq:props of Psi 6-1}). To prove (\ref{eq:diff between Gs}),
we note that
\begin{align*}
& \left\Vert G_{\rho_{1}}-G_{\rho_{2}}\right\Vert ^{2} 
\\
& =\left\Vert \left(G_{\rho_{1}}-G_{\rho_{2}}\right)^{*}\left(G_{\rho_{1}}-G_{\rho_{2}}\right)\right\Vert \\
 & =\left\Vert G_{\rho_{1}}^{*}G_{\rho_{1}}+G_{\rho_{2}}^{*}G_{\rho_{2}}-G_{\rho_{1}}^{*}G_{\rho_{2}}-G_{\rho_{2}}^{*}G_{\rho_{1}} \right\Vert \\
 & =\Big\|2I_{K}-2M_{\rho_{1},\rho_{2}}-\left(I_{K}-G_{\rho_{1}}^{*}G_{\rho_{1}}\right)-\left(I_{K}-G_{\rho_{2}}^{*}G_{\rho_{2}}\right)\\
 & \qquad \qquad \qquad 
-\left(G_{\rho_{1}}^{*}G_{\rho_{2}}-M_{\rho_{1},\rho_{2}}\right)-\left(G_{\rho_{2}}^{*}G_{\rho_{1}} 
-M_{\rho_{1},\rho_{2}}\right) \Big\|
 \\
 & \le 2\left\Vert I_{K}-M_{\rho_{1},\rho_{2}}\right\Vert +\left\Vert I_{K}-G_{\rho_{1}}^{*}G_{\rho_{1}}\right\Vert
+ \left\Vert I_{K}-G_{\rho_{2}}^{*}G_{\rho_{2}}\right\Vert \\
& \qquad \qquad \qquad   
  +2\left\Vert G_{\rho_{1}}^{*}G_{\rho_{2}}-M_{\rho_{1},\rho_{2}}\right\Vert
\quad \mbox{(see \eqref{eq:def of M12})}  
   \\
 & =2\left\Vert I_{K}-M_{\rho_{1},\rho_{2}}\right\Vert +O(e^{-Cc}),\qquad\mbox{(see (\ref{eq:props of Psi no 2}) and (\ref{eq:props of Psi 4}))}.
\end{align*}
Lastly, to prove (\ref{eq:diff between Gdaggers}), we write that
\begin{align*}
 & \left\Vert G_{\rho_{1}}^{\dagger}-G_{\rho_{2}}^{\dagger}\right\Vert \\
 & =\left\Vert \overset{A}{\overbrace{\left(G_{\rho_{1}}^{*} G_{\rho_{1}}\right)^{-1}}}
\cdot  
 \overset{B}{\overbrace{G_{\rho_{1}}}}-\overset{C}{\overbrace{\left(G_{\rho_{2}}^{*}G_{\rho_{2}}\right)^{-1}}}\cdot \overset{D}{\overbrace{G_{\rho_{2}}}}\right\Vert \\
 & \le\left\Vert \left[\left(G_{\rho_{1}}^{*}G_{\rho_{1}}\right)^{-1}-\left(G_{\rho_{2}}^{*}G_{\rho_{2}}\right)^{-1}\right]G_{\rho_{1}}\right\Vert +\left\Vert \left(G_{\rho_{2}}^{*}G_{\rho_{2}}\right)^{-1}\left[G_{\rho_{1}}-G_{\rho_{2}}\right]\right\Vert \\
 & \le\left\Vert \left(G_{\rho_{1}}^{*}G_{\rho_{1}}\right)^{-1}-\left(G_{\rho_{2}}^{*}G_{\rho_{2}}\right)^{-1}\right\Vert \cdot \left\Vert G_{\rho_{1}}\right\Vert +\left\Vert \left(G_{\rho_{2}}^{*}G_{\rho_{2}}\right)^{-1}\right\Vert \cdot  \left\Vert G_{\rho_{1}}-G_{\rho_{2}}\right\Vert \\
 & \le\left\Vert \left(G_{\rho_{1}}^{*}G_{\rho_{1}}\right)^{-1}\right\Vert \cdot\left\Vert G_{\rho_{1}}^{*}G_{\rho_{1}}-G_{\rho_{2}}^{*}G_{\rho_{2}}\right\Vert \cdot\left\Vert \left(G_{\rho_{2}}^{*}G_{\rho_{2}}\right)^{-1}\right\Vert
 \cdot \left\| G_{\rho_1}\right\|  +\left\Vert \left(G_{\rho_{2}}^{*}G_{\rho_{2}}\right)^{-1}\right\Vert \\
& \qquad \qquad\qquad   \cdot\left\Vert G_{\rho_{1}}-G_{\rho_{2}}\right\Vert \\
 & =O(1)\left\Vert \overset{A}{\overbrace{G_{\rho_{1}}^{*}}}\cdot \overset{B}{\overbrace{G_{\rho_{1}}}}-\overset{C}{\overbrace{G_{\rho_{2}}^{*}}}\cdot \overset{D}{\overbrace{G_{\rho_{2}}}}\right\Vert +O(1)\left\Vert G_{\rho_{1}}-G_{\rho_{2}}\right\Vert \qquad\mbox{(see  \eqref{eq:props of Psi no 1} and  (\ref{eq:props of Psi 3}))}\\
 & \le O(1)\left\Vert G_{\rho_{1}}-G_{\rho_{2}}\right\Vert \cdot\max\left[\left\Vert G_{\rho_{1}}\right\Vert ,\left\Vert G_{\rho_{2}}\right\Vert \right]+O(1)\left\Vert G_{\rho_{1}}-G_{\rho_{2}}\right\Vert \\
 & =O(1)\left\Vert G_{\rho_{1}}-G_{\rho_{2}}\right\Vert +O(1)\left\Vert G_{\rho_{1}}-G_{\rho_{2}}\right\Vert \qquad\mbox{(see (\ref{eq:props of Psi no 1}))}\\
 & =O(1)\left(\sqrt{2\left\Vert I_{K}-M_{\rho_{1},\rho_{2}}\right\Vert }+O(e^{-Cc})\right)\qquad\mbox{(see (\ref{eq:diff between Gs}))}\\
 & =O(1)\left\Vert I_{K}-M_{\rho_{1},\rho_{2}}\right\Vert ^{\frac{1}{2}}+O(e^{-Cc}).
\end{align*}
Above, we twice used the identity $AB-CD=(A-C)B+C(B-D)$
for conformal matrices $A,B,C,D$. The fifth line owes itself to the
identity $A^{-1}-B^{-1}=A^{-1}\left(B-A\right)B^{-1}$ for (conformal
and invertible) matrices $A,B$. 
{}

This concludes the proof of Lemma \ref{lem:props of Psi}.
\end{proof}
If the entries of $\rho\in\I^K$ are distinct, $G_{\rho}\in\mathbb{C}^{N\times K}$ has nearly orthonormal
columns asymptotically (by Lemma \ref{lem:props of Psi}),   and it holds that $\|G_{\rho}\beta\|_{2}\approx\|\beta\|_{2}$
for any $\beta\in\mathbb{C}^{K}$. This is recorded next.
\begin{lem}
\label{lem: alpha and Psi alpha}\textbf{ }For integer $N$ and $c=c(N)>0$,
let $\sigma=\frac{c}{N}$. Consider a kernel $g_{\sigma,N}(\cdot)$ 
that satisfies Criterion \ref{lem:corr decay away}. Fix $\rho\in\mathbb{I}^{K}$ with distinct entries, and $\beta\in\mathbb{R}^{K}$, and recall (\ref{eq:def of G_tau}). 
{} It then holds asymptotically that
\[
\left(1-O(e^{-Cc})\right)\|\beta\|_{2}^{2}\le\left\Vert G_{\rho}\beta\right\Vert _{2}^{2}\le\left(1+O(e^{-Cc})\right)\|\beta\|_{2}^{2},
\]
\[
\left(1-O(e^{-Cc})\right)\|\beta\|_{2}^{2}\le
\frac{\left\Vert L G_{\rho}\beta\right\Vert _{2}^{2}}
{\left\| g'_{\sigma,N}(\cdot) \right\|^2_{L_2}}
\le\left(1+O(e^{-Cc})\right)\|\beta\|_{2}^{2},
\]
 when $c,N\rightarrow\infty$ and $c=\Theta(\log N)$ (with a large enough lower bound).
  \end{lem}
\begin{proof}
This is a direct consequence of Lemma \ref{lem:props of Psi}. Indeed, 
it holds asymptotically that
\begin{align*}
\left|\left\Vert G_{\rho}\beta\right\Vert _{2}^{2}-\|\beta\|_{2}^{2}\right| & =\left|\beta^{*}\left(G_{\rho}^{*}G_{\rho}-I_{K}\right)\beta\right|\\
 & \le\left\Vert G_{\rho}^{*}G_{\rho}-I_{K}\right\Vert \|\beta\|_{2}^{2}\\
 & =O(e^{-Cc})\|\beta\|_{2}^{2}.
 \qquad \mbox{(see \eqref{eq:props of Psi no 2})}
\end{align*}
The second claim is proved similarly using \eqref{eq:props of Psi no 5.1}.  This completes the proof of Lemma \ref{lem: alpha and Psi alpha}.
\end{proof}
We close this section with the following auxiliary result that approximates certain 
projection matrices with simpler quantities.
\begin{lem}
\label{lem:P is simple}\textbf{ }For integer $N$ and $c=c(N)>0$,
let $\sigma=\frac{c}{N}$. Consider a kernel $g_{\sigma,N}(\cdot)$ 
that satisfies Criterion \ref{lem:corr decay away}. Consider a vector
$u\in\mathbb{R}^{K}$ and let $U=\mbox{diag}(u)\in\mathbb{R}^{K\times K}$
be the diagonal matrix formed from $u$. Suppose that $\rho\in\mathbb{I}^{K}$
has distinct entries and set%
{} $\mathcal{P}_{\rho,U}:=G_{\rho}UG_{\rho}^{\dagger}\in\mathbb{C}^{N\times N}$ (after recalling (\ref{eq:def of G_tau})).\footnote{In particular, when $U=I_{K}$, $\mathcal{P}_{\rho,I_{K}}$ is the
orthogonal projection onto $\mbox{span}(G_{\rho})$.} 
Then, it holds
asymptotically that
\begin{equation}
\left\Vert \mathcal{P}_{\rho,U}-G_{\rho}UG_{\rho}^{*}\right\Vert =O(e^{-Cc})\|u\|_{\infty},\label{eq:P is simple}
\end{equation}
\begin{equation}
\left\Vert \mathcal{P}_{\rho,U}\right\Vert \le2\|u\|_{\infty},\label{eq:P is simple 2}
\end{equation}
when $c,N\rightarrow\infty$ and $c=\Theta(\log N)$ (with a large enough lower bound). 

Furthermore, suppose that $\rho_{1},\rho_{2}\in\I^{K}$ both have
distinct entries and $\rho_{1}[i]\ne\rho_{2}[j]$ when $i\ne j$.
Then, for any $\beta\in\mathbb{R}^{K}$, we asymptotically have that
\begin{equation}
\left\Vert \mathcal{P}_{\rho_{2},U}G_{\rho_{1}}\beta-G_{\rho_{2}}UM_{\rho_{1},\rho_{2}}\beta\right\Vert _{2}=O(e^{-Cc})\|u\|_{\infty}\|G_{\rho_{1}}\beta\|_{2},\label{eq:P is simple 3}
\end{equation}
with $M_{\rho_{1},\rho_{2}}\in\mathbb{R}^{K\times K}$ defined as 
in (\ref{eq:def of M12}).
\end{lem}
\begin{proof}
We show that $\mathcal{P}_{\rho,U}$ can be well approximated with
$G_{\rho}UG_{\rho}^{*}$, and do so by bounding $\|\mathcal{P}_{\rho,U}-G_{\rho}UG_{\rho}^{*}\|$
next. We use the fact that $G_{\rho}\in\mathbb{C}^{N\times K}$ has
nearly orthonormal columns (thanks to the distinct entries of $\rho$).
Asymptotically, it holds that
\begin{align}
\left\Vert \mathcal{P}_{\rho,U}-G_{\rho}UG_{\rho}^{*}\right\Vert  & =\left\Vert G_{\rho}U\left(G_{\rho}^{*}G_{\rho}\right)^{-1}G_{\rho}^{*}-G_{\rho}UG_{\rho}^{*}\right\Vert \nonumber \\
 & \le\left\Vert G_{\rho}\right\Vert \|U\|\left\Vert \left(G_{\rho}^{*}G_{\rho}\right)^{-1}-I_{K}\right\Vert \left\Vert G_{\rho}\right\Vert \nonumber \\
 & =\left\Vert G_{\rho}\right\Vert ^{2}\|u\|_{\infty}\left\Vert \left(G_{\rho}^{*}G_{\rho}\right)^{-1}-I_{K}\right\Vert
\qquad 
\left(  U = \mbox{diag}(u) \right)  \nonumber \\
 & \le\left\Vert G_{\rho}\right\Vert ^{2}\|u\|_{\infty}\left\Vert \left(G_{\rho}^{*}G_{\rho}\right)^{-1}\right\Vert \left\Vert I_{K}-G_{\rho}^{*}G_{\rho}\right\Vert \nonumber \\
 & =O(e^{-Cc})\|u\|_{\infty}.
 \qquad \mbox{(see Lemma \ref{lem:props of Psi})}
 \label{eq:part 1 of lemma}
\end{align}
This
proves (\ref{eq:P is simple}). Also, (\ref{eq:P is simple 2}) is proved by noting that 
\begin{align*}
\left\Vert \mathcal{P}_{\rho,U}\right\Vert 
& \le\left\Vert G_{\rho}UG_{\rho}^{*}\right\Vert +\left\Vert \mathcal{P}_{\rho,U}-G_{\rho}UG_{\rho}^{*}\right\Vert \\
& \le\left\Vert G_{\rho}\right\Vert ^{2}\|u\|_{\infty}+O(e^{-Cc})\|u\|_{\infty}
\qquad \mbox{(see \eqref{eq:part 1 of lemma})}
\\
& \le2\|u\|_{\infty},
\qquad \mbox{(see \eqref{eq:props of Psi no 1})}
\end{align*}
asymptotically. Lastly, using the just-established (\ref{eq:P is simple})
and (\ref{eq:P is simple 2}), we prove (\ref{eq:P is simple 3})
as follows. (We will use the triangle inequality and basic manipulations, and also the fact that the spectral norm of a diagonal matrix equals its maximum entry.)   Asymptotically, it holds that 
\begin{align*}
 & \left\Vert \mathcal{P}_{\rho_{2},U}G_{\rho_{1}}\beta-G_{\rho_{2}}UM_{\rho_{1},\rho_{2}}\beta\right\Vert _{2}\\
& \le\left\Vert \mathcal{P}_{\rho_{2},U}G_{\rho_{1}}\beta-G_{\rho_{2}}UG_{\rho_{2}}^{*}G_{\rho_{1}}\beta\right\Vert _{2}+\left\Vert G_{\rho_{2}}UG_{\rho_{2}}^{*}G_{\rho_{1}}\beta-G_{\rho_{2}}UM_{\rho_{1},\rho_{2}}\beta\right\Vert _{2}\\
 & \le\left\Vert \mathcal{P}_{\rho_{2},U}-G_{\rho_{2}}UG_{\rho_{2}}^{*}\right\Vert \|G_{\rho_{1}}\beta\|_{2}+\|G_{\rho_{2}}\|\|U\|\left\Vert G_{\rho_{1}}^{*}G_{\rho_{2}}-M_{\rho_{1},\rho_{2}}\right\Vert \|\beta\|_{2}\\
 & =O(e^{-Cc})\|u\|_{\infty}\|G_{\rho_{1}}\beta\|_{2},
\end{align*}
where we also used Lemmas \ref{lem:props of Psi} and \ref{lem: alpha and Psi alpha}.
This proves (\ref{eq:P is simple 3}) and completes the proof of Lemma
\ref{lem:P is simple}.\end{proof}
\begin{lem}
\label{lem:ls dist to true}
For integer $N$ and $c=c(N)>0$, let $\sigma=\frac{c}{N}$. Consider a kernel $g_{\sigma,N}(\cdot)$ that satisfies Criterion \ref{lem:corr decay away}. 
Suppose that $\rho_{1},\rho_{2}\in\mathbb{I}^{K}$
satisfy $\rho_{1}[i]\ne\rho_{2}[j]$ for all $i\ne j$. Recall (\ref{eq:def of G_tau}), and for vectors $\alpha\in\mathbb{R}^{K}$ and $\widehat{n}\in\mathbb{C}^{N}$, set\footnote{Dependence of $\beta_{\rho_2}$ on other parameters (particularly, $\rho_1$) is suppressed for convenience. }
$$
\beta_{\rho_{2}}=G_{\rho_{2}}^{\dagger}(G_{\rho_{1}}\alpha+\widehat{n}).
$$
Then, it holds asymptotically
that
\[
\left\Vert \beta_{\rho_{2}}-\alpha\right\Vert _{\infty}=O(1)\cdot\left(\sqrt{K}\left\Vert I_{K}-M_{\rho_{1},\rho_{2}}\right\Vert ^{\frac{1}{2}}+e^{-Cc}\right)\|\alpha\|_{\infty}+O(1)\cdot\left\Vert \widehat{n}\right\Vert _{2},
\]
\[
\left\Vert \beta_{\rho_{2}}-\beta_{\rho_{1}}\right\Vert _{\infty}=O(1)\cdot\left(\left\Vert I_{K}-M_{\rho_{1},\rho_{2}}\right\Vert ^{\frac{1}{2}}+e^{-Cc}\right)\cdot\left(\sqrt{K}\|\alpha\|_{\infty}+\left\Vert \widehat{n}\right\Vert _{2}\right),
\]
when $c,N\rightarrow\infty$ and $c=\Theta(\log N)$ (with a large enough lower bound). 
\end{lem}
\begin{proof}
Note that
\begin{align*}
 & \left\Vert \beta_{\rho_{2}}-\alpha\right\Vert _{\infty}\\
 & \le\left\Vert G_{\rho_{2}}^{\dagger}G_{\rho_{1}}\alpha-\alpha\right\Vert _{\infty}+\left\Vert G_{\rho_{2}}^{\dagger}\widehat{n}\right\Vert _{\infty}
\qquad \left( \beta_{\rho_2}= G_{\rho_2}^{\dagger}\left(  
G_{\rho_1}\alpha+\widehat{n}
\right)  \right) 
 \\
 & =\left\Vert \left(G_{\rho_{2}}^{\dagger}-G_{\rho_{1}}^{\dagger}\right)G_{\rho_{1}}\alpha\right\Vert _{\infty}+\left\Vert G_{\rho_{2}}^{\dagger}\widehat{n}\right\Vert _{\infty}\\
 & \le\left\Vert \left(G_{\rho_{2}}^{\dagger}-G_{\rho_{1}}^{\dagger}\right)G_{\rho_{1}}\right\Vert _{\infty\rightarrow\infty}\|\alpha\|_{\infty}+\left\Vert G_{\rho_{2}}^{\dagger}\widehat{n}\right\Vert _{2}\\
 & \le\sqrt{K}\left\Vert \left(G_{\rho_{2}}^{\dagger}-G_{\rho_{1}}^{\dagger}\right)G_{\rho_{1}}\right\Vert \cdot\|\alpha\|_{\infty}+\left\Vert G_{\rho_{2}}^{\dagger}\widehat{n}\right\Vert _{2}\qquad\left(\|A\|_{\infty\rightarrow\infty}\le\sqrt{K}\cdot\|A\|,\quad A\in\mathbb{C}^{N\times K}\right)\\
 & \le\sqrt{K}\left\Vert G_{\rho_{2}}^{\dagger}-G_{\rho_{1}}^{\dagger}\right\Vert \left\Vert G_{\rho_{1}}\right\Vert \cdot\|\alpha\|_{\infty}+\left\Vert G_{\rho_{1}}^{\dagger}\right\Vert \left\Vert \widehat{n}\right\Vert _{2}\\
  & =O(1)\left(\sqrt{K}\left\Vert I_{K}-M_{\rho_{1},\rho_{2}}\right\Vert ^{\frac{1}{2}}+e^{-Cc}\right)\cdot\|\alpha\|_{\infty}\\
& \qquad \qquad  +O(1)\left\Vert \widehat{n}\right\Vert _{2}.\qquad \mbox{(Lemma \ref{lem:props of Psi} and }c=\Theta(\log N)\mbox{)}.
\end{align*}
The last line above requires the lower bound in $c=\Theta(\log N)$ to be sufficiently large. 
Similarly,
\begin{align*}
& \left\Vert \beta_{\rho_{2}}-\beta_{\rho_{1}}\right\Vert _{\infty}\\
 & \le\left\Vert \left(G_{\rho_{2}}^{\dagger}-G_{\rho_{1}}^{\dagger}\right)G_{\rho_{1}}\alpha\right\Vert _{\infty}+\left\Vert \left(G_{\rho_{2}}^{\dagger}-G_{\rho_{1}}^{\dagger}\right)\widehat{n}\right\Vert _{\infty}
\qquad \left(
 \beta_{\rho_1}= G_{\rho_1}^{\dagger}\left(  
G_{\rho_1}\alpha+\widehat{n} \right) 
\right) 
 \\
 & \le\left\Vert \left(G_{\rho_{2}}^{\dagger}-G_{\rho_{1}}^{\dagger}\right)G_{\rho_{1}}\right\Vert _{\infty\rightarrow\infty}\|\alpha\|_{\infty}+\left\Vert G_{\rho_{2}}^{\dagger}-G_{\rho_{1}}^{\dagger}\right\Vert \cdot\left\Vert \widehat{n}\right\Vert _{2}\qquad\left(\|a\|_{\infty}\le\|a\|_{2},\qquad\forall a\in\mathbb{C}^{N}\right)\\
 & \le\sqrt{K}\left\Vert \left(G_{\rho_{2}}^{\dagger}-G_{\rho_{1}}^{\dagger}\right)G_{\rho_{1}}\right\Vert \cdot\|\alpha\|_{\infty}+\left\Vert G_{\rho_{2}}^{\dagger}-G_{\rho_{1}}^{\dagger}\right\Vert \cdot\left\Vert \widehat{n}\right\Vert _{2}\\
 & \le\sqrt{K}\left\Vert G_{\rho_{2}}^{\dagger}-G_{\rho_{1}}^{\dagger}\right\Vert \left\Vert G_{\rho_{1}}\right\Vert \cdot\|\alpha\|_{\infty}+\left\Vert G_{\rho_{2}}^{\dagger}-G_{\rho_{1}}^{\dagger}\right\Vert \cdot\left\Vert \widehat{n}\right\Vert _{2}\\
 & =O(1)\cdot\left\Vert G_{\rho_{2}}^{\dagger}-G_{\rho_{1}}^{\dagger}\right\Vert \cdot\left(\sqrt{K}\|\alpha\|_{\infty}+\left\Vert \widehat{n}\right\Vert _{2}\right)\qquad\mbox{(Lemma \ref{lem:props of Psi})}\\
 & =O(1)\cdot\left(\left\Vert I_{K}-M_{\rho_{1},\rho_{2}}\right\Vert ^{\frac{1}{2}}+e^{-Cc}\right)\cdot\left(\sqrt{K}\|\alpha\|_{\infty}+\left\Vert \widehat{n}\right\Vert _{2}\right).\qquad\mbox{(Lemma \ref{lem:props of Psi})}.
\end{align*}
This completes the proof of Lemma \ref{lem:ls dist to true}.
\end{proof}

\section{Hessian of $F(\cdot)$ is Positive Definite on $\mathbb{B}(\tau^{0},\sigma_{1})$ \label{sec:Hessian is PD section}}

Throughout, assume that Proposition \ref{thm:algorithm 1} is in force, so that $\tau^0\in\I^{\Kt}$ satisfies both $\Kt=K$ and $d(\tau^0,\tau)\le \sigma_1$. In this section, with $F(\cdot)$ defined as in \eqref{eq:main program}, we will establish that $\frac{\partial^{2}F}{\partial\rho^{2}}(\cdot)$
is asymptotically positive definite in the small neighborhood of $\tau$, namely $\mathbb{B}(\tau^0,\sigma_1)$ (see \eqref{eq:ball}). 
(This will prove necessary for the projected Newton algorithm to converge
to a local minimizer of $F(\cdot)$.)  To do so, we first show that $\frac{\partial^{2}F}{\partial\rho^{2}}(\tau)\succ0$
asymptotically, and next control the variation of the Hessian under
small changes of its argument.

We assume that the entries of $\tau\in\I^{K}$ are distinct. Then, for 
$\rho\in\mathbb{B}\left(\tau^{0},\sigma_{1}\right)$, the entries
of $\rho$ too are distinct asymptotically (i.e., for large enough
$N$). Moreover, $\rho[i]\ne\tau[j]$ for $i\ne j$ (asymptotically). Therefore, we
are in position  to apply the technical lemmas in the Toolbox (Appendix \ref{sec:Toolbox}).

\subsection{Establishing $\frac{\partial^{2}F}{\partial\rho^{2}}(\tau)\succ0$ }

From (\ref{eq:def of z hat}), recall that $\widehat{z}_{\sigma_{2}}=G_{\tau}\alpha+\widehat{n}_{\sigma_{2}}\in\mathbb{C}^{N}$ contains the Fourier coefficients of
the (possibly noisy) measurement signal. Recall also
 the orthogonal projection onto $\mbox{span}(G_{\tau})$,
namely $\mathcal{P}_{\tau}\in\mathbb{C}^{N\times N}$. Then, clearly,
\[
\left(I_{N}-\mathcal{P}_{\tau}\right)\widehat{z}_{\sigma_{2}}=\left(I_{N}-\mathcal{P}_{\tau}\right)G_{\tau}\alpha+\left(I_{N}-\mathcal{P}_{\tau}\right)\widehat{n}_{\sigma_{2}}=\left(I_{N}-\mathcal{P}_{\tau}\right)\widehat{n}_{\sigma_{2}}.
\]
With this in mind and using (\ref{eq:Hess exp}), we rewrite the expression
for Hessian at $\tau$ as
\begin{align*}
\frac{\partial^{2}F}{\partial\rho^{2}}(\tau) & =2\cdot\mbox{diag}\left(\beta_{\tau}\right)\cdot G_{\tau}^{*}LL^{*}G_{\tau}\cdot\mbox{diag}\left(\beta_{\tau}\right)\qquad\quad\left(L^{2}=-LL^{*},\quad\beta_{\tau}=\widetilde{\beta}(\tau)=G_{\tau}^{\dagger}\cdot\widehat{z}_{\sigma_{2}}\right)\\
 & \qquad-2\cdot\mbox{diag}\left(\beta_{\tau}\right)\cdot\mbox{diag}\left(G_{\tau}^{*}L^{2}\left(I_{N}-\mathcal{P}_{\tau}\right)\widehat{n}_{\sigma_{2}}\right)\\
 & \qquad-2\left[\mbox{diag}\left(\beta_{\tau}\right)G_{\tau}^{*}LG_{\tau}-\mbox{diag}\left(G_{\tau}^{*}L\left(I_{N}-\mathcal{P}_{\tau}\right)\widehat{n}_{\sigma_{2}}\right)\right]\\
 & \qquad\cdot\left(G_{\tau}^{*}G_{\tau}\right)^{-1}\cdot\left[G_{\tau}^{*}L^{*}G_{\tau}\cdot\mbox{diag}(\beta_{\tau})-\mbox{diag}\left(G_{\tau}^{*}L\left(I_{N}-\mathcal{P}_{\tau}\right)\widehat{n}_{\sigma_{2}}\right)\right].
\end{align*}
After rearranging the expression above, we find   that
\begin{align}
\frac{\partial^{2}F}{\partial\rho^{2}}(\tau)=\mbox{signal}_{\tau}+\mbox{noise}_{\tau},\label{eq:Hess psd argue}
\end{align}
\begin{equation}
\mbox{signal}_{\tau}:=2\cdot\mbox{diag}\left(\beta_{\tau}\right)\cdot G_{\tau}^{*}L\left(I_{N}-\mathcal{P}_{\tau}\right)L^{*}G_{\tau}\cdot\mbox{diag}\left(\beta_{\tau}\right),\label{eq:signal part}
\end{equation}
\begin{align}
\mbox{noise}_{\tau} & :=-2\cdot\mbox{diag}\left(\beta_{\tau}\right)\cdot\mbox{diag}\left(G_{\tau}^{*}L^{2}\left(I_{N}-\mathcal{P}_{\tau}\right)\widehat{n}_{\sigma_{2}}\right)\nonumber \\
 & \qquad+2\cdot\mbox{diag}\left(G_{\tau}^{*}L\left(I_{N}-\mathcal{P}_{\tau}\right)\widehat{n}_{\sigma_{2}}\right)\cdot G_{\tau}^{\dagger}L^{*}G_{\tau}\cdot\mbox{diag}\left(\beta_{\tau}\right)\nonumber \\
 & \qquad+2\cdot\mbox{diag}\left(\beta_{\tau}\right)\cdot\left(G_{\tau}^{\dagger}L^{*}G_{\tau}\right)^{*}\cdot\mbox{diag}\left(G_{\tau}^{*}L\left(I_{N}-\mathcal{P}_{\tau}\right)\widehat{n}_{\sigma_{2}}\right)\nonumber \\
 & \qquad-2\cdot\mbox{diag}\left(G_{\tau}^{*}L\left(I_{N}-\mathcal{P}_{\tau}\right)\widehat{n}_{\sigma_{2}}\right)\cdot\left(G_{\tau}^{*}G_{\tau}\right)^{-1}\cdot\mbox{diag}\left(G_{\tau}^{*}L\left(I_{N}-\mathcal{P}_{\tau}\right)\widehat{n}_{\sigma_{2}}\right).\label{eq:all noise terms}
\end{align}
As detailed presently, the signal term above is  ``strongly'' positive
definite because $L^{*}G_{\tau}$ %
(associated with the translated copies of $g'_{\sigma_{2},N}(\cdot)$,
the derivative of our kernel) is nearly orthogonal to $\mbox{span}(G_{\tau})$.
Therefore, as long as the noise term is negligible, we have $\frac{\partial^{2}F}{\partial\rho^{2}}(\tau) \succ 0$. Let us consider the details now.

We first control all four terms in $\mbox{noise}_{\tau}$. For the
first term in (\ref{eq:all noise terms}), it holds asymptotically
that
\begin{align*}
 & \left\Vert \mbox{diag}\left(\beta_{\tau}\right)\cdot\mbox{diag}\left(G_{\tau}^{*}L^{2}\left(I_{N}-\mathcal{P}_{\tau}\right)\widehat{n}_{\sigma_{2}}\right)\right\Vert \\
 & \le\left\Vert \beta_{\tau}\right\Vert _{\infty}\cdot\left\Vert G_{\tau}^{*}L^{2}\left(I_{N}-\mathcal{P}_{\tau}\right)\widehat{n}_{\sigma_{2}}\right\Vert _{2}
\qquad \left( \|a\|_\infty \le \|a\|_2,\, \forall a \right) 
 \\
 & \le\left\Vert \beta_{\tau}\right\Vert _{\infty}\cdot\left\Vert G_{\tau}\right\Vert \|L\|^{2}\left\Vert I_{N}-\mathcal{P}_{\tau}\right\Vert \cdot\left\Vert \widehat{n}_{\sigma_{2}}\right\Vert _{2}\\
 & =\left\Vert \beta_{\tau}\right\Vert _{\infty}\cdot O(N^{2})\cdot\left\Vert \widehat{n}_{\sigma_{2}}\right\Vert _{2}, & \mbox{(Lemma \ref{lem:props of Psi}, }\|L\|\le2\pi N\mbox{)}
\end{align*}
as $c,N\rightarrow\infty$, $c=\Theta(\log N)$ (with a large enough lower bound). Similarly, for the second term in \eqref{eq:all noise terms},  
it is true asymptotically that
\begin{align*}
 & \left\Vert \mbox{diag}\left(G_{\tau}^{*}L\left(I_{N}-\mathcal{P}_{\tau}\right)\widehat{n}_{\sigma_{2}}\right)\cdot G_{\tau}^{\dagger}L^{*}G_{\tau}\cdot\mbox{diag}\left(\beta_{\tau}\right)\right\Vert \\
 & \le\left\Vert G_{\tau}^{*}L\left(I_{N}-\mathcal{P}_{\tau}\right)\widehat{n}_{\sigma_{2}}\right\Vert _{2}\cdot\left\Vert G_{\tau}^{\dagger}\right\Vert \|L\|\left\Vert G_{\tau}\right\Vert \cdot\left\Vert \beta_{\tau}\right\Vert _{\infty}
\qquad \left( \|a\|_\infty\le \|a\|_2,\,\forall a \right) 
 \\
 & \le\left\Vert G_{\tau}^{*}L\left(I_{N}-\mathcal{P}_{\tau}\right)\widehat{n}_{\sigma_{2}}\right\Vert _{2}\cdot O(N)\cdot\left\Vert \beta_{\tau}\right\Vert _{\infty}\qquad\mbox{(Lemma \ref{lem:props of Psi}, }\|L\|\le2\pi N\mbox{)}\\
 & \le\|G_{\tau}\|\|L\|\left\Vert I_{N}-\mathcal{P}_{\tau}\right\Vert \cdot\left\Vert \widehat{n}_{\sigma_{2}}\right\Vert _{2}\cdot O(N)\cdot\left\Vert \beta_{\tau}\right\Vert _{\infty}\\
 & \le O(N)\cdot\left\Vert \widehat{n}_{\sigma_{2}}\right\Vert _{2}\cdot O(N)\cdot\left\Vert \beta_{\tau}\right\Vert _{\infty}\qquad\mbox{(Lemma \ref{lem:props of Psi}, }\|L\|\le2\pi N\mbox{)}\\
 & =O(N^{2})\cdot\left\Vert \beta_{\tau}\right\Vert _{\infty}\left\Vert \widehat{n}_{\sigma_{2}}\right\Vert _{2}.
\end{align*}
An identical bound holds the third noise term. As for the last term
in (\ref{eq:all noise terms}), we asymptotically have that
\[
\left\Vert \mbox{diag}\left(G_{\tau}^{*}L\left(I_{N}-\mathcal{P}_{\tau}\right)\widehat{n}_{\sigma_{2}}\right)\left(G_{\tau}^{*}G_{\tau}\right)^{-1}\mbox{diag}\left(G_{\tau}^{*}L\left(I_{N}-\mathcal{P}_{\tau}\right)\widehat{n}_{\sigma_{2}}\right)\right\Vert =O(N^{2})\cdot\left\Vert \widehat{n}_{\sigma_{2}}\right\Vert _{2}^{2},
\]
where we invoked Lemma \ref{lem:props of Psi} again. Overall, using
the triangle inequality, we obtain that
\begin{equation}
\left\Vert \mbox{noise}_{\tau}\right\Vert =O(N^{2})\cdot\left(\left\Vert \beta_{\tau}\right\Vert _{\infty}\left\Vert \widehat{n}_{\sigma_{2}}\right\Vert _{2}+\left\Vert \widehat{n}_{\sigma_{2}}\right\Vert _{2}^{2}\right).
\qquad \mbox{(see \eqref{eq:all noise terms})}
\label{eq:noise term bnd pre}
\end{equation}
To eliminate $\beta_{\tau}$ from the expression above, we apply Lemma
\ref{lem:ls dist to true} (with $\rho_{1}=\rho_{2}=\tau$) to obtain
that
\begin{align}
\left\Vert \beta_{\tau}\right\Vert _{\infty} & \le\|\alpha\|_{\infty}+\left\Vert \beta_{\tau}-\alpha\right\Vert _{\infty}\nonumber \\
 & \le2 \|\alpha\|_{\infty}+O(1)\left\Vert \widehat{n}_{\sigma_{2}}\right\Vert _{2}.\qquad\left(M_{\tau,\tau}=I_{K}\right)\label{eq:some bounds on beta}
\end{align}
Therefore,
\begin{align*}
\left\Vert \mbox{noise}_{\tau}\right\Vert  & =O(N^{2})\cdot\left(\left\Vert \beta_{\tau}\right\Vert _{\infty}\left\Vert \widehat{n}_{\sigma_{2}}\right\Vert _{2}+\left\Vert \widehat{n}_{\sigma_{2}}\right\Vert _{2}^{2}\right),
\qquad \mbox{(see \eqref{eq:noise term bnd pre})}\\
 & =O(N^{2})\left(\left(\|\alpha\|_{\infty}+\left\Vert \widehat{n}_{\sigma_{2}}\right\Vert _{2}\right)\left\Vert \widehat{n}_{\sigma_{2}}\right\Vert _{2}+\left\Vert \widehat{n}_{\sigma_{2}}\right\Vert _{2}^{2}\right)
\qquad \mbox{(see \eqref{eq:some bounds on beta})} 
 \\
 & =O(N^{2})\left(\|\alpha\|_{\infty}\left\Vert \widehat{n}_{\sigma_{2}}\right\Vert _{2}+\left\Vert \widehat{n}_{\sigma_{2}}\right\Vert _{2}^{2}\right),
\end{align*}
and, therefore,
\begin{equation}
\mbox{noise}_{\tau}\preccurlyeq O(N^{2})\left(\left\Vert \alpha\right\Vert _{\infty}\left\Vert \widehat{n}_{\sigma_{2}}\right\Vert _{2}+\left\Vert \widehat{n}_{\sigma_{2}}\right\Vert _{2}^{2}\right)\cdot I_{K},\label{eq:noise bound}
\end{equation}
both valid asymptotically. Next, we establish that the signal term
in (\ref{eq:Hess psd argue}) is a positive definite matrix. For arbitrary
$v\in\mathbb{R}^{K}$, it holds asymptotically that
\begin{align*}
& v^{*}\left(G_{\tau}^{*}L\left(I_{N}-\mathcal{P}_{\tau}\right)L^{*}G_{\tau}\right)v \\
& =v^{*}\left(G_{\tau}^{*}LL^{*}G_{\tau}\right)v-v^{*}\left(G_{\tau}^{*}L\mathcal{P}_{\tau}L^{*}G_{\tau}\right)v\\
 & =\left\Vert L^{*}G_{\tau}v\right\Vert _{2}^{2}-\left\Vert \mathcal{P}_{\tau}L^{*}G_{\tau}v\right\Vert _{2}^{2}\qquad\qquad(\mathcal{P}_{\tau}^{2}=\mathcal{P}_{\tau})\\
 & =\left\Vert LG_{\tau}v\right\Vert _{2}^{2}-\left\Vert \left(G_{\tau}^{\dagger}\right)^{*}G_{\tau}^{*}L^{*}G_{\tau}v\right\Vert _{2}^{2}\qquad\left(L^{*}=-L,\quad\mathcal{P}_{\tau}=\left(G_{\tau}^{\dagger}\right)^{*}G_{\tau}^{*}\right)\\
 & \ge\left\Vert LG_{\tau}v\right\Vert _{2}^{2}-\left\Vert G_{\tau}^{\dagger}\right\Vert ^{2}\cdot\left\Vert G_{\tau}^{*}L^{*}G_{\tau}v\right\Vert _{2}^{2}\\
 & =\left\Vert LG_{\tau}v\right\Vert _{2}^{2}-O(1)\cdot\left\Vert G_{\tau}^{*}L^{*}G_{\tau}\right\Vert^{2}\cdot \|v\|_2^2\qquad\mbox{(see \eqref{eq:propes of Psi 3.1})}\\
 & \ge\left(\left\Vert g'_{\sigma,N}(\cdot)\right\Vert _{2}^{2}-O(e^{-Cc_2})\right)\cdot\|v\|_{2}^{2}-O(1)\cdot O(e^{-Cc_2})\cdot\|v\|_{2}^{2}\qquad\mbox{(Lemmas \ref{lem:props of Psi} and \ref{lem: alpha and Psi alpha})}\\
 & =\Omega(N^{2})\cdot\|v\|_{2}^{2},\qquad\mbox{(Criterion \ref{fact: slow decay in bound-1})}
\end{align*}
as $c_2,N\rightarrow\infty$  and $c_2=\Theta(\log N$) (with large enough lower bound). Since
the choice of $v$ was arbitrary, we conclude that
\[
G_{\tau}^{*}L\left(I_{N}-\mathcal{P}_{\tau}\right)L^{*}G_{\tau}\succcurlyeq\Omega(N^{2})\cdot I_{K},
\]
asymptotically. From (\ref{eq:signal part}), it follows that
\begin{equation}
\mbox{signal}_{\tau}=2\cdot\mbox{diag}\left(\beta_{\tau}\right)\cdot G_{\tau}^{*}L\left(I_{N}-\mathcal{P}_{\tau}\right)L^{*}G_{\tau}\cdot\mbox{diag}\left(\beta_{\tau}\right)\succcurlyeq\Omega(N^{2})\cdot\mbox{diag}\left(\beta_{\tau}\right)^{2}.\label{eq:mid step 110}
\end{equation}
We remove $\beta_{\tau}$ from the right hand side above by invoking
Lemma \ref{lem:ls dist to true}: Note that 
\begin{align*}
\mbox{diag}\left(\beta_{\tau}\right) & =\mbox{diag}\left(\alpha\right)-\mbox{diag}\left(\alpha-\beta_{\tau}\right)\\
 & \succcurlyeq\mbox{diag}\left(\alpha\right)-\left\Vert \alpha-\beta_{\tau}\right\Vert _{\infty}\cdot I_{K},\\
 & \succcurlyeq
 \frac{1}{2}\mbox{diag}\left(\alpha\right)-O(1)\cdot\left\Vert \widehat{n}_{\sigma_{2}}\right\Vert _{2}\cdot I_{K},
\end{align*}
\begin{align*}
\mbox{diag}\left(\beta_{\tau}\right)^{2} & \succcurlyeq\frac{1}{8}\mbox{diag}\left(\alpha\right)^{2}-O(1)\cdot\left\Vert \widehat{n}_{\sigma_{2}}\right\Vert _{2}^{2}\cdot I_{K},\qquad\left(\left(a-b\right)^{2}\ge\frac{a^{2}}{2}-b^{2},\quad\forall a,b\in\mathbb{R}\right),
\end{align*}
asymptotically. Therefore, revisiting (\ref{eq:mid step 110}), we can
write that
\begin{align}
\mbox{signal}_{\tau} & \succcurlyeq\Omega(N^{2})\cdot\mbox{diag}\left(\beta_{\tau}\right)^{2}\nonumber \\
 & \succcurlyeq\Omega(N^{2})\cdot\left(\mbox{diag}\left(\alpha\right)^{2}-O(1)\cdot\left\Vert \widehat{n}_{\sigma_{2}}\right\Vert _{2}^{2}\cdot I_{K}\right).\label{eq:mid step 111}
\end{align}
Suppose that
\begin{equation}
\left\Vert \widehat{n}_{\sigma_{2}}\right\Vert _{2}\le
O(1)\cdot \|\alpha\|_{\infty},\label{eq:NSR const}
\end{equation}
with a small enough constant. 
Then, combining (\ref{eq:mid step 111})
with (\ref{eq:noise bound}) yields
\begin{align}
 & \frac{\partial^{2}F}{\partial\rho^{2}}(\tau)\nonumber \\
 & =\mbox{signal}_{\tau}+\mbox{noise}_{\tau}\qquad\mbox{(see (\ref{eq:Hess psd argue}))}\nonumber \\
 & \succcurlyeq\Omega(N^{2})\cdot\left(\mbox{diag}\left(\alpha\right)^{2}-O(1)\cdot\left\Vert \widehat{n}_{\sigma_{2}}\right\Vert _{2}^{2}\cdot I_{K}\right)-O(N^{2})\cdot\left(\left\Vert \alpha\right\Vert _{\infty}\left\Vert \widehat{n}_{\sigma_{2}}\right\Vert _{2}+\left\Vert \widehat{n}_{\sigma_{2}}\right\Vert _{2}^{2}\right)\cdot I_{K}\nonumber \\
 & \succcurlyeq\Omega(N^{2})\cdot\mbox{diag}\left(\alpha\right)^{2}-O(N^{2})\cdot\left\Vert \alpha\right\Vert _{\infty}\left\Vert \widehat{n}_{\sigma_{2}}\right\Vert _{2}\cdot I_{K}\qquad\mbox{(see (\ref{eq:NSR const}))}\nonumber \\
 & =\Omega(N^{2})\cdot\|\alpha\|_{\infty}^{2}\cdot\frac{\mbox{diag}\left(\alpha\right)^{2}}{\|\alpha\|_{\infty}^{2}}-O(N^{2})\cdot\left\Vert \alpha\right\Vert _{\infty}\left\Vert \widehat{n}_{\sigma_{2}}\right\Vert _{2}\cdot I_{K}\nonumber \\
 & \succcurlyeq\Omega(N^{2})\cdot\|\alpha\|_{\infty}^{2}\cdot\frac{\min_{i}\left|\alpha[i]\right|^{2}}{\max_{i}\left|\alpha[i]\right|^{2}}\cdot I_{K}-O(N^{2})\cdot\left\Vert \alpha\right\Vert _{\infty}\left\Vert \widehat{n}_{\sigma_{2}}\right\Vert _{2}\cdot I_{K}\nonumber \\
 & =\Omega(N^{2})\cdot\|\alpha\|_{\infty}^{2}\cdot\left(\mbox{dyn}\left(x_{\tau,\alpha}\right)^{-2}-O(1)\cdot\frac{\left\Vert \widehat{n}_{\sigma_{2}}\right\Vert _{2}}{\|\alpha\|_{\infty}}\right)\cdot I_{K},
\qquad \mbox{(see \eqref{eq:dyn range def})}
 \label{eq:low bnd on Hessian at tau}
\end{align}
asymptotically. 
 Therefore, as long as
\[
\frac{\left\Vert \widehat{n}_{\sigma_{2}}\right\Vert _{2}}{\|\alpha\|_{\infty}}=O(1)\cdot\mbox{dyn}\left(x_{\tau,\alpha}\right)^{-2}\le1,
\]
with a small enough constant, $\frac{\partial^{2}F}{\partial\rho^{2}}(\tau)\succ0$ asymptotically (as we hoped  to establish).

\subsection{Establishing $\frac{\partial^{2}F}{\partial\rho^{2}}(\rho)\succ0$ When $\rho$
is Close to $\tau$}

It should be clear that, by continuity,
\[
\frac{\partial^{2}F}{\partial\rho^{2}}(\tau)\succ0\Longrightarrow\frac{\partial^{2}F}{\partial\rho^{2}}(\rho)\succ0,
\]
when $\rho\in\mathbb{I}^{K}$ is sufficiently close to $\tau$. In
this section, we precisely calculate the neighborhood of $\tau$ in $\mathbb{I}^{K}$
over which the Hessian of $F(\cdot)$ is positive definite.
To that end, for $\rho\in\mathbb{B}(\tau,2\sigma_{1})$, we write
that
\[
\frac{\partial^{2}F}{\partial\rho^{2}}(\rho)=\frac{\partial^{2}F}{\partial\rho^{2}}(\tau)+\left(\frac{\partial^{2}F}{\partial\rho^{2}}(\rho)-\frac{\partial^{2}F}{\partial\rho^{2}}(\tau)\right),
\]
and, to control the variation, note that 
\begin{align}
\left\Vert \frac{\partial^{2}F}{\partial\rho^{2}}(\rho)-\frac{\partial^{2}F}{\partial\rho^{2}}(\tau)\right\Vert  & = \left\Vert \left(\mbox{signal}_{\rho}+\mbox{noise}_{\rho}\right)-\left(\mbox{signal}_{\tau}+\mbox{noise}_{\tau}\right)\right\Vert \qquad\mbox{(see (\ref{eq:Hess psd argue}))}\nonumber \\
 & \le\left\Vert \mbox{signal}_{\rho}-\mbox{signal}_{\tau}\right\Vert +\left\Vert \mbox{noise}_{\rho}\right\Vert +\left\Vert \mbox{noise}_{\tau}\right\Vert .\label{eq:varitation of Hessian}
\end{align}
We begin by comparing the signal terms of the Hessian at $\rho$ and
$\tau$ in the asymptotic regime $c,N\rightarrow\infty$ and  $c=\Theta(\log N)$ (with a large enough lower bound). Below, we repeatedly use the identity $AB-CD=(A-C)D+A(B-D)$
for conformal matrices $A,B,C,D$. After recalling (\ref{eq:signal part}),
we write that 
\begin{align*}
 & \left\Vert \mbox{signal}_{\rho}-\mbox{signal}_{\tau}\right\Vert \nonumber \\
 & =2\Bigg\Vert \overset{A}{\overbrace{\mbox{diag}\left(\beta_{\rho}\right)\cdot G_{\rho}^{*}L}}\cdot\overset{B}{\overbrace{\left(I_{N}-\mathcal{P}_{\rho}\right)L^{*}G_{\rho}\cdot\mbox{diag}\left(\beta_{\rho}\right)}}\nonumber\\
& \qquad\qquad  
 -\overset{C}{\overbrace{\mbox{diag}\left(\beta_{\tau}\right)\cdot G_{\tau}^{*}L}}\cdot\overset{D}{\overbrace{\left(I_{N}-\mathcal{P}_{\tau}\right)L^{*}G_{\tau}\cdot\mbox{diag}\left(\beta_{\tau}\right)}}\Bigg\Vert \nonumber \\
 & \le4\left\Vert L^{*}G_{\rho}\cdot\mbox{diag}\left(\beta_{\rho}\right)-L^{*}G_{\tau}\cdot\mbox{diag}\left(\beta_{\tau}\right)\right\Vert \cdot\max\left[\left\Vert L^{*}G_{\rho}\cdot\mbox{diag}\left(\beta_{\rho}\right)\right\Vert ,\left\Vert L^{*}G_{\tau}\cdot\mbox{diag}\left(\beta_{\tau}\right)\right\Vert \right]\nonumber \\
 & \le16\pi^2N^2 \left\Vert \overset{A}{\overbrace{G_{\rho}}}\cdot\overset{B}{\overbrace{\mbox{diag}\left(\beta_{\rho}\right)}}-\overset{C}{\overbrace{G_{\tau}}}\cdot\overset{D}{\overbrace{\mbox{diag}\left(\beta_{\tau}\right)}}\right\Vert \nonumber\\
& \qquad\qquad   \cdot \max\left[\left\Vert G_{\rho}\right\Vert \left\Vert \beta_{\rho}\right\Vert _{\infty},\left\Vert G_{\tau}\right\Vert \left\Vert \beta_{\tau}\right\Vert _{\infty}\right],
\quad\left(\|L\|\le2\pi N\right)
\end{align*}
and, consequently, 
\begin{align}
 & \left\Vert \mbox{signal}_{\rho}-\mbox{signal}_{\tau}\right\Vert \nonumber \\
 & \le16\pi^2N^2\left(\left\Vert G_{\rho}-G_{\tau}\right\Vert \left\Vert \beta_{\tau}\right\Vert _{\infty}+\left\Vert G_{\rho}\right\Vert \left\Vert \beta_{\rho}-\beta_{\tau}\right\Vert _{\infty}\right)\cdot \max\left[\left\Vert G_{\rho}\right\Vert \left\Vert \beta_{\rho}\right\Vert _{\infty},\left\Vert G_{\tau}\right\Vert \left\Vert \beta_{\tau}\right\Vert _{\infty}\right]\nonumber \\
 & = O(N^{2})\left(\left(\left\Vert I_{K}-M_{\rho,\tau}\right\Vert ^{\frac{1}{2}}+e^{-Cc}\right)\left\Vert \beta_{\tau}\right\Vert _{\infty}+\left\Vert \beta_{\rho}-\beta_{\tau}\right\Vert _{\infty}\right)\nonumber\\
& \qquad\qquad  \cdot\max\left[\left\Vert \beta_{\rho}\right\Vert _{\infty},\left\Vert \beta_{\tau}\right\Vert _{\infty}\right].\qquad\mbox{(Lemma \ref{lem:props of Psi})}\label{eq:mid step 100}
\end{align}
We further simplify the last line above as follows. 
Since $\rho\in\mathbb{B}(\tau,2\sigma_{1})$ and $2\sigma_{1}\le h(\sigma_{2},N)\le\sigma_{2}$
(all by hypothesis), Criterion \ref{fact: slow decay in bound-1}
is in force and, asymptotically, we may write that 
\begin{align}
& \left\Vert I_{K}-M_{\rho,\tau}\right\Vert  \nonumber\\
& =\max_{i\in [1:K]}\left|1-\left\langle g_{\sigma_{2},N}(t\ominus\rho[i]),g_{\sigma_{2},N}(t\ominus\tau[i])\right\rangle \right|
\qquad \mbox{(see \eqref{eq:def of M12})}
\nonumber \\
 & =O(1)\cdot\max_{i\in [1:K]}\, d\left(\rho[i],\tau[i]\right)\qquad\mbox{(Criterion \ref{fact: slow decay in bound-1})}\nonumber \\
 & = O(1)\cdot d(\rho,\tau).\qquad\mbox{(see \eqref{eq:Hausdoff dist})}\label{eq:mid step 112}
\end{align}
Moreover, to remove the terms involving $\beta_{\tau}$ and $\beta_{\rho}$
in (\ref{eq:mid step 100}), we invoke Lemma \ref{lem:ls dist to true}
to write the following asymptotic estimates:
\begin{align}
\left\Vert \beta_{\rho}-\beta_{\tau}\right\Vert _{\infty} & =O(1)\cdot\left(\left\Vert I_{K}-M_{\rho,\tau}\right\Vert ^{\frac{1}{2}}+e^{-Cc_2}\right)\cdot \left(\sqrt{K}\|\alpha\|_{\infty}+\left\Vert \widehat{n}_{\sigma_{2}}\right\Vert _{2}\right)\nonumber\\
 & =O(1)\cdot\left(d(\rho,\tau)^{\frac{1}{2}}+e^{-Cc_2}\right)
 \cdot \left(\sqrt{K}\|\alpha\|_{\infty}+\left\Vert \widehat{n}_{\sigma_{2}}\right\Vert _{2}\right),
 \qquad \mbox{(see \eqref{eq:mid step 112})}
 \label{eq:some bnds on beta real pre}
\end{align}
\begin{align}
\left\Vert \beta_{\rho}\right\Vert _{\infty} & \le\|\alpha\|_{\infty}+\left\Vert \beta_{\rho}-\alpha\right\Vert _{\infty}\nonumber \\
 & =\|\alpha\|_{\infty}+O(1)\cdot\left(\sqrt{K}\left\Vert I_{K}-M_{\rho,\tau}\right\Vert ^{\frac{1}{2}}+e^{-Cc_2}\right)\cdot\|\alpha\|_{\infty}+O(1)\left\Vert \widehat{n}_{\sigma_{2}}\right\Vert _{2}\nonumber \\
 & =O(1)\cdot\left(\left(1+\sqrt{K}\cdot d(\rho,\tau)^{\frac{1}{2}}\right)\cdot\|\alpha\|_{\infty}+\left\Vert \widehat{n}_{\sigma_{2}}\right\Vert _{2}\right).\label{eq:some bnds on beta real}
\end{align}
Using the estimates above, we revisit (\ref{eq:mid step 100}):  
\begin{align}
 & \left\Vert \mbox{signal}_{\rho}-\mbox{signal}_{\tau}\right\Vert \nonumber \\
 & = O(N^{2})\left(\left(\left\Vert I_{K}-M_{\rho,\tau}\right\Vert ^{\frac{1}{2}}+e^{-Cc_2}\right)\left\Vert \beta_{\tau}\right\Vert _{\infty}+\left\Vert \beta_{\rho}-\beta_{\tau}\right\Vert _{\infty}\right)\cdot\max\left[\left\Vert \beta_{\rho}\right\Vert _{\infty},\left\Vert \beta_{\tau}\right\Vert _{\infty}\right]\nonumber \\
 & =O(N^{2})\left(\left(d(\rho,\tau)^{\frac{1}{2}}+e^{-Cc_2}\right)+\left(d(\rho,\tau)^{\frac{1}{2}}+e^{-Cc_2}\right)\right)\left(\sqrt{K}\|\alpha\|_{\infty}+\left\Vert \widehat{n}_{\sigma_{2}}\right\Vert _{2}\right)\nonumber \\
 & \qquad\cdot\left(\left(1+\sqrt{K}\cdot d(\rho,\tau)^{\frac{1}{2}}\right)\|\alpha\|_{\infty}+\left\Vert \widehat{n}_{\sigma_{2}}\right\Vert _{2}\right)
\qquad \mbox{(see \eqref{eq:some bounds on beta}
, (\ref{eq:mid step 112}-\ref{eq:some bnds on beta real})
)} 
 \nonumber \\
 & =O(N^{2})\left(d(\rho,\tau)^{\frac{1}{2}}+e^{-Cc_2}\right)\left(\sqrt{K}\|\alpha\|_{\infty}+\left\Vert \widehat{n}_{\sigma_{2}}\right\Vert _{2}\right)^{2}
\,\left(\rho \in \mathbb{B}\left(\tau,\sigma_1\right) \Longrightarrow K\cdot d(\rho,\tau)=o(1) \right) 
 \nonumber \\
 & =O(N^{2})\cdot\left(d(\rho,\tau)^{\frac{1}{2}}+e^{-Cc_2}\right)\cdot\left(K\|\alpha\|_{\infty}^{2}+\left\Vert \widehat{n}_{\sigma_{2}}\right\Vert _{2}^{2}\right)\qquad\left((a+b)^{2}\le2a^{2}+2b^{2},\quad\forall a,b\in\mathbb{R}\right)\nonumber \\
 & =O(N^{2})\cdot\left(d(\rho,\tau)^{\frac{1}{2}}+e^{-Cc_2}\right)\cdot\left(K\|\alpha\|_{\infty}^{2}+\|\alpha\|_{\infty}\left\Vert \widehat{n}_{\sigma_{2}}\right\Vert _{2}\right).\qquad\left(\mbox{if }\left\Vert \widehat{n}_{\sigma_{2}}\right\Vert _{2}= O(1)\cdot \|\alpha\|_{\infty}\right)\label{eq:signal rho - signal tau}
\end{align}

It remains to control $\mbox{noise}_\tau$ and $\mbox{noise}_\rho$ in \eqref{eq:varitation of Hessian}. In the analysis that started in  \eqref{eq:all noise terms} and led to \eqref{eq:noise bound}, we earlier bounded $\|\mbox{noise}_\tau\|$. So we turn our attention to $\mbox{noise}_\rho$.  Note that 
\begin{align*}
\widehat{z}_{\sigma_2} & =  G_\tau\alpha+\widehat{n}_{\sigma_2} \qquad \mbox{(see \eqref{eq:def of z hat})}\\
& = G_\rho \alpha + \widehat{n}_{\sigma_2} +\left(  G_\tau-G_\rho \right)\alpha =: G_\rho\alpha+\widehat{n}'_{\sigma_2},
\end{align*}
\begin{align*}
\beta_\rho & = G_\rho^\dagger \widehat{z}_{\sigma_2}\\
& =  G_\rho^\dagger \left( G_\rho \alpha + \widehat{n}'_{\sigma_2} \right), 
\end{align*}
\begin{equation*}
\left(I_{N}-\mathcal{P}_{\rho}\right)\widehat{z}_{\sigma_{2}} 
  =\left(I_{N}-\mathcal{P}_{\rho}\right)\widehat{n}'_{\sigma_2},
\end{equation*}
so that we next apply  \eqref{eq:noise bound} but with $\rho$ and $\widehat{n}'_{\sigma_2}=\widehat{n}_{\sigma_2}+(G_\tau-G_\rho)\alpha$ (instead of $\tau$ and $\widehat{n}_{\sigma_2}$) to obtain that 
\begin{align*}
& \left\Vert \mbox{noise}_{\rho}\right\Vert \nonumber \\
 & =O(N^{2})\cdot\left[\left\Vert \beta_{\rho}\right\Vert _{\infty}\left\Vert \widehat{n}'_{\sigma_2}\right\Vert _{\infty}+\left\Vert \widehat{n}'_{\sigma_2}\right\Vert _{\infty}^{2}\right]\nonumber \\
 & =O(N^{2})\cdot\left\Vert \beta_{\rho}\right\Vert _{\infty}\left\Vert \widehat{n}'_{\sigma_2}\right\Vert _{\infty}\qquad\mbox{(if }\left\Vert \widehat{n}'_{\sigma_2}\right\Vert _{\infty}\le\left\Vert \beta_{\rho}\right\Vert _{\infty}\mbox{)}\nonumber \\
 & = O(N^2)\cdot \left\|\beta_\rho\right\|_\infty 
 \left\| \widehat{n}_{\sigma_{2}}+\left(G_{\tau}-G_{\rho}\right)\alpha \right\|_\infty\nonumber\\
 & \le O(N^{2})\cdot\left(\left(1+\sqrt{K}\cdot d(\rho,\tau)^{\frac{1}{2}}\right)\cdot\left\Vert \alpha\right\Vert _{\infty}+\left\Vert \widehat{n}_{\sigma_{2}}\right\Vert _{2}\right)\nonumber\\
& \qquad \qquad \cdot  \left(\left\Vert \widehat{n}_{\sigma_{2}}\right\Vert _{2}+\left\Vert G_{\tau}-G_{\rho}\right\Vert _{\infty\rightarrow\infty}\left\Vert \alpha\right\Vert _{\infty}\right),\quad \mbox{(see }(\ref{eq:some bnds on beta real}))
\end{align*}
and, consequently, 
\begin{align}
 & \left\Vert \mbox{noise}_{\rho}\right\Vert \nonumber \\
& =O(N^{2})\cdot\left(\left(1+\sqrt{K}\cdot d(\rho,\tau)^{\frac{1}{2}}\right)\cdot\left\Vert \alpha\right\Vert _{\infty}+\left\Vert \widehat{n}_{\sigma_{2}}\right\Vert _{2}\right)\nonumber \\
& \qquad\cdot\left(\left\Vert \widehat{n}_{\sigma_{2}}\right\Vert _{2}+\left(d(\rho,\tau)^{\frac{1}{2}}+e^{-Cc_2}\right)\left\Vert \alpha\right\Vert _{\infty}\right)\hfill\mbox{(see \eqref{eq:diff between Gs} and (\ref{eq:mid step 112}))}\nonumber \\
 & =O(N^{2})\cdot\left(\left(\sqrt{K}\cdot d(\rho,\tau)^{\frac{1}{2}} + e^{-Cc_2} \right)\left\Vert \alpha\right\Vert _{\infty}+\left\Vert \widehat{n}_{\sigma_{2}}\right\Vert _{2}\right)^{2}
\quad \left(  d(\rho,\tau)\le \frac{1}{2},\,\, 
e^{-Cc_2} = o(1) \right) 
 \nonumber \\
 & =O(N^{2})\cdot\left(  \left(K\cdot d(\rho,\tau)+e^{-Cc_2}\right) \left\Vert \alpha\right\Vert _{\infty}^{2}+\left\Vert \widehat{n}_{\sigma_{2}}\right\Vert _{2}^{2}\right)\qquad\left((a+b)^{2}\le2a^{2}+2b^{2},\quad\forall a,b\in\mathbb{R}\right)\nonumber \\
 & =O(N^{2})\cdot\left(
\left( K\cdot d(\rho,\tau)+e^{-Cc_2}\right) \left\Vert \alpha\right\Vert _{\infty}^{2}+\|\alpha\|_{\infty}\left\Vert \widehat{n}_{\sigma_{2}}\right\Vert _{2}\right).\qquad\left(\mbox{if }\left\Vert \widehat{n}_{\sigma_{2}}\right\Vert _{2}\le\|\alpha\|_{\infty}\right)\label{eq:bound on noise rho}
\end{align}
Using Lemmas \ref{lem:props of Psi} and \ref{lem:ls dist to true}, 
it is not difficult to verify that both conditions imposed while deriving  \eqref{eq:bound on noise rho} hold if
\[
\left\Vert \widehat{n}_{\sigma_{2}}\right\Vert _{2}=O(1)\cdot\frac{\|\alpha\|_{\infty}}{1-\sqrt{K}\cdot d(\rho,\tau)^{\frac{1}{2}}},
\] 
with a small enough constant. Since $K\cdot d(\rho,\tau) = o(1)$ for any $\rho\in \mathbb{B}(\tau,2\sigma_1)$, the condition above is met asymptotically when $\|\widehat{n}_{\sigma_2}\|_2=O(1)\|\alpha\|_\infty$ (with a small enough constant). In light of (\ref{eq:varitation of Hessian}), we can
combine the estimates above to obtain that
\begin{align}
 & \left\Vert \frac{\partial^{2}F}{\partial\rho^{2}}(\rho)-\frac{\partial^{2}F}{\partial\rho^{2}}(\tau)\right\Vert \nonumber \\
 & \le\left\Vert \mbox{signal}_{\rho}-\mbox{signal}_{\tau}\right\Vert +\left\Vert \mbox{noise}_{\rho}\right\Vert +\left\Vert \mbox{noise}_{\tau}\right\Vert \nonumber \\
 & =O(N^{2})\cdot\|\alpha\|_{\infty}^{2}\cdot\left(d(\rho,\tau)^{\frac{1}{2}}+e^{-Cc_2}\right)\cdot\left(K+\frac{\left\Vert \widehat{n}_{\sigma_{2}}\right\Vert _{2}}{\|\alpha\|_{\infty}}\right)\nonumber\\
 & \,\,  
 +O(N^{2})\cdot\left\Vert \alpha\right\Vert _{\infty}^{2}\cdot\left( K\cdot d(\rho,\tau)+e^{-Cc_2}+\frac{\left\Vert \widehat{n}_{\sigma_{2}}\right\Vert _{2}}{\|\alpha\|_{\infty}}\right)
 \qquad \mbox{(see  (\ref{eq:noise bound}), (\ref{eq:signal rho - signal tau}),
and (\ref{eq:bound on noise rho}))}
 \nonumber \\
 & =O(N^{2})\cdot\|\alpha\|_{\infty}^{2}\nonumber\\
 & 
 \cdot\left[K\cdot d(\rho,\tau)^{\frac{1}{2}}+e^{-Cc_2}+\frac{\left\Vert \widehat{n}_{\sigma_{2}}\right\Vert _{2}}{\|\alpha\|_{\infty}}\right].
\quad \left(c=\Theta(\log N)\Rightarrow
Ke^{-Cc_2}=O\left(e^{-Cc_2} \right)
\right) 
 \label{eq:mid step 112-1}
\end{align}
The lower bound in $c=\Theta(\log N)$ must be sufficiently large for the last line above to hold. It immediately follows that 
\begin{align*}
  \frac{\partial^{2}F}{\partial\rho^{2}}(\rho)
 & \succcurlyeq\frac{\partial^{2}F}{\partial\rho^{2}}(\tau)-\left\Vert \frac{\partial^{2}F}{\partial\rho^{2}}(\rho)-\frac{\partial^{2}F}{\partial\rho^{2}}(\tau)\right\Vert \cdot I_{K},
\end{align*}
and, consequently, 
\begin{align*}
& \frac{\partial^{2}F}{\partial\rho^{2}}(\rho)\nonumber\\
 & \succcurlyeq\Omega(N^{2})\cdot\|\alpha\|_{\infty}^{2}\cdot\left[\mbox{dyn}\left(x_{\tau,\alpha}\right)^{-2}-O(1)\cdot\frac{\left\Vert \widehat{n}_{\sigma_{2}}\right\Vert _{2}}{\|\alpha\|_{\infty}}\right]\cdot I_{K}\\
 & \qquad-O(N^{2})\cdot\|\alpha\|_{\infty}^{2}\cdot\left[K\cdot d(\rho,\tau)^{\frac{1}{2}} +e^{-Cc_2}+\frac{\left\Vert \widehat{n}_{\sigma_{2}}\right\Vert _{2}}{\|\alpha\|_{\infty}}\right]\cdot I_{K}\qquad\mbox{(see (\ref{eq:low bnd on Hessian at tau}) and (\ref{eq:mid step 112-1}))}\\
 & =\Omega(N^{2})\cdot\|\alpha\|_{\infty}^{2}\cdot\left[\mbox{dyn}\left(x_{\tau,\alpha}\right)^{-2}-O(K)\cdot d(\rho,\tau)^{\frac{1}{2}}-O\left( e^{-Cc_2}\right)-O(1) \cdot\frac{\left\Vert \widehat{n}_{\sigma_{2}}\right\Vert _{2}}{\|\alpha\|_{\infty}}\right]\cdot I_{K},
\end{align*}
asymptotically. The component $O(K)\cdot d(\rho,\tau)^{\frac{1}{2}}+O(e^{-Cc_2})$ in the last line above is asymptotically negligible.  
%
Indeed, $d(\rho,\tau)=o(1)$ when $\rho\in\mathbb{B}(\tau^0,\sigma_1)$ (with $\sigma_1=\frac{c_1}{N}$),  and also $e^{-Cc_2}=o(1)$.  Therefore, as long as 
\begin{equation}
\frac{\left\Vert \widehat{n}_{\sigma_{2}}\right\Vert _{2}}{\|\alpha\|_{\infty}}
\le \frac{\left\Vert n(\cdot)\right\Vert _{L_{2}}}{\|\alpha\|_{\infty}}
=\frac{O(1)}{\mbox{dyn}\left(x_{\tau,\alpha}\right)^{2}},\label{eq:cnd on SNR for PD}
\end{equation}
and with a small enough constant, $\frac{\partial^{2}F}{\partial\rho^{2}}(\rho)\succ0$
holds asymptotically for every $\rho\in\mathbb{B}\left(\tau,2\sigma_{1}\right)$. In particular, because 
\[
\rho\in\mathbb{B}\left(\tau^{0},\sigma_{1}\right)\Longrightarrow d\left(\rho,\tau\right)\le d\left(\rho,\tau^{0}\right)+d\left(\tau^{0},\tau\right)\le2\sigma_1,
\]
we asymptotically have that $\frac{\partial^{2}F}{\partial\rho^{2}}(\rho)\succ0$
 for every $\rho\in\mathbb{B}\left(\tau^0,\sigma_{1}\right)$. 
The inequality in \eqref{eq:cnd on SNR for PD} is established next:

\begin{align}
\left\Vert \widehat{n}_{\sigma_{2}}\right\Vert _{2} & =\left\Vert \left(g_{\sigma_{2},N}\circledast n\right)(\cdot)\right\Vert _{L_{2}}\qquad\mbox{(Parseval's identity and (\ref{eq:filt meas}))}\nonumber \\
 & \le\left\Vert g_{\sigma_{2},N}(\cdot)\right\Vert _{L_{2}}\|n(\cdot)\|_{L_{2}}\nonumber \\
 & =\|n(\cdot)\|_{L_{2}}.\qquad
\mbox{(see Criterion \ref{lem:corr decay away})}
 \label{eq:filtered noise to noise}
\end{align}

\section{Proof of Lemma \protect\ref{lem:grad is away from zero}\label{sec:Proof-of-Lemma grad is away-1}}

Recall that the entries of $\tau\in\I^{K}$ are distinct. Because
$\rho\in\mathbb{B}\left(\tau^{0},\sigma_{1}\right)$ (by Proposition
\ref{thm:algorithm 1}), the entries of $\rho$ too are distinct asymptotically
(i.e., for large enough $N$). Moreover, $\rho[i]\ne\tau[j]$ for
$i\ne j$. Therefore, we are in position to apply the technical lemmas
in the Toolbox Section.

With $u\in\I^{K}$ to be set later, let $U=\mbox{diag}(u)\in\mathbb{R}^{K\times K}$
be the diagonal matrix formed by the vector $u$. Using the expression
for the gradient of $F(\cdot)$ from the accompanying document \cite{support}, 
we can write that
\begin{align*}
& \left\langle \frac{\partial F}{\partial\rho}(\rho),-u\right\rangle \nonumber \\
 & =\left\langle -\frac{\partial F}{\partial\rho}(\rho),u\right\rangle \nonumber \\
 & =\left\langle \mbox{diag}(\beta_{\rho})\cdot G_{\rho}^{*}L\left(\widehat{z}_{\sigma_{2}}-G_{\rho}\beta_{\rho}\right),u\right\rangle, 
 \qquad 
\left(\beta_{\rho}=\widetilde{\beta}(\rho)=G_{\rho}^{\dagger}\widehat{z}_{\sigma_{2}}\right) 
\end{align*}
and, consequently,
\begin{align}
 & \left\langle \frac{\partial F}{\partial\rho}(\rho),-u\right\rangle \nonumber \\
 & =\left\langle \mbox{diag}(\beta_{\rho})\cdot G_{\rho}^{*}L\left(G_{\tau}\alpha+\widehat{n}_{\sigma_{2}}-G_{\rho}\beta_{\rho}\right),u\right\rangle
\qquad \mbox{(see \eqref{eq:def of z hat})}  \nonumber \\
 & =\left\langle LG_{\tau}\alpha+L\widehat{n}_{\sigma_{2}}-LG_{\rho}\beta_{\rho},G_{\rho}U\beta_{\rho}\right\rangle \nonumber \\
 & =\left\langle LG_{\tau}\alpha+L\widehat{n}_{\sigma_{2}}-LG_{\rho}G_{\rho}^{\dagger}\left(G_{\tau}\alpha+\widehat{n}_{\sigma_{2}}\right),G_{\rho}UG_{\rho}^{\dagger}\left(G_{\tau}\alpha+\widehat{n}_{\sigma_{2}}\right)\right\rangle \nonumber \\
 & =\left\langle L\left(I_{N}-\mathcal{P}_{\rho}\right)G_{\tau}\alpha+L\left(I_{N}-\mathcal{P}_{\rho}\right)\widehat{n}_{\sigma_{2}},\mathcal{P}_{\rho,U}G_{\tau}\alpha+\mathcal{P}_{\rho,U}\widehat{n}_{\sigma_{2}}\right\rangle \,\,\left(\mathcal{P}_{\rho,U}:=G_{\rho}UG_{\rho}^{\dagger},\quad\mathcal{P}_{\rho}=\mathcal{P}_{\rho,I}\right)\nonumber \\
 & =\underset{\mbox{signal terms}}{\underbrace{\left\langle LG_{\tau}\alpha,\mathcal{P}_{\rho,U}G_{\tau}\alpha\right\rangle -\left\langle L\mathcal{P}_{\rho}G_{\tau}\alpha,\mathcal{P}_{\rho,U}G_{\tau}\alpha\right\rangle }}\nonumber \\
 & +\underset{\mbox{noise terms}}{\underbrace{\left\langle L\left(I_{N}-\mathcal{P}_{\rho}\right)\widehat{n}_{\sigma_{2}},\mathcal{P}_{\rho,U}G_{\tau}\alpha\right\rangle +\left\langle L\left(I_{N}-\mathcal{P}_{\rho}\right)G_{\tau}\alpha,\mathcal{P}_{\rho,U}\widehat{n}_{\sigma_{2}}\right\rangle +\left\langle L\left(I_{N}-\mathcal{P}_{\rho}\right)\widehat{n}_{2},\mathcal{P}_{\rho,U}\widehat{n}_{\sigma_{2}}\right\rangle }}.\label{eq:1st step in bounding grad of F}
\end{align}
In order to find a lower bound for the inner product $\langle-\frac{\partial F}{\partial\rho}(\rho),u\rangle$,
we will study each of the five terms in the last identity in (\ref{eq:1st step in bounding grad of F}).
The first term there can be approximated with a simpler quantity as
follows. Asymptotically, we have that
\begin{align}
 & \left|\left\langle LG_{\tau}\alpha,\mathcal{P}_{\rho,U}G_{\tau}\alpha\right\rangle -\left\langle M{}_{\rho,\tau}^{d}\alpha,M_{\rho,\tau}U\alpha\right\rangle \right|
\qquad \mbox{(see \eqref{eq:def of M12} and \eqref{eq:def of M12d})} 
 \nonumber \\
 & \le\left|\left\langle LG_{\tau}\alpha,\mathcal{P}_{\rho,U}G_{\tau}\alpha\right\rangle -\left\langle LG_{\tau}\alpha,G_{\rho}UM_{\rho,\tau}\alpha\right\rangle \right|+\left|\left\langle LG_{\tau}\alpha,G_{\rho}UM_{\rho,\tau}\alpha\right\rangle -\left\langle M{}_{\rho,\tau}^{d}\alpha,M_{\rho,\tau}U\alpha\right\rangle \right|\nonumber \\
 & \le\left\Vert LG_{\tau}\alpha\right\Vert _{2}\cdot\left\Vert \mathcal{P}_{\rho,U}G_{\tau}\alpha-G_{\rho}UM_{\rho,\tau}\alpha\right\Vert _{2}+\left\Vert G_{\rho}^{*}LG_{\tau}-M{}_{\rho,\tau}^{d}\right\Vert \|u\|_{\infty}\|M_{\rho,\tau}\|_{\infty}\|\alpha\|_{2}^{2}\nonumber \\
 & =\left\Vert LG_{\tau}\alpha\right\Vert _{2}\cdot\left\Vert \mathcal{P}_{\rho,U}G_{\tau}\alpha-G_{\rho}UM_{\rho,\tau}\alpha\right\Vert _{2} \nonumber\\
& \qquad +O(e^{-Cc_{2}})\|u\|_{\infty}\|\alpha\|_{2}^{2}\qquad\mbox{(see (\ref{eq:props of Psi 6-1}) and text below)}\nonumber \\
 & =\left\Vert LG_{\tau}\alpha\right\Vert _{2}O(e^{-Cc_{2}})\|u\|_{\infty}\left\Vert G_{\tau}\alpha\right\Vert _{2}+O(e^{-Cc_{2}})\|u\|_{\infty}\|\alpha\|_{2}^{2}\qquad\mbox{(Lemma \ref{lem:P is simple})}\nonumber \\
 & =\|L\|\cdot O(e^{-Cc_{2}})\|u\|_{\infty}\left\Vert \alpha\right\Vert _{2}^{2}+O(e^{-Cc_{2}})\|u\|_{\infty}\|\alpha\|_{2}^{2}\qquad\mbox{(see (\ref{eq:props of Psi no 1}))}\nonumber \\
 & =O(e^{-Cc_{2}})\|u\|_{\infty}\left\Vert \alpha\right\Vert _{2}^{2}.\qquad\left(\|L\|\le2\pi N,\,\, c_2=\Theta(\log N)\right)\label{eq:simlification of grad of F}
\end{align}
 In the fourth line above, $\|M_{\rho,\tau}\|_{\infty}$ is absorbed
as a constant on account of the asymptotic bound
\[
\left\Vert M_{\rho,\tau}\right\Vert _{\infty}=\left\Vert M_{\rho,\tau}\right\Vert \le\left\Vert G_{\rho}^{*}G_{\tau}\right\Vert +O(e^{-Cc_2})\le\left\Vert G_{\rho}\right\Vert \left\Vert G_{\tau}\right\Vert +O(e^{-Cc_2})\le2,
\]
which holds because $M_{\rho,\tau}$ is diagonal and by Lemma \ref{lem:props of Psi}
(see (\ref{eq:props of Psi no 1}) and (\ref{eq:props of Psi 4})).
In the last line of \eqref{eq:simlification of grad of F}, the lower bound in $c_2=\Theta(\log N)$ must be sufficiently large. 
{} %
Next, we can asymptotically upper-bound the second term in  the last identity
 in (\ref{eq:1st step in bounding grad of F}) as follows:
\begin{align}
\left|\left\langle L\mathcal{P}_{\rho}G_{\tau}\alpha,\mathcal{P}_{\rho,U}G_{\tau}\alpha\right\rangle \right| & \le\|\mathcal{P}_{\rho,U}L\mathcal{P}_{\rho}\|\|G_{\tau}\alpha\|_{2}^{2}\nonumber \\
 & =\left\Vert \left[\left(G_{\rho}^{\dagger}\right)^{*}UG_{\rho}^{*}\right]L\left[G_{\rho}G_{\rho}^{\dagger}\right]\right\Vert \|G_{\tau}\alpha\|_{2}^{2}\qquad
\left( \mathcal{P}_{\rho,U} =  \mathcal{P}_{\rho,U}^* = 
G_\rho U G_\rho^\dagger
\right)
\nonumber \\
 & \le\left\Vert G_{\rho}^{\dagger}\right\Vert \cdot\|u\|_{\infty}\cdot\left\Vert G_{\rho}^{*}LG_{\rho}\right\Vert \cdot\left\Vert G_{\rho}^{\dagger}\right\Vert \cdot\|G_{\tau}\alpha\|_{2}^{2}\nonumber \\
 & =O(e^{-Cc_{2}})\|u\|_{\infty}\|\alpha\|_{2}^{2}.\qquad\mbox{(Lemmas \ref{lem:props of Psi} and \ref{lem: alpha and Psi alpha})}\label{eq:simplification of grad of F 1}
\end{align}
These lemmas are applicable because the entries of $\tau$ and $\rho$
are each distinct. Similarly, we asymptotically upper-bound the third
term on the last identity in (\ref{eq:1st step in bounding grad of F})
as follows:
\begin{align}
& \left|\left\langle L\left(I_{N}-\mathcal{P}_{\rho}\right)\widehat{n}_{\sigma_{2}},\mathcal{P}_{\rho,U}G_{\tau}\alpha\right\rangle \right| \nonumber\\
& \le\|L\|\left\Vert I_{N}-\mathcal{P}_{\rho}\right\Vert \left\Vert \widehat{n}_{\sigma_{2}}\right\Vert _{2}\left\Vert \mathcal{P}_{\rho,U}\right\Vert \left\Vert G_{\tau}\alpha\right\Vert _{2}\nonumber \\
 & =O(N)\left\Vert \widehat{n}_{\sigma_{2}}\right\Vert _{2}\left\Vert u\right\Vert _{\infty}\left\Vert \alpha\right\Vert _{2}.
\qquad \left( \|L\|\le 2\pi N,\,\, \left\|I_N-\mathcal{P}_\rho\right\|\le 1, \mbox{ Lemmas \ref{lem: alpha and Psi alpha} and \ref{lem:P is simple}} \right) 
 \label{eq:simplification of grad of F 4}
\end{align}
Next, 
consider the fourth term in the last identity in (\ref{eq:1st step in bounding grad of F}).
Asymptotically, it holds that
\begin{align}
\left|\left\langle L\left(I_{N}-\mathcal{P}_{\rho}\right)G_{\tau}\alpha,\mathcal{P}_{\rho,U}\widehat{n}_{\sigma_{2}}\right\rangle \right| & \le\|L\|\left\Vert I_{N}-\mathcal{P}_{\rho}\right\Vert \left\Vert G_{\tau}\alpha\right\Vert _{2}\left\Vert \mathcal{P}_{\rho,U}\right\Vert \left\Vert \widehat{n}_{\sigma_{2}}\right\Vert _{2}\nonumber \\
 & =O(N)\left\Vert \alpha\right\Vert _{2}\left\Vert u\right\Vert _{\infty}\left\Vert \widehat{n}_{\sigma_{2}}\right\Vert _{2},\label{eq:simplification of grad of F 5}
\end{align}
with a similar argument. 
Finally, consider the fifth term on the last line of (\ref{eq:1st step in bounding grad of F}):
\begin{align}
\left|\left\langle L\left(I_{N}-\mathcal{P}_{\rho}\right)\widehat{n}_{\sigma_{2}},\mathcal{P}_{\rho,U}\widehat{n}_{\sigma_{2}}\right\rangle \right| & \le\|L\|\left\Vert I_{N}-\mathcal{P}_{\rho}\right\Vert \left\Vert \mathcal{P}_{\rho,U}\right\Vert _{2}\left\Vert \widehat{n}_{\sigma_{2}}\right\Vert _{2}^{2}\nonumber \\
 & =O(N)\left\Vert u\right\Vert _{\infty}\left\Vert \widehat{n}_{\sigma_{2}}\right\Vert _{2}^{2}.\label{eq:simplification of grad of F 6}
\end{align}
We now use (\ref{eq:simlification of grad of F}-\ref{eq:simplification of grad of F 6})
to find a lower bound for the inner product in (\ref{eq:1st step in bounding grad of F}):
\begin{align}
\left\langle -\frac{\partial F}{\partial\rho}(\rho),u\right\rangle  & \ge\left\langle M{}_{\rho,\tau}^{d}\alpha,M_{\rho,\tau}U\alpha\right\rangle -O(e^{-Cc_{2}})\|u\|_{\infty}\|\alpha\|_{2}^{2}\nonumber \\
 & \qquad-O(N)\|u\|_{\infty}\left\Vert \widehat{n}_{\sigma_{2}}\right\Vert _{2}\left\Vert \alpha\right\Vert _{2}-O(N)\left\Vert u\right\Vert _{\infty}\left\Vert \widehat{n}_{\sigma_{2}}\right\Vert _{2}^{2}.\label{eq:2nd step in bounding grad of F}
\end{align}
Let us simplify the lower bound above. To that
end, observe that
\begin{equation}
\left\langle M{}_{\rho,\tau}^{d}\alpha,M_{\rho,\tau}U\alpha\right\rangle =\sum_{i=1}^{K}\left|\alpha[i]\right|^{2}\cdot u[i]\cdot M_{\rho,\tau}[i,i]\cdot M{}_{\rho,\tau}^{d}[i,i],\label{eq:3nd step in bounding grad of F}
\end{equation}
which owes itself to the fact that $U$, $M_{\rho,\tau}$, and $M_{\rho,\tau}^{d}$
are all diagonal matrices. %
First, by design,
\[
\rho,\tau\in\mathbb{B}(\tau^{0},\sigma_{1})\Rightarrow d(\rho,\tau)\le2\sigma_{1}\le h(\sigma_{2},N).
\]
Then, on the account of Criterion \ref{fact: slow decay in bound-1},
we asymptotically have that 
\begin{equation}
M_{\rho,\tau}[i,i]=\left\langle g_{\sigma_{2,N}}(t\ominus\rho[i]),g_{\sigma_{2},N}(t\ominus\tau[i])\right\rangle =\Omega(1),\label{eq:middle step pre}
\end{equation}
\begin{align}
\left|M{}_{\rho,\tau}^{d}[i,i]\right| & =\left|\left\langle g_{\sigma_{2},N}(t\ominus\rho[i]),g'_{\sigma_{2},N}(t\ominus\tau[i])\right\rangle \right|\nonumber \\
 & =\mbox{sign}\left(\rho[i]\ominus\tau[i]-\frac{1}{2}\right)\cdot\left\langle g_{\sigma_{2},N}(t\ominus\rho[i]),g'_{\sigma_{2},N}(t\ominus\tau[i])\right\rangle \nonumber \\
 & =\Omega(N^{2})\cdot d(\rho[i],\tau[i]).\label{eq:middle step}
\end{align}
Second, we choose 
$$
u= \mbox{sign}\left((\rho\ominus \tau)-\frac{1}{2}\right).
$$
With this choice of $u$, 
it asymptotically holds that
\begin{align*}
& \left\langle M{}_{\rho,\tau}^{d}\alpha,M_{\rho,\tau}U\alpha\right\rangle \\
 & =\sum_{i=1}^{K}\left|\alpha[i]\right|^{2}\cdot u[i]\cdot M_{\rho,\tau}[i,i]\cdot M{}_{\rho,\tau}^{d}[i,i]
\qquad \mbox{(see (\ref{eq:3nd step in bounding grad of F}))} 
 \\
 & =\sum_{i=1}^{K}\left|\alpha[i]\right|^{2}\cdot \left|M{}_{\rho,\tau}^{d}[i,i]\right|
 \cdot M_{\rho,\tau}[i,i]\quad\left(\mbox{sign}\left(\rho[i]\ominus\tau[i]-\frac{1}{2}\right)=\mbox{sign}\left(M{}_{\rho,\tau}^{d}[i,i]\right)\right)\\
 & =\Omega(N^{2})\sum_{i=1}^{K}\left|\alpha[i]\right|^{2}\cdot d(\rho[i],\tau[i])\qquad\mbox{(see (\ref{eq:middle step pre}) and (\ref{eq:middle step}))}\\
 & \ge \Omega(N^{2})\cdot\min_{i}|\alpha[i]|^{2}\cdot\max_{i}d(\rho[i],\tau[i])\\
 & =\Omega(N^{2})\cdot\min_{i}|\alpha[i]|^{2}\cdot d(\rho,\tau)\qquad\mbox{(definition of Hausdorff distance in \eqref{eq:Hausdoff dist})}\\
 & =\Omega(N^{2})\cdot\frac{\min_{i}|\alpha[i]|^{2}}{\|\alpha\|_{2}^{2}}\cdot\|\alpha\|_{2}^{2}\cdot d(\rho,\tau)\\
 & =\Omega(K^{-1}N^{2})\cdot\frac{\min_{i}|\alpha[i]|^{2}}{\max_{i}|\alpha[i]|^{2}}\cdot\|\alpha\|_{2}^{2}\cdot d(\rho,\tau),
\end{align*}
and, consequently, 
\begin{align*}
 & \left\langle M{}_{\rho,\tau}^{d}\alpha,M_{\rho,\tau}U\alpha\right\rangle \\
 & \ge \frac{\Omega(K^{-1}N^{2})}{\mbox{dyn}(x_{\tau,\alpha})^{2}}\cdot\|\alpha\|_{2}^{2}\cdot d(\rho,\tau)\\
 & =\frac{\Omega(N)}{\mbox{dyn}(x_{\tau,\alpha})^{2}}\cdot\|\alpha\|_{2}^{2}\cdot d(\rho,\tau).\qquad\left(K\le f_{C}+1=\frac{N+1}{2}\right)
\end{align*}
With our choice of $u$ earlier, we can substitute the bound above into
(\ref{eq:2nd step in bounding grad of F}) to finally obtain that
\begin{align}
\left\langle \frac{\partial F}{\partial\rho}(\rho),
\mbox{sign}\left( (\rho\ominus \tau)-\frac{1}{2}\right)
\right\rangle  & =-\frac{\Omega(N)}{\mbox{dyn}(x_{\tau,\alpha})^{2}}\cdot\|\alpha\|_{2}^{2}\cdot d(\rho,\tau)+O(e^{-Cc_{2}})\cdot\|\alpha\|_{2}^{2}\nonumber \\
 & \qquad\quad+O(N)\cdot\left\Vert \widehat{n}_{\sigma_{2}}\right\Vert _{2}\left\Vert \alpha\right\Vert _{2}+O(N)\cdot\left\Vert \widehat{n}_{\sigma_{2}}\right\Vert _{2}^{2}.\label{eq:5th step bounding grad of F}
\end{align}
The proof of Lemma \ref{lem:grad is away from zero} is complete because
$\|\widehat{n}_{\sigma_{2}}\|_{2}\le\|n(\cdot)\|_{L_{2}}$ (by  (\ref{eq:filtered noise to noise})).

\part*{Supplementary Material}

\section{Computing the Gradient of $F(\cdot)$\label{sec:Computing gradient} }

Here, for fixed $\rho_{0}\in\mathbb{I}^{\Kt}$, we wish to calculate
$\frac{\partial F}{\partial\rho}(\rho_{0})\in\R^{\Kt}$ and verify the  explicit expression in (28).
 Set 
\begin{equation}\label{eq:def of betat}
\beta_{\rho_{0}}=\widetilde{\beta}\left(\rho_{0}\right):=G_{\rho_{0}}^{\dagger}\cdot \widehat{z}_{\sigma_{2}}\in\mathbb{R}^{\Kt},
\end{equation}
where, from (14),
 recall that the entries of $G_{\rho_{0}}\in\mathbb{C}^{N\times \Kt}$
are specified as 
$$
G_{\rho_{0}}[l,i]=\widehat{g}_{\sigma_{2},N}[l]\cdot e^{-\imag2\pi l\rho_{0}[i]},\qquad l\in\mathbb{F},\,\, i\in[1:\Kt].
$$
We use the following identity (which we later establish in Section \ref{sec:analysis}): 
\begin{equation}
\frac{\partial F}{\partial\rho}(\rho_{0})=\frac{\partial f}{\partial\rho}\left(\rho_{0},\widetilde{\beta}(\rho_{0})\right).\label{eq:analysis 1}
\end{equation}
{} It suffices then to compute the right hand side of the above
identity:
\begin{align}
\frac{\partial f}{\partial\rho}\left(\rho_{0},\widetilde{\beta}(\rho_{0})\right) & =\left[\frac{\partial}{\partial\rho}\left\Vert G_{\rho}\beta-\widehat{z}_{\sigma_{2}}\right\Vert _{2}^{2}\right]\left(\rho_{0},\widetilde{\beta}(\rho_{0})\right)\nonumber \\
 & =\left[\frac{\partial}{\partial\rho}\left\langle G_{\rho}\beta-\widehat{z}_{\sigma_{2}},G_{\rho}\beta-\widehat{z}_{\sigma_{2}}\right\rangle \right]\left(\rho_{0},\widetilde{\beta}(\rho_{0})\right)\nonumber \\
 & =2\left(\frac{\partial G_{\rho}\beta}{\partial\rho}\left(\rho_{0},\widetilde{\beta}(\rho_{0})\right)\right)^{*}\left(G_{\rho_{0}}\cdot\widetilde{\beta}(\rho_{0})-\widehat{z}_{\sigma_{2}}\right).\label{eq:pre nec calc}
\end{align}
It only remains to calculate the derivative of $G_{\rho}\beta$
with respect to $\rho$. To that end, we next do some elementary 
calculations.

For $i\in[1:\Kt]$, we can compute the derivative of $G_{\rho}[:,i]\in\mathbb{C}^{N}$
(the $i$th column of $G_{\rho}\in\mathbb{C}^{N\times \Kt}$) with respect
to $\rho[i]$ as
\begin{equation}
\frac{\partial \left(G_{\rho}[:,i]\right) }{\partial\rho[i]}\left(\rho_{0}[i]\right)=\left[\begin{array}{c}
\vdots\\
\frac{\partial e^{-\imag2\pi l\rho[i]}}{\partial\rho[i]}(\rho_{0}[i])\\
\vdots
\end{array}\right]=L^{*}\cdot G_{\rho_{0}}[:,i]\in\mathbb{C}^{N},\label{eq:grad mid s1}
\end{equation}
where the diagonal matrix $L\in\mathbb{C}^{N\times N}$ is specified
by $L[l,l]=\imag2\pi l$ for $ l\in\mathbb{F}$. Above, for clarity,
only the $l$th entry of the long vector is shown. 
In addition, for a vector $v\in\mathbb{R}^{\Kt}$, we observe that
\begin{align}
\frac{\partial\left(G_{\rho}v\right)}{\partial\rho}(\rho_{0}) & =\sum_{i=1}^{\Kt}v[i]\cdot\left[\frac{\partial \left( G_{\rho}[:,i] \right)  }{\partial\rho}\right](\rho_{0})\nonumber \\
 & =\left[\begin{array}{ccc}
\cdots & v[i]\cdot\frac{\partial G_{\rho}[:,i]}{\partial\rho[i]}\left(\rho_{0}[i]\right) & \cdots\end{array}\right]\in \mathbb{C}^{N\times \widetilde{K}}\nonumber \\
 & =\left[\begin{array}{ccc}
\cdots & v[i]\cdot L^{*}\cdot G_{\rho_{0}}[:,i] & \cdots\end{array}\right] \qquad \mbox{(see \eqref{eq:grad mid s1})}
\nonumber \\
 & =L^{*}G_{\rho_{0}}\cdot\mbox{diag}\left(v\right),\label{eq:nec calc  1}
\end{align}
where the second line follows because $G_{\rho}[:,i]$ depends only
on $\rho[i]$. Above, $\mbox{diag}(v)\in\mathbb{R}^{\Kt\times \Kt}$ is
the diagonal vector formed from the entries of $v$. With (\ref{eq:nec calc  1})
at hand, we can plug in for the derivitave of $G_{\rho}\beta$ in
(\ref{eq:pre nec calc}) to obtain that
\begin{align}
& \frac{\partial F}{\partial\rho}(\rho_{0}) \nonumber\\
& =\frac{\partial f}{\partial\rho}\left(\rho_{0},\widetilde{\beta}(\rho_{0})\right) \nonumber\\
& =2\left(\frac{\partial G_{\rho}\beta}{\partial\rho}\left(\rho_{0},\widetilde{\beta}(\rho_{0})\right)\right)^{*}\left(G_{\rho_{0}}\cdot\widetilde{\beta}(\rho_{0})-\widehat{z}_{\sigma_{2}}\right)\nonumber \\
 & =2\left(L^{*}G_{\rho_{0}}\cdot\mbox{diag}(\widetilde{\beta}(\rho_{0}))\right)^{*}\left(G_{\rho_{0}}\cdot\widetilde{\beta}(\rho_{0})-\widehat{z}_{\sigma_{2}}\right)\nonumber \\
 & =2\cdot\mbox{diag}\left(\widetilde{\beta}(\rho_{0})\right)G_{\rho_{0}}^{*}L\left(G_{\rho_{0}}\cdot\widetilde{\beta}(\rho_{0})-\widehat{z}_{\sigma_{2}}\right),
\qquad \left( \widetilde{\beta}(\rho_0) \in \mathbb{R}^{\Kt} \right) \nonumber\\
& = -2\cdot\mbox{diag}\left(\widetilde{\beta}(\rho_{0})\right)G_{\rho_{0}}^{*}L\left(I_{N}-\mathcal{P}_{\rho_{0}}\right)\widehat{z}_{\sigma_{2}},\qquad \mbox{(see \eqref{eq:def of betat})}
 \label{eq:der 2}
\end{align}
where $\mathcal{P}_{\rho_{0}}=G_{\rho_{0}}G_{\rho_{0}}^{\dagger}$
is the orthogonal projection onto the column span of $G_{\rho_{0}}$.
We therefore found an explicit expression for $\frac{\partial F}{\partial\rho}(\rho_{0})$.

\section{Computing the Hessian of $F(\cdot)$ \label{sec:Computing-the-Hessian}}

Here, for fixed $\rho_{0}\in\mathbb{I}^{\Kt}$, we wish to calculate
$\frac{\partial^{2}F}{\partial\rho^{2}}(\rho_{0})\in\R^{\Kt\times \Kt}$ and verify the explicit expression in 
(32). 
 With $\widetilde{\beta}(\rho_0)$ as in \eqref{eq:def of betat}, we will use the following identity (to be established in Section \ref{sec:analysis}):
\begin{align}
\frac{\partial^{2}F}{\partial\rho^{2}}(\rho_{0}) 
& =\frac{\partial^{2}f}{\partial\rho^{2}}\left(\rho_{0},\widetilde{\beta}(\rho_{0})\right)+2\cdot \frac{\partial^{2}f}{\partial\rho\partial\beta}\left(\rho_{0},\widetilde{\beta}(\rho_{0})\right)\cdot\frac{\partial\widetilde{\beta}}{\partial\rho}(\rho_{0})\nonumber \\
 & \qquad+\left(\frac{\partial\widetilde{\beta}}{\partial\rho}(\rho_{0})\right)^{*}\cdot\frac{\partial^{2}f}{\partial\beta^{2}}\left(\rho_{0},\widetilde{\beta}(\rho_{0})\right)\cdot\frac{\partial\widetilde{\beta}}{\partial\rho}(\rho_{0}).\label{eq:analysis 2}
\end{align}
We are now burdened with the laborious task of computing the following
derivatives:
\begin{equation}
\frac{\partial^{2}f}{\partial\rho^{2}}\left(\rho,\beta\right),\qquad\frac{\partial^{2}f}{\partial\rho\partial\beta}\left(\rho,\beta\right),\qquad\frac{\partial^{2}f}{\partial\beta^{2}}\left(\rho,\beta\right),\qquad\frac{\partial\widetilde{\beta}}{\partial\rho}(\rho).\label{eq:all you can differentiate}
\end{equation}
Recall \eqref{eq:grad mid s1} and \eqref{eq:nec calc  1} to facilitate the ensuing arguments. 
%
 Three fresh estimates are needed before calculating the derivatives
in (\ref{eq:all you can differentiate}). These estimates will be presented immediately next and then followed by the body of calculations throughout the rest of this section. As for the first auxiliary result, for a vector $u\in\mathbb{R}^{N}$,
we note that 
\begin{align}
\frac{\partial\left(G_{\rho}^{*}u\right)}{\partial\rho}(\rho) & =\left[\begin{array}{c}
\vdots\\
\frac{\partial\left(\left(G_{\rho}[:,i]\right)^{*}u\right)}{\partial\rho}(\rho)\\
\vdots
\end{array}\right]=\left[\begin{array}{cccc}
\ddots\\
 & \frac{\partial\left(\left(G_{\rho}[:,i]\right)^{*}u\right)}{\partial\rho[i]}(\rho[i])\\
 &  & \ddots
\end{array}\right]\nonumber \\
 & =\left[\begin{array}{cccc}
\ddots\\
 & \left(\frac{\partial\left(G_{\rho}[:,i]\right)}{\partial\rho[i]}(\rho[i])\right)^{*}u\\
 &  & \ddots
\end{array}\right]\nonumber \\
 & =\left[\begin{array}{cccc}
\ddots\\
 & \left(G_{\rho}[:,i]\right)^{*}Lu\\
 &  & \ddots
\end{array}\right] \qquad \mbox{(see \eqref{eq:grad mid s1})} \nonumber \\
 & =\mbox{diag}\left(G_{\rho}^{*}Lu\right)\in\mathbb{R}^{\Kt\times \Kt},\label{eq:nec calc 12}
\end{align}
where the second identity holds because $G_{\rho}[:,i]$ depends only
on $\rho[i]$. Also, note that
\begin{align*}
\mathbb{R}^{\Kt\times \Kt}&
\ni  \frac{\partial\left(G_{\rho}^{*}G_{\rho}\beta\right)}{\partial\rho}(\rho) 
=\left[\begin{array}{c}
\vdots\\
\frac{\partial\left(\left(G_{\rho}[:,i]\right)^{*}G_{\rho}\beta\right)}{\partial\rho}(\rho)\\
\vdots
\end{array}\right]\\
& =\left[\begin{array}{c}
\vdots\\
\left(\frac{\partial \left(G_{\rho}[:,i]\right)}{\partial\rho}(\rho)\right)^{*}{G_{\rho}\beta}\\
\vdots
\end{array}\right]
+
\left[\begin{array}{c}
\vdots\\
\left(G_{\rho}[:,i]\right)^{*}\cdot\frac{\partial\left(G_{\rho}\beta\right)}{\partial\rho}(\rho)\\
\vdots
\end{array}\right]\nonumber \\
 & =\left[\begin{array}{cccc}
\ddots\\
 & \left( L^*\cdot{G_{\rho}[:,i]}\right)^{*} G_{\rho}\beta\\
 &  & \ddots
\end{array}\right]\nonumber\\
& \qquad \qquad +\left[\begin{array}{c}
\vdots\\
\left(G_{\rho}[:,i]\right)^{*}L^{*}G_{\rho}\cdot\mbox{diag}\left(\beta\right)\\
\vdots
\end{array}\right]\quad \mbox{(see \eqref{eq:grad mid s1} and \eqref{eq:nec calc  1})}\\
& =\left[\begin{array}{cccc}
\ddots\\
 & \left(G_{\rho}[:,i]\right)^{*}LG_{\rho}\beta\\
 &  & \ddots
\end{array}\right]+\left[\begin{array}{c}
\vdots\\
\left(G_{\rho}[:,i]\right)^{*}L^{*}G_{\rho}\cdot\mbox{diag}\left(\beta\right)\\
\vdots
\end{array}\right],
\end{align*}
and, consequently, 
\begin{align}
\mathbb{R}^{\Kt\times \Kt}&
\ni  \frac{\partial\left(G_{\rho}^{*}G_{\rho}\beta\right)}{\partial\rho}(\rho) \nonumber\\
 & =\mbox{diag}\left(G_{\rho}^{*}LG_{\rho}\beta\right)+G_{\rho}^{*}L^{*}G_{\rho}\cdot\mbox{diag}(\beta).\label{eq:nec calc 11}
\end{align}
\textbf{}%
Similarly,
\begin{align}
\mathbb{R}^{\Kt\times \Kt}
\ni
\frac{\partial\left(G_{\rho}^{*}LG_{\rho}\beta\right)}{\partial\rho}(\rho,\beta) 
 & =\mbox{diag}\left(G_{\rho}^{*}L^{2}G_{\rho}\beta\right)+G_{\rho}^{*}LL^{*}G_{\rho}\cdot\mbox{diag}(\beta)\nonumber \\
 & =\mbox{diag}\left(G_{\rho}^{*}L^{2}G_{\rho}\beta\right)-G_{\rho}^{*}L^{2}G_{\rho}\cdot\mbox{diag}(\beta).\qquad 
 \left( L^* = -L\right)
 \label{eq:nec calc 10}
\end{align}
Armed with the necessary estimates, we embark on calculating the derivatives
in (\ref{eq:all you can differentiate}). Beginning with $\frac{\partial^{2}f}{\partial\rho^{2}}(\cdot,\cdot)$,
note that
\begin{align*}
\mathbb{R}^{\Kt\times \Kt} & \ni \frac{\partial^{2}f}{\partial\rho^{2}}(\rho,\beta)\\
 & =\frac{\partial}{\partial\rho}\left(\frac{\partial f}{\partial\rho}(\rho,\beta)\right)\\
 & =\left[\frac{\partial}{\partial\rho}\left(2\cdot\mbox{diag}\left(\beta\right)G_{\rho}^{*}L\left(G_{\rho}\beta-\widehat{z}_{\sigma_{2}}\right)\right)\right](\rho,\beta) 
 \qquad \mbox{(see \eqref{eq:der 2})},
\end{align*}
and, consequently,
\begin{align*}
\mathbb{R}^{\Kt\times \Kt} & \ni \frac{\partial^{2}f}{\partial\rho^{2}}(\rho,\beta)\\
 & =2\cdot\mbox{diag}\left(\beta\right)\cdot\frac{\partial\left(G_{\rho}^{*}L\left(G_{\rho}\beta-\widehat{z}_{\sigma_{2}}\right)\right)}{\partial\rho}(\rho,\beta)\\
 & =2\cdot\mbox{diag}\left(\beta\right)\cdot\frac{\partial\left(G_{\rho}^{*}LG_{\rho}\beta\right)}{\partial\rho}(\rho,\beta)-2\cdot\mbox{diag}(\beta)\cdot\frac{\partial\left(G_{\rho}^{*}L\widehat{z}_{\sigma_{2}}\right)}{\partial\rho}(\rho,\beta)\\
 & =-2\cdot\mbox{diag}\left(\beta\right)\cdot G_{\rho}^{*}L^2 G_{\rho}\cdot\mbox{diag}(\beta)+2\cdot\mbox{diag}\left(\beta\right)\cdot\mbox{diag}\left(G_{\rho}^{*}L^{2}G_{\rho}\beta\right)
 \\
 & \qquad-2\cdot\mbox{diag}(\beta)\cdot\mbox{diag}\left(G_{\rho}^{*}L^{2}\widehat{z}_{\sigma_{2}}\right)
\qquad \qquad \mbox{(see (\ref{eq:nec calc 10}) and (\ref{eq:nec calc 12}))} 
 \\
 & =-2\cdot\mbox{diag}\left(\beta\right)\cdot G_{\rho}^{*}L^{2}G_{\rho}\cdot\mbox{diag}(\beta)\nonumber\\
&\qquad\qquad  +2\cdot\mbox{diag}\left(\beta\right)\cdot\mbox{diag}\left(G_{\rho}^{*}L^{2}\left(G_{\rho}\beta-\widehat{z}_{\sigma_{2}}\right)\right).
\end{align*}
In particular, using \eqref{eq:def of betat}, we find that 
\begin{align}
 \frac{\partial^{2}f}{\partial\rho^{2}}\left(\rho_{0},\widetilde{\beta}(\rho_{0})\right) 
& =-2\cdot\mbox{diag}\left(\widetilde{\beta}(\rho_{0})\right)\cdot G_{\rho_{0}}^{*}L^{2}G_{\rho_{0}}\cdot\mbox{diag}\left(\widetilde{\beta}(\rho_{0})\right)\nonumber \\
 & \qquad-2\cdot\mbox{diag}\left(\widetilde{\beta}(\rho_{0})\right)\cdot\mbox{diag}\left(G_{\rho_{0}}^{*}L^{2}\left(I_{N}-\mathcal{P}_{\rho_{0}}\right)\widehat{z}_{\sigma_{2}}\right).
 \label{eq:f r r}
\end{align}
As usual, $\mathcal{P}_{\rho_{0}}=G_{\rho_{0}}G_{\rho_{0}}^{\dagger}$.
In a similar fashion, we compute $\frac{\partial^{2}f}{\partial\beta\partial\rho}(\cdot,\cdot)$
by writing that
\begin{align*}
\R^{\Kt\times \Kt}
 & 
\ni \frac{\partial^{2}f}{\partial\beta\partial\rho}(\rho,\beta) 
\nonumber\\
& =\frac{\partial}{\partial\rho}\left(\frac{\partial f}{\partial\beta}(\rho,\beta)\right)\\
 & =\frac{\partial}{\partial\rho}\left(\frac{\partial\left\Vert G_{\rho}\beta-\widehat{z}_{\sigma_{2}}\right\Vert _{2}^{2}}{\partial\beta}(\rho,\beta)\right) \qquad \mbox{(see (21))}
 \\
 & =\left[\frac{\partial}{\partial\rho}\left(2\cdot G_{\rho}^{*}\left(G_{\rho}\beta-\widehat{z}_{\sigma_{2}}\right)\right)\right](\rho,\beta)\\
 & =2\cdot\frac{\partial\left(G_{\rho}^{*}G_{\rho}\beta\right)}{\partial\rho}(\rho,\beta)-\frac{\partial\left(G_{\rho}^{*}\widehat{z}_{\sigma_{2}}\right)}{\partial\rho}(\rho,\beta)\\
 & =
2\cdot G_{\rho}^{*}L^{*}G_{\rho}\cdot\mbox{diag}\left(\beta\right)+ 
 2\cdot\mbox{diag}\left(G_{\rho}^{*}LG_{\rho}\beta\right)\nonumber\\
&\qquad\qquad  -2\cdot\mbox{diag}\left(G_{\rho}^{*}L\widehat{z}_{\sigma_{2}}\right)
 \qquad 
\mbox{(see  (\ref{eq:nec calc 12}) and(\ref{eq:nec calc 11}))} 
 \\
 & =2\cdot G_{\rho}^{*}L^{*}G_{\rho}\cdot\mbox{diag}\left(\beta\right)+2\cdot\mbox{diag}\left(G_{\rho}^{*}L\left(G_{\rho}\beta-\widehat{z}_{\sigma_{2}}\right)\right).
\end{align*}
Therefore,
\begin{align}
& \frac{\partial^{2}f}{\partial\rho\partial\beta}\left(\rho_{0},\widetilde{\beta}(\rho_{0})\right)  \nonumber\\
&
= \left( \frac{\partial^{2}f}{\partial\beta\partial\rho}\left(\rho_{0},\widetilde{\beta}(\rho_{0})\right) 
\right)^*
\nonumber\\
& =2\cdot \mbox{diag}\left(\widetilde{\beta}(\rho_{0})\right)
\cdot 
 G_{\rho_{0}}^{*}L^{*}G_{\rho_{0}}
+2\cdot\mbox{diag}\left(G_{\rho_{0}}^{*}L\left(G_{\rho_{0}}\widetilde{\beta}(\rho_{0})-\widehat{z}_{\sigma_{2}}\right)\right)\nonumber \\
 & =2\cdot 
\mbox{diag}\left(\widetilde{\beta}(\rho_{0})\right)
\cdot  
 G_{\rho_{0}}^{*}L^{*}G_{\rho_{0}}
 -2\cdot\mbox{diag}\left(G_{\rho_{0}}^{*}L\left(I_{N}-\mathcal{P}_{\rho_{0}}\right)\widehat{z}_{\sigma_{2}}\right).\label{eq:f b r}
\end{align}
Also,
\begin{align*}
\frac{\partial^{2}f}{\partial\beta^{2}}(\rho,\beta) 
& =\frac{\partial}{\partial\beta}\left(\frac{\partial\left\Vert G_{\rho}\beta-\widehat{z}_{\sigma_{2}}\right\Vert _{2}^{2}}{\partial\beta}(\rho,\beta)\right)
\qquad \mbox{(see (21))}
\\
& =2\cdot\frac{\partial\left(G_{\rho}^{*}\left(G_{\rho}\beta-\widehat{z}_{\sigma_{2}}\right)\right)}{\partial\beta}(\rho,\beta)\\
 & =2G_{\rho}^{*}G_{\rho},
\end{align*}
and, clearly,
\begin{equation}
\frac{\partial^{2}f}{\partial\beta^{2}}\left(\rho_{0},\widetilde{\beta}(\rho_{0})\right)=2G_{\rho_{0}}^{*}G_{\rho_{0}}.\label{eq:f b b}
\end{equation}
Lastly, in order to compute $\frac{\partial\widetilde{\beta}}{\partial\rho}(\rho)$,
recall from \eqref{eq:def of betat} that
\begin{align*}
\widetilde{\beta}(\rho) & =G_{\rho}^{\dagger}\widehat{z}_{\sigma_{2}}=\left(G_{\rho}^{*}G_{\rho}\right)^{-1}G_{\rho}^{*}\widehat{z}_{\sigma_{2}},
\end{align*}
or, equivalently, 
\[
G_{\rho}^{*}G_{\rho}\cdot\widetilde{\beta}(\rho)=G_{\rho}^{*}\widehat{z}_{\sigma_{2}}.
\]
The $i$th row of the above identity reads
\[
\left(G_{\rho}^{*}\cdot G_{\rho}[:,i]\right)^{*}\widetilde{\beta}(\rho)=\left(G_{\rho}[:,i]\right)^{*}\widehat{z}_{\sigma_{2}}.
\]
Taking derivatives of both sides (with respect to $\rho$) yields
\[
\left(\frac{\partial\left(G_{\rho}^{*}\cdot G_{\rho}[:,i]\right)}{\partial\rho}(\rho)\right)^{*}\widetilde{\beta}(\rho)+ 
\left(\frac{\partial\widetilde{\beta}}{\partial\rho}(\rho) \right)^*
G_{\rho}^{T}\cdot\overline {G_{\rho}[:,i]} 
=\left(\frac{\partial G_{\rho}[:,i]}{\partial\rho}(\rho)\right)^{*}\widehat{z}_{\sigma_{2}},
\]
where $a^T$ is the transpose of vector $a$, and $\overline{b}$ denotes the complex conjugate of scalar $b$. After rearranging to isolate the target term $\frac{\partial\widetilde{\beta}}{\partial\rho}(\rho)$,
we continue to simplify the above identity:
\begin{align*}
 & \left(\frac{\partial\widetilde{\beta}}{\partial\rho}(\rho) \right)^*
G_{\rho}^{T}\cdot\overline {G_{\rho}[:,i]}  \nonumber \\
 & =\left(\frac{\partial G_{\rho}[:,i]}{\partial\rho}(\rho)\right)^{*}\widehat{z}_{\sigma_{2}}-\left(\frac{\partial\left(G_{\rho}^{*}\cdot G_{\rho}[:,i]\right)}{\partial\rho}(\rho)\right)^{*}\widetilde{\beta}(\rho) \nonumber\\
 & =\left(\frac{\partial G_{\rho}[:,i]}{\partial\rho}(\rho)\right)^{*}\widehat{z}_{\sigma_{2}}\nonumber\\
& \qquad \qquad  
 -\left(\frac{\partial\left(G_{\rho}^{*}G_{\rho}\cdot e_{i}\right)}{\partial\rho}(\rho)\right)^{*}\widetilde{\beta}(\rho), \qquad \left(e_i: \,i\mbox{th canonical vector in }\R^{\Kt} \right)
\end{align*}
and, consequently, 
\begin{align}
 & \left(\frac{\partial\widetilde{\beta}}{\partial\rho}(\rho) \right)^*
G_{\rho}^{T}\cdot\overline {G_{\rho}[:,i]}  \nonumber \\
 & =\left(\left(L^{*}\cdot G_{\rho}[:,i]\right)^{*}\widehat{z}_{\sigma_{2}}\right)\cdot e_{i}\nonumber\\
&\qquad \qquad
 - \left(G_{\rho}^{*}L^{*}G_{\rho}\cdot\mbox{diag}(e_{i})  +\mbox{diag}\left(G_{\rho}^{*}LG_{\rho}e_{i}\right)\right)^{*}\widetilde{\beta}(\rho)
\quad \mbox{(see \eqref{eq:grad mid s1} and \eqref{eq:nec calc 11})} 
 \nonumber\\  
 & =\left(\left(L^{*}\cdot G_{\rho}[:,i]\right)^{*}\widehat{z}_{\sigma_{2}}\right)\cdot e_{i}-\mbox{diag}(e_{i})\cdot G_{\rho}^{*}LG_{\rho}\cdot\widetilde{\beta}(\rho)-\mbox{diag}\left(e_{i}^{*}G_{\rho}^{*}L^{*}G_{\rho}\right)\cdot\widetilde{\beta}(\rho)
 \nonumber\\
 & =\left(\left(L^{*}\cdot G_{\rho}[:,i]\right)^{*}\widehat{z}_{\sigma_{2}}\right)\cdot e_{i}-\left(\left(G_{\rho}[:,i]\right)^{*}LG_{\rho}\cdot\widetilde{\beta}(\rho)\right)\cdot e_{i}-\mbox{diag}\left(\left(G_{\rho}[:,i]\right)^{*}L^{*}G_{\rho}\right)\cdot\widetilde{\beta}(\rho)
 \nonumber\\
 & =\left(\left(G_{\rho}[:,i]\right)^{*}L\widehat{z}_{\sigma_{2}}\right)\cdot e_{i}-\left(\left(G_{\rho}[:,i]\right)^{*}LG_{\rho}\cdot\widetilde{\beta}(\rho)\right)\cdot e_{i}-\mbox{diag}(\widetilde{\beta}(\rho))\cdot G_{\rho}^{*}L{G_{\rho}[:,i]}
 \nonumber\\
 & =\left(\left(G_{\rho}[:,i]\right)^{*}L\left(\widehat{z}_{\sigma_{2}}-G_{\rho}\cdot\widetilde{\beta}(\rho)\right)\right)\cdot e_{i}-\mbox{diag}(\widetilde{\beta}(\rho))\cdot G_{\rho}^{*}L {G_{\rho}[:,i]}.
 \label{eq:pre stack}
\end{align}
The second to last line above uses the identity $\mbox{diag}(a)\cdot b = \mbox{diag}(b) \cdot a$ for vectors $a$ and $b$ of the same length. By stacking the columns for all values of $i$, we obtain that
\begin{align*}
& \left(\frac{\partial\widetilde{\beta}}{\partial\rho}(\rho) \right)^*
G_{\rho}^{T}\overline {G_{\rho}} \\
& =\left[\begin{array}{ccc}
\ddots\\
 & \left(G_{\rho}[:,i]\right)^{*}L\left(\widehat{z}_{\sigma_{2}}-G_{\rho}\cdot\widetilde{\beta}(\rho)\right)\\
 &  & \ddots
\end{array}\right]-
\mbox{diag}(\widetilde{\beta}(\rho))\cdot G_{\rho}^{*}L {G_{\rho}}
\qquad 
 \mbox{(see \eqref{eq:pre stack})}
\\
 & =\mbox{diag}\left(G_{\rho}^{*}L\left(\widehat{z}_{\sigma_{2}}-G_{\rho}\cdot\widetilde{\beta}(\rho)\right)\right)-\mbox{diag}(\widetilde{\beta}(\rho))\cdot G_{\rho}^{*}L {G_{\rho}},
\end{align*}
or
\[
G_{\rho}^{*}G_{\rho}\cdot\frac{\partial\widetilde{\beta}}{\partial\rho}(\rho)=\mbox{diag}\left(G_{\rho}^{*}L\left(\widehat{z}_{\sigma_{2}}-G_{\rho}\cdot\widetilde{\beta}(\rho)\right)\right)-G_{\rho}^{*}L^{*}G_{\rho}\cdot\mbox{diag}(\widetilde{\beta}(\rho)).
\]
We conclude that
\begin{align*}
\frac{\partial\widetilde{\beta}}{\partial\rho}(\rho) & =\left(G_{\rho}^{*}G_{\rho}\right)^{-1}\mbox{diag}\left(G_{\rho}^{*}L\left(\widehat{z}_{\sigma_{2}}-G_{\rho}\cdot\widetilde{\beta}(\rho)\right)\right)-\left(G_{\rho}^{*}G_{\rho}\right)^{-1}G_{\rho}^{*}L^{*}G_{\rho}\cdot\mbox{diag}(\widetilde{\beta}(\rho))\\
 & =\left(G_{\rho}^{*}G_{\rho}\right)^{-1}\mbox{diag}\left(G_{\rho}^{*}L\left(\widehat{z}_{\sigma_{2}}-G_{\rho}\cdot\widetilde{\beta}(\rho)\right)\right)-G_{\rho}^{\dagger}L^{*}G_{\rho}\cdot\mbox{diag}(\widetilde{\beta}(\rho)),
\end{align*}
and, in particular,
\begin{align}
\frac{\partial\widetilde{\beta}}{\partial\rho}(\rho_{0}) & =\left(G_{\rho_{0}}^{*}G_{\rho_{0}}\right)^{-1}\mbox{diag}\left(G_{\rho_{0}}^{*}L\left(I_{N}-\mathcal{P}_{\rho_{0}}\right)\widehat{z}_{\sigma_{2}}\right)-G_{\rho_{0}}^{\dagger}L^{*}G_{\rho_{0}}\cdot\mbox{diag}(\widetilde{\beta}(\rho_{0})).\label{eq:beta tilde}
\end{align}
To summarize, we finished computing all the quantities involved in \eqref{eq:analysis 2} (see (\ref{eq:f r r}-\ref{eq:f b b}), and \eqref{eq:beta tilde}). 
We can simplify the above expression for the Hessian of $F(\cdot)$
by noting that the second and third summands in (\ref{eq:analysis 2})
differ only by a constant factor. More specifically, from \eqref{eq:f b r} and \eqref{eq:beta tilde}, it follows that 
\begin{align*}
&    \frac{\partial^{2}f}{\partial\rho\partial\beta}\left(\rho_{0},\widetilde{\beta}(\rho_{0})\right)\cdot\frac{\partial\widetilde{\beta}}{\partial\rho}(\rho_{0})\\
 & =-2\left[\mbox{diag}\left(\widetilde{\beta}(\rho_{0})\right)G_{\rho_{0}}^{*}LG_{\rho_{0}}-\mbox{diag}\left(G_{\rho_{0}}^{*}L\left(I_{N}-\mathcal{P}_{\rho_{0}}\right)\widehat{z}_{\sigma_{2}}\right)\right]\\
 & \qquad\cdot\left(G_{\rho_{0}}^{*}G_{\rho_{0}}\right)^{-1}\cdot\left[G_{\rho_{0}}^{*}L^{*}G_{\rho_{0}}\cdot\mbox{diag}(\widetilde{\beta}(\rho_{0}))-\mbox{diag}\left(G_{\rho_{0}}^{*}L\left(I_{N}-\mathcal{P}_{\rho_{0}}\right)\widehat{z}_{\sigma_{2}}\right)\right],\\
 & =-\left(\frac{\partial\widetilde{\beta}}{\partial\rho}(\rho_{0})\right)^{*}\cdot\frac{\partial^{2}f}{\partial\beta^{2}}\left(\rho_{0},\widetilde{\beta}(\rho_{0})\right)\cdot\frac{\partial\widetilde{\beta}}{\partial\rho}(\rho_{0}),
\end{align*}
so that
\begin{align}
\frac{\partial^{2}F}{\partial\rho^{2}}(\rho_{0}) & =\frac{\partial^{2}f}{\partial\rho^{2}}\left(\rho_{0},\widetilde{\beta}(\rho_{0})\right)+\frac{\partial^{2}f}{\partial\rho\partial\beta}\left(\rho_{0},\widetilde{\beta}(\rho_{0})\right)\cdot\frac{\partial\widetilde{\beta}}{\partial\rho}(\rho_{0})\nonumber \\
 & =-2\cdot\mbox{diag}\left(\widetilde{\beta}(\rho_{0})\right)\cdot G_{\rho_{0}}^{*}L^{2}G_{\rho_{0}}\cdot\mbox{diag}\left(\widetilde{\beta}(\rho_{0})\right)\nonumber \\
 & \qquad-2\cdot\mbox{diag}\left(\widetilde{\beta}(\rho_{0})\right)\cdot\mbox{diag}\left(G_{\rho_{0}}^{*}L^{2}\left(I_{N}-\mathcal{P}_{\rho_{0}}\right)\widehat{z}_{\sigma_{2}}\right)\nonumber \\
 & \qquad-2\left[\mbox{diag}\left(\widetilde{\beta}(\rho_{0})\right)G_{\rho_{0}}^{*}LG_{\rho_{0}}-\mbox{diag}\left(G_{\rho_{0}}^{*}L\left(I_{N}-\mathcal{P}_{\rho_{0}}\right)\widehat{z}_{\sigma_{2}}\right)\right]\nonumber \\
 & \qquad\cdot\left(G_{\rho_{0}}^{*}G_{\rho_{0}}\right)^{-1}\cdot\left[G_{\rho_{0}}^{*}L^{*}G_{\rho_{0}}\cdot\mbox{diag}(\widetilde{\beta}(\rho_{0}))-\mbox{diag}\left(G_{\rho_{0}}^{*}L\left(I_{N}-\mathcal{P}_{\rho_{0}}\right)\widehat{z}_{\sigma_{2}}\right)\right.,\label{eq:Hess app}
\end{align}
which might be simplified slightly further.

\section{Ingredients for Computing $\frac{\partial F}{\partial\rho}(\cdot)$
and $\frac{\partial^{2}F}{\partial\rho^{2}}(\cdot)$\label{sec:analysis}}

Here, we establish (\ref{eq:analysis 1}) and (\ref{eq:analysis 2}).
Fix $\rho_{0}$ and suppose that $f(\cdot,\cdot)$ is analytic, i.e.,
has convergent power series everywhere . Moreover, assume that
\[
\widetilde{\beta}(\rho):=\arg\min_{\beta}f(\rho,\beta)
\]
is always well-defined, i.e., $\widetilde{\beta}(\rho)$ is the unique
minimizer of $f(\rho,\cdot)$ for every $\rho$. In particular, by implicit function theorem, $\widetilde{\beta}(\rho)$
is smooth (i.e., infinitely differentiable with respect to $\rho$).
 We wish to calculate the first and second derivatives
of $F(\cdot)$, the map that takes $\rho$ to $F(\rho)=\min_{\beta}f(\rho,\beta)=f(\rho,\widetilde{\beta}(\rho))$.
(The existence of these derivatives is established along the way.)

To that end, we note that the following expansion holds for small
enough $|\rho-\rho_{0}|$ and $|\widetilde{\beta}(\rho)-\widetilde{\beta}(\rho_{0})|$:
\begin{align*}
F(\rho) & =f\left(\rho,\widetilde{\beta}(\rho)\right)\\
 & =f\left(\rho_{0},\widetilde{\beta}(\rho_{0})\right)+\left(\rho-\rho_{0}\right)^{T}\cdot\frac{\partial f}{\partial\rho}\left(\rho_{0},\widetilde{\beta}(\rho_{0})\right)+\left(\widetilde{\beta}(\rho)-\widetilde{\beta}(\rho_{0})\right)^{T}\cdot\frac{\partial f}{\partial\beta}\left(\rho_{0},\widetilde{\beta}(\rho_{0})\right)\\
 & \qquad+\frac{1}{2}\left(\rho-\rho_{0}\right)^{T}\cdot\frac{\partial^{2}f}{\partial\rho^{2}}\left(\rho_{0},\widetilde{\beta}(\rho_{0})\right)\cdot\left(\rho-\rho_{0}\right)\\
 & \qquad+\left(\rho-\rho_{0}\right)^{T}\cdot\frac{\partial^{2}f}{\partial\rho\partial\beta}\left(\rho_{0},\widetilde{\beta}(\rho_{0})\right)\cdot\left(\widetilde{\beta}(\rho)-\widetilde{\beta}(\rho_{0})\right)\\
 & \qquad+\frac{1}{2}\left(\widetilde{\beta}(\rho)-\widetilde{\beta}(\rho_{0})\right)^{T}\cdot\frac{\partial^{2}f}{\partial^{2}\beta}\left(\rho_{0},\widetilde{\beta}(\rho_{0})\right)\cdot\left(\widetilde{\beta}(\rho)-\widetilde{\beta}(\rho_{0})\right)+o_{2}.
\end{align*}
Above, $o_{2}$ comprises of negligible terms. Note that $f(\rho_{0},\widetilde{\beta}(\rho_{0}))=F(\rho_{0})$
and that 
$$
\frac{\partial f}{\partial\beta}(\rho_{0},\widetilde{\beta}(\rho_{0}))=0, 
$$ 
because $\widetilde{\beta}(\rho_{0})$ minimizes $f(\rho_{0},\cdot)$.
On the other hand, because $\widetilde{\beta}(\rho)$ is a smooth
function of $\rho$, $\widetilde{\beta}(\rho)-\widetilde{\beta}(\rho_{0})=\frac{\partial\widetilde{\beta}}{\partial\rho}(\rho_{0})\cdot\left(\rho-\rho_{0}\right)+o_{1}$
for small enough $|\rho-\rho_{0}|$. Here, $o_{1}$ collects the negligible
terms.Therefore, the above expansion simplifies to%
{}
\begin{align*}
F(\rho) & =F(\rho_{0})+\left(\rho-\rho_{0}\right)^{T}\cdot\frac{\partial f}{\partial\rho}\left(\rho_{0},\widetilde{\beta}(\rho_{0})\right)+\frac{1}{2}\left(\rho-\rho_{0}\right)^{T}\cdot\frac{\partial^{2}f}{\partial\rho^{2}}\left(\rho_{0},\widetilde{\beta}(\rho_{0})\right)\cdot\left(\rho-\rho_{0}\right)\\
 & \qquad+\left(\rho-\rho_{0}\right)^{T}\cdot\frac{\partial^{2}f}{\partial\rho\partial\beta}\left(\rho_{0},\widetilde{\beta}(\rho_{0})\right)\cdot\frac{\partial\widetilde{\beta}}{\partial\rho}(\rho_{0})\cdot\left(\rho-\rho_{0}\right)\\
 & \qquad+\frac{1}{2}(\rho-\rho_{0})^{T}\cdot\left(\frac{\partial\widetilde{\beta}}{\partial\rho}(\rho_{0})\right)^{T}\cdot\frac{\partial^{2}f}{\partial\beta^{2}}\left(\rho_{0},\widetilde{\beta}(\rho_{0})\right)\cdot\frac{\partial\widetilde{\beta}}{\partial\rho}(\rho_{0})\cdot\left(\rho-\rho_{0}\right)+o_{2}.
\end{align*}
We conclude that
\[
\frac{\partial F}{\partial \rho}(\rho_{0})=\frac{\partial f}{\partial\rho}\left(\rho_{0},\widetilde{\beta}(\rho_{0})\right),
\]
\begin{align*}
\frac{\partial^{2}F}{\partial\rho^{2}}(\rho_{0}) & =\frac{\partial^{2}f}{\partial\rho^{2}}\left(\rho_{0},\widetilde{\beta}(\rho_{0})\right)+2\cdot \frac{\partial^{2}f}{\partial\rho\partial\beta}\left(\rho_{0},\widetilde{\beta}(\rho_{0})\right)\cdot\frac{\partial\widetilde{\beta}}{\partial\rho}(\rho_{0})\\
 & \qquad+\left(\frac{\partial\widetilde{\beta}}{\partial\rho}(\rho_{0})\right)^{T}\cdot\frac{\partial^{2}f}{\partial\beta^{2}}\left(\rho_{0},\widetilde{\beta}(\rho_{0})\right)\cdot\frac{\partial\widetilde{\beta}}{\partial\rho}(\rho_{0}).
\end{align*}
Note that, despite the nonsymmetric appearance of the second term
in the Hessian, $\frac{d^{2}F}{d\rho^{2}}(\rho_{0})\in\mathbb{R}^{K\times K}$
is indeed a symmetric matrix.

\end{document}